\documentclass[11pt,a4paper]{article}

\usepackage[lmargin=1.0in,rmargin=1.0in,bottom=1.0in,top=1.0in,twoside=False]{geometry}
\usepackage{microtype}
\usepackage{libertine}

\usepackage{enumerate}
\usepackage{enumitem}
\usepackage{hyperref}
\usepackage{amsmath,amssymb,amsthm}
\usepackage{tikz}
\usepackage{algorithm}
\usepackage{algpseudocode}
\usepackage[capitalise]{cleveref}
\usepackage{fullpage}
\usepackage{todonotes}
\usepackage{authblk}
\usepackage{xspace}
\usepackage[normalem]{ulem}

\usepackage[absolute]{textpos}
\usepackage[all=normal,bibliography=tight]{savetrees}
\usepackage{thmtools}
\usepackage{thm-restate}

\newtheorem{theorem}{Theorem}[section]
\newtheorem{lemma}[theorem]{Lemma}

\newtheorem{proposition}[theorem]{Proposition}

\newtheorem{claim}[theorem]{Claim}

\theoremstyle{definition}
\newtheorem{definition}[theorem]{Definition}

\theoremstyle{remark}

\newcommand\abs[1]{\lvert #1\rvert}
\newenvironment{proofofclaim}{\noindent \textsc{Proof of the Claim:}}{\hfill$\Diamond$\medskip}

\newcommand{\yes}{\textsc{Yes}}
\newcommand{\no}{\textsc{No}}
\newcommand{\Oh}{\mathcal{O}}

\newcommand{\A}{\mathcal{A}}
\newcommand{\B}{\mathcal{B}}
\newcommand{\F}{\mathcal{F}}
\newcommand{\N}{\mathbb{N}}

\renewcommand{\P}{\mathcal{P}}

\newcommand{\weight}{\mathbf{w}}
\newcommand{\maxarity}{b}
\newcommand{\idtwopositive}{\Delta^{2K_2}_{\maxarity}}
\newcommand{\isdpositive}{\Sigma^{2K_2}_{\maxarity}}

\newcommand{\flow}{\mathcal{P}}
\newcommand{\witnessflow}{\widehat{\flow}}
\newcommand{\corecut}[1]{\mathrm{core}(#1)}
\newcommand{\corecutG}[2]{\mathrm{core}_{#2}(#1)}

\newcommand{\MinSAT}[1]{\textup{\textsc{Min SAT}(\ensuremath{#1})}}
\newcommand{\WeightedMinSAT}[1]{\textup{\textsc{Weighted Min SAT}(\ensuremath{#1})}}

\newcommand{\MaxSAT}[1]{\textup{\textsc{Max SAT}($#1$)}}
\newcommand{\MinOnes}[1]{\textup{\textsc{Min Ones}($#1$)}}
\newcommand{\MaxOnes}[1]{\textup{\textsc{Max Ones}($#1$)}}
\newcommand{\ExactOnes}[1]{\textup{\textsc{Exact Ones}($#1$)}}

\newcommand{\SAT}[1]{\textup{\textsc{SAT}($#1$)}}
\newcommand{\CSP}[1]{\textup{\textsc{CSP}($#1$)}}
\newcommand{\MinCSP}[1]{\textup{\textsc{Min CSP}($#1$)}}

\newcommand{\VCSP}[1]{\textup{\textsc{VCSP}}($#1$)}

\newcommand{\PMstCl}{\textsc{Paired Minimum $s$,$t$-Cut($\ell$)}\xspace}
\newcommand{\MCCk}{\textsc{Multicolored Clique($k$)}\xspace}

\newcommand{\constraint}{\ensuremath{\boldsymbol{R}}\xspace}

\newcommand{\classFPT}{\ensuremath{\mathsf{FPT}}\xspace}
\newcommand{\classWone}{\ensuremath{\mathsf{W[1]}}\xspace}

\newcommand{\vars}[1]{\mathsf{Vars}_{#1}}
\newcommand{\ones}[1]{\mathsf{True}_{#1}}
\newcommand{\zeroes}[1]{\mathsf{False}_{#1}}
\newcommand{\rest}[1]{\mathsf{Rest}_{#1}}
\newcommand{\candef}[1]{\Phi^{\mathrm{can}}_{#1}}
\newcommand{\clauses}{\mathcal{C}}

\newcommand{\inst}{\mathcal{I}}

\newcommand{\cF}{\mathcal{F}}

\newcommand{\gdpcshort}{\textsc{GDPC}}
\newcommand{\gdpcfull}{\textsc{Generalized Digraph Pair Cut}}
\newcommand{\bundles}{\mathcal{B}}
\newcommand{\pairs}{\mathcal{C}}
\newcommand{\preim}{\mathcal{C}}

\newcommand{\vtx}{{\sf v}}

\newcommand{\karc}{\kappa}
\newcommand{\karcout}{\kappa_{out}}
\newcommand{\kclause}{\kappa_c}

\title{Flow-augmentation III: Complexity dichotomy for Boolean CSPs parameterized by the number of unsatisfied constraints%
\thanks{An extended abstract of this work has been presented
at SODA 2023~\cite{csp-soda}.\\
This research is a part of a project that has received funding from the European Research Council (ERC)
under the European Union's Horizon 2020 research and innovation programme
Grant Agreement 714704 (M. Pilipczuk). Eun Jung Kim is supported by the grant from French National Research Agency under JCJC program (ASSK: ANR-18-CE40-0025-01). A preliminary version of this work has been presented at SODA 2023.}
}

\date{}

\author[1]{Eun Jung Kim}
\author[2]{Stefan Kratsch}
\author[3]{Marcin Pilipczuk}
\author[4]{Magnus Wahlstr\"{o}m}
\affil[1]{Universit\'{e} Paris-Dauphine, PSL Research University, CNRS, UMR 7243, LAMSADE, 75016, Paris, France.}
\affil[2]{Humboldt-Universit\"at zu Berlin, Germany}
\affil[3]{University of Warsaw, Warsaw, Poland}
\affil[4]{Royal Holloway, University of London, TW20 0EX, UK}

\begin{document}

\begin{titlepage}
\def\thepage{}
\thispagestyle{empty}

\maketitle

\begin{textblock}{20}(0, 13.0)
\includegraphics[width=40px]{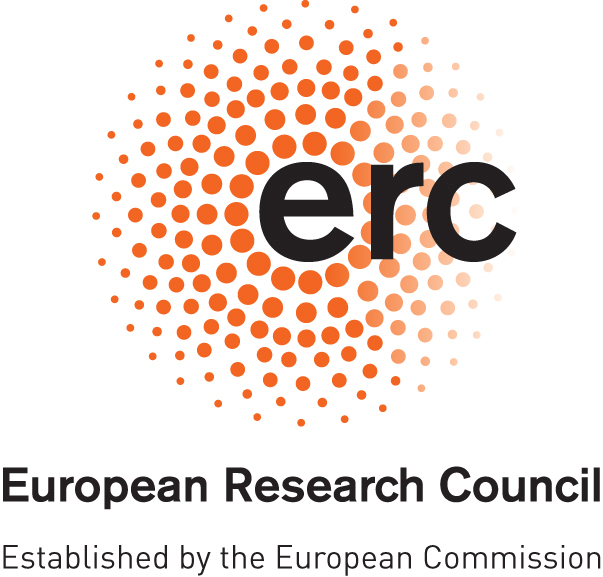}%
\end{textblock}
\begin{textblock}{20}(0, 13.9)
\includegraphics[width=40px]{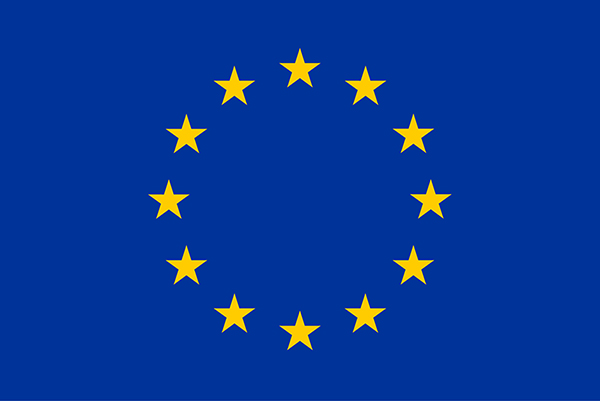}%
\end{textblock}

\begin{abstract}
  We study the parameterized problem of satisfying ``almost all''
  constraints of a given formula $\F$ over a fixed, finite Boolean constraint
  language $\Gamma$, with or without weights. More precisely, for each
  finite Boolean constraint language~$\Gamma$, we
  consider the following two problems. In \MinSAT{\Gamma}, the input
  is a formula~$\F$ over~$\Gamma$ and an
  integer~$k$, and the task is to find an assignment
  $\alpha \colon V(\F) \to \{0,1\}$ that satisfies all but at most $k$
  constraints of $\F$, or determine that no such assignment exists. 
  In \WeightedMinSAT{\Gamma}, the input additionally contains a
  \emph{weight function} $\weight \colon \F \to  \mathbb{Z}_+$
  and an integer $W$, and the task is to find an assignment $\alpha$ such that
  (1) $\alpha$ satisfies all but at most $k$ constraints of $\F$, and
  (2) the total weight of the violated constraints is at most $W$.
  We give a complete dichotomy for the fixed-parameter tractability of
  these problems: 
  We show that for every Boolean constraint language~$\Gamma$,
  either \WeightedMinSAT{\Gamma} is \classFPT; or \WeightedMinSAT{\Gamma} is
  \classWone-hard but \MinSAT{\Gamma} is \classFPT; or \MinSAT{\Gamma} is \classWone-hard.
  This generalizes recent work of Kim et al. (SODA 2021) which
  did not consider weighted problems, and only considered languages
  $\Gamma$ that cannot express implications $(u \to v)$ (as is used
  to, e.g., model digraph cut problems). Our result generalizes and subsumes
  multiple previous results, including the FPT algorithms for
  \textsc{Weighted Almost 2-SAT}, weighted and unweighted \textsc{$\ell$-Chain SAT},
  and \textsc{Coupled Min-Cut}, as well as weighted and directed versions of the latter. 
  The main tool used in our algorithms is the recently developed
  method of \emph{directed flow-augmentation} (Kim et al., STOC 2022).
\end{abstract}

\newpage
\tableofcontents
\end{titlepage}

\section{Introduction}

Constraint satisfaction problems (CSPs) are a popular, heavily studied
framework that allows a wide range of problems to be expressed and
studied in a uniform manner. Informally speaking, a CSP is defined by
fixing a domain $D$ and a \emph{constraint language} $\Gamma$ over $D$
controlling the types of constraints that are allowed in the problem.
The problem \CSP{\Gamma} then takes as input a conjunction of
such constraints, and the question is whether there is an
assignment that satisfies all constraints in the input. Some
examples of problems \CSP{\Gamma} for particular constraint languages
$\Gamma$ include 2-SAT, 
\textsc{$k$-Coloring}, 
linear equations over a finite field, and many more.
We use \SAT{\Gamma} for the special case of
constraints over the Boolean domain $D=\{0,1\}$.

More precisely, a constraint language over a domain $D$ is a set of
finite-arity relations $R \subseteq D^r$ (where $r$ is the arity of
$R$). A \emph{constraint} over a constraint language $\Gamma$ 
is formally a pair $(X,R)$, where $R \in \Gamma$ is a relation from the language,
say of arity $r$, and $X=(x_1,\ldots,x_r)$ is an $r$-tuple of variables 
called the \emph{scope} of the constraint. We typically write our
constraints as $R(X)$ instead of $(X,R)$, or $R(x_1,\ldots,x_r)$ when
the individual participating variables $x_i$ need to be highlighted.
Let $\alpha \colon X \to D$ be an assignment. Then $\alpha$
\emph{satisfies} the constraint $R(X)$ if $(\alpha(x_1),\ldots,\alpha(x_r)) \in R$,
and we say that $\alpha$ \emph{violates} the constraint otherwise.
A \emph{formula over $\Gamma$} is then a conjunction of constraints
over $\Gamma$, and the problem \CSP{\Gamma} is to decide, given a
formula $\F$ over $\Gamma$, whether $\F$ is satisfiable, i.e., if
there is an assignment that satisfies all constraints of $\F$.
To revisit the examples above, if $D=\{0,1\}$ is the Boolean domain
and $\Gamma$ contains only relations of arity at most 2 over $D$,
then \SAT{\Gamma} is polynomial-time decidable by reduction to \textsc{2-SAT}. 
Similarly, if each relation $R \in \Gamma$ can be defined via linear equations over GF(2),
e.g., $R(x_1,\ldots,x_r) \equiv (x_1 +\ldots + x_r =1 \pmod 2)$,
then \SAT{\Gamma} is polynomial-time decidable via Gaussian elimination.
Finally, \textsc{$k$-Coloring} corresponds to \CSP{\Gamma} over a
domain $D=\{1,\ldots,k\}$ of cardinality $k$, and with the constraint
language $\Gamma$ containing only the relation $R \subseteq D^2$
defined as $R(u,v) \equiv (u \neq v)$. 
Note that these reductions can also easily be turned into equivalences,
i.e., there is a specific constraint language $\Gamma$
such that \SAT{\Gamma} respectively \CSP{\Gamma} is effectively
equivalent to \textsc{2-SAT}, \textsc{$k$-Coloring}, 
linear equations over a fixed finite field, and so on.

By capturing such a range of problems in one framework, the CSP framework
also allows us to study these problems in a uniform manner.
In particular, it allows for the complete
characterisation of the complexity of every problem in the framework
--- so-called \emph{dichotomy theorems}. The most classical is by
Schaefer~\cite{Schaefer78}, who showed that for every finite Boolean language
$\Gamma$, either $\Gamma$ is contained in one of six maximal tractable
classes and \SAT{\Gamma} is in P, or else \SAT{\Gamma} is NP-complete.
Since then, many other dichotomy theorems have been settled (many of them mentioned later in this introduction).
Perhaps chief among them is the general
\emph{CSP dichotomy theorem}: For every finite constraint language
$\Gamma$ over a finite domain, the problem \CSP{\Gamma} is either in P
or NP-complete. This result was conjectured by Feder and Vardi in the
90's~\cite{FederV93}, and only fully settled a few years ago,
independently by Bulatov~\cite{Bulatov17CSP} and Zhuk~\cite{Zhuk20CSP}.

The existence of dichotomy theorems allows us to formally study the
question of what makes a problem in a problem category hard --- or
rather, since hardness appears to be the default state, what makes a
problem in a problem category tractable? From a technical perspective,
the answer is often phrased algebraically, in terms of algebraic
closure properties of the constraint language which describe abstract
symmetries of the solution space (see, for example, the collection
edited by Krokhin and Zivn\'y~\cite{dagstuhl2017CSPs}). But the answer
can also be seen as answering a related question: What algorithmic
techniques are required to handle all tractable members of a problem
class? In other words, what are the maximal ``islands of tractability''
in a problem domain, and what algorithmic techniques do they require?

Thus in particular, a dichotomy theorem requires you to both discover
all the necessary tools in your algorithmic toolbox, and to hone each
of these tools to the maximum generality required by the domain.

%

As a natural variation on the CSP problem, when a formula $\F$ is
not satisfiable, we might look for an assignment under which as few
constraints of $\F$ as possible are violated. This defines an
\emph{optimization problem} for every language $\Gamma$. Formally, for
a constraint language $\Gamma$, the problem \MinCSP{\Gamma} takes as
input a formula $\F$ over $\Gamma$ and an integer $k$, and asks if
there is an assignment under which at most $k$ constraints of $\F$ are
violated. Again, we use \MinSAT{\Gamma} to denote the special case
where $\Gamma$ is over the domain $\{0,1\}$. Equivalently, we may
consider the \emph{constraint deletion} version of \MinCSP{\Gamma} and
\MinSAT{\Gamma}: Given a formula $\F$ over $\Gamma$ and integer $k$,
is there a set $Z \subseteq \F$ of at most $k$ constraints such that
$\F-Z$ is satisfiable? This version tends to fit better with our
algorithms. We refer to such a set $Z$ as \emph{deletion set}.

Let us consider an example. Let $\Gamma=\{(x=1),(x=0),(x \to y)\}$.
Then \MinSAT{\Gamma} is effectively equivalent to finding a minimum
$st$-cut in a digraph. Indeed, let $\F$ be a formula over $\Gamma$
and define a digraph $G$ on vertex set $V(G)=V(\F) \cup \{s,t\}$,
with an arc $(s,v)$ for every constraint $(v=1)$ in $\F$, an
arc $(v,t)$ for every constraint $(v=0)$ in $\F$, and an arc $(u,v)$
for every constraint $(u \to v)$ in $\F$. Let $S \subseteq V(G)$
be a vertex set with $s \in S$ and $t \notin S$, and define
an assignment $\alpha_S \colon V(\F) \to \{0,1\}$ by
$\alpha(v)=1$ if and only if $v \in S$. Then the constraints of $\F$
violated by $\alpha_S$ are precisely the edges $\delta_G(S)$ leaving
$S$ in $G$, i.e., an $st$-cut in $G$. In particular, \MinSAT{\Gamma}
is solvable in polynomial time. Naturally, by the same mapping we can
also generate an instance $(\F,k)$ of \MinSAT{\Gamma} from a graph $G$
with marked vertices $s, t \in V(G)$, such that $(\F,k)$ is a
\yes-instance if and only if $G$ has an $st$-cut of at most $k$ edges,
justifying the claim that the problems are equivalent.

Unfortunately, for most languages $\Gamma$ the resulting problem
\MinSAT{\Gamma} is NP-hard. Indeed, Khanna et al.~\cite{KhannaSTW00}
showed that the above example is essentially the only non-trivial
tractable case; for every constraint language $\Gamma$, either
formulas over $\Gamma$ are always satisfiable for trivial reasons, or 
\MinSAT{\Gamma} reduces to $st$-Min Cut, or \MinSAT{\Gamma} is
APX-hard. (Furthermore, many interesting examples of \MinSAT{\Gamma}
do not appear to even allow constant-factor approximations; see
discussion below.)

However, \MinSAT{\Gamma} is a natural target for studies in
parameterized complexity. Indeed, taking $k$ as a natural parameter,
many cases of \MinSAT{\Gamma} have been shown to be \classFPT when
parameterized by $k$, including the classical problems of
\textsc{Edge Bipartization}, corresponding to a language $\Gamma=\{(x \neq y)\}$, 
and \textsc{Almost 2-SAT}, corresponding to a language $\Gamma$
containing all 2-clauses. The former is \classFPT by Reed et al.~\cite{ReedSV04},
the latter by Razgon and O'Sullivan~\cite{RazgonO09},
both classical results in the parameterized complexity literature. 
It is thus natural to ask for a general characterisation:
For which Boolean languages $\Gamma$ is \MinSAT{\Gamma} \classFPT
parameterized by $k$?

Indeed, following early FPT work on related CSP optimization
problems~\cite{Marx05CSP,BulatovM14SICOMP,KratschW10ICALP,KratschMW16TOCT}, the \MinSAT{\Gamma}
question was a natural next target. 
Unfortunately, for a long time this question appeared out of reach,
due to some very challenging open problems, specifically
\textsc{Coupled Min-Cut} and \textsc{$\ell$-Chain SAT}.
\textsc{Coupled Min-Cut} is a graph separation problem,
never publically posed as an open problem, but long known to be a
difficult question, unknown to be either \classFPT or \classWone-hard.
\textsc{$\ell$-Chain SAT} is a digraph cut problem posed by
Chitnis et al.~in 2013~\cite{ChitnisEM13,ChitnisEM17}, conjectured to be \classFPT,
and repeatedly asked as an open question since then. 
(Both problems are fully defined below.)
Since completing a \MinSAT{\Gamma} \classFPT/\classWone-hardness dichotomy is clearly
harder than settling the \classFPT or \classWone-hardness-status of the
individual problems \textsc{Coupled Min-Cut} and \textsc{$\ell$-Chain SAT},
these problems worked as roadblocks against a dichotomy result.
However, recently, using the new \emph{flow augmentation} technique,
both of the above problems have been proven to be \classFPT~\cite{ufl-soda,dfl-stoc},
along with multiple other long-open weighted FPT problems.
With the new technique in place, we find it is now time to attack
the \classFPT/\classWone dichotomy question for \MinSAT{\Gamma} directly.

\subsection{Our results}

As mentioned, we consider two variants of \MinSAT{\Gamma}, with and without
constraint weights. Let $\Gamma$ be a finite Boolean constraint language.
\MinSAT{\Gamma} is the problem defined above: Given input $(\F,k)$,
where $\F$ is a conjunction of constraints using relations of $\Gamma$ (we treat $\F$ as a set),
decide if there is a deletion set of
cardinality at most $k$, i.e., a set $Z \subseteq \F$
of at most $k$ constraints such that $\F-Z$ is satisfiable.
In the weighted version \WeightedMinSAT{\Gamma},
the input is $(\F, \weight, k, W)$, where the formula $\F$
comes equipped with a weight function $\weight \colon \F \to \mathbb{Z}_+$ 
assigning weights to the constraints of $\F$, and the goal is to
find a deletion set of cardinality at most $k$ and weight at most $W$.
Note that this is a fairly general notion of a weighted
problem; e.g., it could be that there is an assignment violating
strictly fewer than $k$ constraints, but that every such assignment 
violates constraints to a weight of more than $W$.

We give a full characterization of \MinSAT{\Gamma} and
\WeightedMinSAT{\Gamma} as being either \classFPT or \classWone-hard
when parameterized by the number of violated constraints $k$.
This extends the partial dichotomy for languages which cannot express directed reachability~\cite{ufl-soda} (which was obtained using the undirected version of the flow-augmentation technique),
and complements the approximate \classFPT dichotomy
of Bonnet et al.~\cite{BonnetEM16ESA,BonnetELM-arXiv}.

\begin{restatable}{theorem}{thmdich}\label{ithm:dichotomy}
  Let $\Gamma$ be a finite Boolean constraint language. Then one of the
  following applies for the parameterization by the number of unsatisfied constraints.
  \begin{enumerate}
  \item \WeightedMinSAT{\Gamma} is \classFPT.
  \item \MinSAT{\Gamma} is \classFPT but \WeightedMinSAT{\Gamma} is \classWone-hard.
  \item \MinSAT{\Gamma} is \classWone-hard.
  \end{enumerate}
\end{restatable}
We remark that the distinction between the cases of Theorem~\ref{ithm:dichotomy} based on the elements of $\Gamma$
is explicit and relatively simple. We describe it in the next few paragraphs.

Our characterization is combinatorial, and is given in terms of
graphs that encode the structure of each constraint. To state it, we
first need some terminology.
We say that a Boolean relation $R$ is \emph{bijunctive} if it is expressible as
a conjunction of 1- and 2-clauses
A relation $R$ is \emph{IHS-B-}
if it is expressible as a conjunction of negative clauses $(\neg x_1 \lor \ldots \lor \neg x_r)$, positive 1-clauses $(x)$, and implications $(x \to y)$.
Similarly, a relation $R$ is \emph{IHS-B+}
if it is expressible as a conjunction of positive clauses $(x_1 \lor \ldots \lor x_r)$, negative
1-clauses $(\neg x)$, and implications $(x \to y)$.
Here, IHS-B is an abbreviation
for \emph{implicative hitting set, bounded}
(i.e., implicative hitting set where there is a bound $r \in \N$ on the maximum arity of a clause).\footnote{The
  fact that clauses are \emph{bounded} is implicit in our assumption that the language $\Gamma$ is finite,
  but these are traditional terms; see, e.g., Khanna et al.~\cite{KhannaSTW00}.}
A constraint language $\Gamma$ is
bijunctive, IHS-B+, respectively IHS-B- if every relation
$R \in \Gamma$ is bijunctive, IHS-B+, respectively IHS-B-.
Finally, $\Gamma$ is IHS-B if it is either IHS-B+ or IHS-B-.
Note that this is distinct from every relation $R \in \Gamma$
being either IHS-B+ or IHS-B-, since the latter would allow
a mixture of, e.g., positive and negative 3-clauses, which defines an
NP-hard problem \SAT{\Gamma}~\cite{Schaefer78}.
From previous work, it is known that if $\Gamma$ is not IHS-B or
bijunctive, then either every formula over $\Gamma$ is satisfiable
and \MinSAT{\Gamma} is trivial, or \MinSAT{\Gamma} does not even allow
FPT-time constant-factor approximation algorithms unless \classFPT=\classWone~\cite{BonnetELM-arXiv}.
Hence we may assume that $\Gamma$ belongs to one of these cases.

We will characterise the structure of relations in two ways.
Let $R \subseteq \{0,1\}^r$ be a Boolean relation. 
First, slightly abusing terminology in reusing a term from the literature, we define the
\emph{Gaifman graph of $R$} as an undirected graph $G_R$ on
vertex set $[r]=\{1,\ldots,r\}$, where there is an edge $\{i,j\} \in E(G_R)$
if and only if the projection of $R$ onto arguments $i$ and $j$ is
non-trivial, i.e., if and only if not every pair of values of $t[i]$ and $t[j]$ is possible
(formally, if and only if there exist values $b_i, b_j \in \{0,1\}$ 
such that for every $t \in R$ it is not the case that both $t[i]=b_i$ and $t[j]=b_j$).
Second, we define the \emph{arrow graph} $H_R$ of $R$ as a directed
graph on vertex set $[r]$ where $(i,j) \in E(H_R)$ if $R(x_1,\ldots,x_r)$
implies the constraint $(x_i \to x_j)$ without also implying
$(x_i=0)$ or $(x_j=1)$.
Finally, we say that $G_R$ is \emph{$2K_2$-free} if there is no
induced subgraph of $G_R$ isomorphic to $2K_2$, i.e., consisting
of two vertex-disjoint edges. Similarly,
the arrow graph $H_R$ is $2K_2$-free if the underlying undirected
graph of $H_R$ is $2K_2$-free.

For an illustration, consider the relation $R(x,y,z) \equiv (x=1) \land (y=z)$.
Let us consider the full set of 2-clauses implied by $R(x,y,z)$, i.e.,
\[
  R(x,y,z) \models (x \lor y) \land (x \lor \neg y) \land (x \lor z)
  \land (x \lor \neg z) \land (\neg y \lor z) \land (y \lor \neg z).
\]
where (naturally) clauses such as $(\neg y \lor z)$ could also be
written $(y \to z)$.
Observe that every pair of variables is involved in some 2-clause
and this 2-clause forbids some pair of values for the variables
(e.g., the clause $(x \lor \neg y)$ forbids $(x,y) = (0,1)$). 
Consequently, $G_R$ is a clique.
(Indeed, for the readers
familiar with the term Gaifman graph from the literature, if $R$ is
bijunctive then the Gaifman graph $G_R$ is precisely the Gaifman graph
of the 2-CNF formula consisting of all 2-clauses implied by $R(X)$.)
The arrow graph $H_R$ contains the arcs $(2,3)$ and $(3,2)$, due to
the last two clauses. On the other hand, despite the 2-clauses
$(y \to x)$ and $(z \to x)$ being valid in $R(x,y,z)$, $H_R$ does not
contain arcs $(2,1)$  or $(3,1)$ since they are only implied by the
assignment $(x=1)$.

We can now present the \classFPT{} results.

\begin{theorem} \label{ithm:gaifman}
  Let $\Gamma$ be a finite, bijunctive Boolean constraint language.
  If for every relation $R \in \Gamma$ the Gaifman graph $G_R$ is
  $2K_2$-free, then \WeightedMinSAT{\Gamma} is \classFPT.
\end{theorem}

\begin{theorem} \label{ithm:arrow}
  Let $\Gamma$ be a finite, IHS-B Boolean constraint language.
  If for every $R \in \Gamma$ the arrow graph $H_R$ is $2K_2$-free,
  then \MinSAT{\Gamma} is \classFPT. 
\end{theorem}

We note that Theorem~\ref{ithm:arrow} encompasses two language classes,
corresponding to IHS-B+ or IHS-B-. By symmetry of the problem,
the resulting \MinSAT{\Gamma} problems are essentially equivalent 
(e.g., by exchanging $x$ and $\neg x$ in all relation definitions);
hence it suffices to provide an FPT algorithm for one of the classes.
We focus on the IHS-B- case.
We also note that for any relation $R$ that is not bijunctive or IHS-B,
such as a ternary linear equation over GF(2) or a Horn clause
$R(z,y,z) \equiv (x \land y \to z)$, any problem \MinSAT{\Gamma}
with $R \in \Gamma$ is either trivially satisfiable or \classWone-hard~\cite{BonnetELM-arXiv}.

We remark that the algorithms of Theorems~\ref{ithm:gaifman}
and~\ref{ithm:arrow} are in fact \classFPT parameterized by the number of unsatisfied
constraints \emph{and} the maximum arity of a constraint in $\Gamma$.
That is, while the CSP framework treats the language $\Gamma$ as fixed, 
in the running time bounds of our algorithms, 
the degree of the polynomial in the factor depending of the input size
is a universal constant independent of $\Gamma$. 

The final dichotomy in Theorem~\ref{ithm:dichotomy} now follows
from showing that, except for a few simple cases, for any language
$\Gamma$ not covered by Theorem~\ref{ithm:gaifman} the problem
\WeightedMinSAT{\Gamma} is \classWone-hard, and if furthermore
Theorem~\ref{ithm:arrow} does not apply then \MinSAT{\Gamma} is \classWone-hard
(see Bonnet et al.~\cite{BonnetEM16ESA}).

Let us provide a few illustrative examples.
\begin{itemize}
\item First consider the problem \MinSAT{\Gamma} for the language
  $\Gamma=\{(x=1),(x=0),R_4\}$, where $R_4$ is the relation defined by
  $R_4(a,b,c,d) \equiv (a=b) \land (c=d)$.
  Then the Gaifman graph $G_{R_4}$ and the arrow graph $H_{R_4}$ both
  contain $2K_2$'s, hence \MinSAT{\Gamma} is \classWone-hard. In fact, 
  this problem, together with the directed version $(a \to b) \land (c \to d)$,
  are the fundamental \classWone-hard case of the dichotomy.
  
  On the other hand, consider $\Gamma'=\{(x=1),(x=0), (x=y)\}$.
  Similarly as \MinSAT{\{(x=1),(x=0),(x \to y)\}} is equivalent to
  the problem of finding a minimum cut in a directed graph,
  \MinSAT{\Gamma'} is equivalent to the problem of finding a minimum cut 
  in an undirected graph, and hence is in P. 
  Furthermore,
  \SAT{\Gamma} and \SAT{\Gamma'} are equivalent problems,
  since any constraint $R_4(a,b,c,d)$ can simply be split into $(a=b)$
  and $(c=d)$. For the same reason, \MinSAT{\Gamma} has a
  2-approximation, since breaking up a constraint over $R_4$ into
  separate constraints $(a=b)$ and $(c=d)$ at most
  doubles the number of violated constraints in any assignment.
  This illustrates the difference in the care that needs to be taken
  in an \classFPT/\classWone-dichotomy, compared to  approximability results.
  
\item The problem \textsc{Edge Bipartization} corresponds to
  \MinSAT{(x \neq y)} and \textsc{Almost 2-SAT} to
  \MinSAT{(x \lor y), (x \to y), (\neg x \lor \neg y)}.
  Since each relation $R$ here is just binary, the graph $G_R$
  has two vertices and hence is vacuously $2K_2$-free.
  Hence Theorem~\ref{ithm:gaifman}  generalizes the FPT
  algorithm for \textsc{Almost 2-SAT}~\cite{RazgonO09}. 
  
\item Let $(x_1 \to \ldots \to x_\ell)$ be shorthand for the constraint
  $R(x_1,\ldots,x_\ell) \equiv (x_1 \to x_2) \land \ldots \land (x_{\ell-1} \to x_\ell)$. 
  Then \MinSAT{(x=1),(x=0),(x_1 \to \ldots \to x_\ell)}
  is precisely the problem \textsc{$\ell$-Chain SAT}~\cite{ChitnisEM17}. 
  For our dichotomy, note that this constraint can also be written
  $
    R(x_1,\ldots,x_\ell) \equiv  \bigwedge_{1 \leq i < j \leq \ell} (x_i \to x_j).
  $
  Hence for this relation $R$, both the graphs $G_R$ and $H_R$ are
  cliques, and \textsc{$\ell$-Chain SAT} is contained in both of our
  tractable classes. This generalizes the very recent FPT algorithm for
  \textsc{$\ell$-Chain SAT}~\cite{dfl-stoc}.

\item Now consider a relation $R_{cmc}(a,b,c,d) \equiv (a=b) \land
  (c=d) \land (\neg a \lor \neg c)$. Then \MinSAT{(x=1),(x=0),R_{cmc}}
  is known as \textsc{Coupled Min-Cut}.
  Note that the Gaifman graph of $R_{cmc}$ is isomorphic to $K_4$,
  hence Theorem~\ref{ithm:gaifman} generalizes the result that
  \textsc{Coupled Min-Cut} is \classFPT~\cite{ufl-soda}.
  The same holds for natural directed variants such as
  $R'(a,b,c,d) \equiv (a \to b) \land (c \to d) \land (\neg a \lor \neg c)$. 
  Note that the Gaifman graph $G_{R'}$ is a $P_4$, i.e., $2K_2$-free.
  On the other hand, the arrow graphs of both these relations contain
  $2K_2$'s, hence, e.g.,
  \MinSAT{(x=1), (x=0), R_{cmc}, (\neg x \lor \neg y \lor \neg z)} is \classWone-hard
  since the language is no longer bijunctive.
  (Observe that adding more relations to a language
   makes it more expressive, and consequently makes
   the corresponding \textsc{Min SAT} problem harder.)
  
\item For an example in the other direction, consider a relation such as
  $
    R(a,b,c,d) \equiv (\neg a \lor \neg b) \land (c \to d).
  $
  Then the Gaifman graph $G_R$ is a $2K_2$, but the arrow graph $H_R$
  contains just one edge, showing that \MinSAT{(x=1),(x=0),R} is \classFPT
  but \WeightedMinSAT{(x=1),(x=0),R} is \classWone-hard. Similarly, adding a
  constraint such as $(x \neq y)$ to the language
  yields a \classWone-hard problem \MinSAT{\Gamma}, since $\Gamma$ is no
  longer IHS-B. Intuitively, this hardness comes about since having access to variable negation allows us
  to transform $R$ to the ``double implication'' constraint
  $R'(a,b,c,d) \equiv R(a, \neg b, c, d)$ from the \classWone-hard case
  mentioned in the first bullet.
  Indeed, a lot of the work of the hardness results
  in this paper is to leverage expressive power of \MinSAT{\Gamma}
  and \WeightedMinSAT{\Gamma} to ``simulate'' negations, in specific
  and restricted ways, when $(x \neq y)$ is not available in the
  language. 
\end{itemize}

\subsection{Previous dichotomies and related work}

This paper is the third part in a series that introduced the algorithmic technique of \emph{flow-augmentation}
and explores its applicability. 
The first part~\cite{dfl-arxiv,dfl-stoc}\footnote{Note that the order
  of publication, and to some extent the distribution of content, differs between the conference papers and the journal
versions. This paragraph describes the order between the journal versions.} introduced the technique and provided a number of applications,
    such as FPT algorithms for $\ell$-\textsc{Chain SAT} or \textsc{Weighted Directed Feedback Vertex Set}.
The second part~\cite{ufl-arxiv} shows improved guarantees in undirected graphs, compared
to the directed setting of~\cite{dfl-arxiv,dfl-stoc}.

Many variations of SAT and CSP with respect to decision and
optimization problems have been considered, and many of them are
relevant to the current work. 

Khanna et al.~\cite{KhannaSTW00} considered four optimization variants
of \SAT{\Gamma} on the Boolean domain, analyzed with respect to
approximation properties. They considered \MinOnes{\Gamma},
where the goal is to find a satisfying assignment with as few
variables set to 1 as possible; \MaxOnes{\Gamma},
where the goal is to find a satisfying assignment with as many
variables set to 1 as possible; \MinSAT{\Gamma}, where the goal is to
find an assignment with as few violated constraints as possible;
and \MaxSAT{\Gamma}, where the goal is to find an assignment
with as many satisfied constraints as possible. They characterized the
P-vs-NP boundary and the approximability properties of all problems in
all four variants.

Note that although, e.g., \MinSAT{\Gamma} and \MaxSAT{\Gamma} 
are equivalent with respect to the optimal assignments, from a
perspective of approximation they are very different. Indeed,
for any finite language $\Gamma$, you can (on expectation) produce a
constant-factor approximation algorithm for \MaxSAT{\Gamma}
simply by taking an assignment chosen uniformly at random.
For the same reason, \MaxSAT{\Gamma} parameterized by the number of
satisfied constraints $k$ is trivially \classFPT, and in fact has a linear
kernel for every finite language~$\Gamma$~\cite{KratschMW16TOCT}.
On the other hand, \MinSAT{\Gamma} is a far more challenging problem
from an approximation and fixed-parameter tractability perspective.
In fact, combining the characterisation of \MinSAT{\Gamma}
approximability classes of Khanna et al.~\cite{KhannaSTW00} with results
assuming the famous \emph{unique games conjecture} (UGC;
or even the weaker \emph{Boolean unique games conjecture}~\cite{EldanM22ITCS}), 
we find that if the UGC is true, then the only cases of \MinSAT{\Gamma} that 
admit a constant-factor approximation are when $\Gamma$ is IHS-B.

The first parameterized CSP dichotomy we are aware of is due to
Marx~\cite{Marx05CSP}, who considered the problem \ExactOnes{\Gamma}: 
Given a formula $\F$ over $\Gamma$, is there a satisfying assignment
that sets \emph{precisely} $k$ variables to 1? Marx gives a full
dichotomy for \ExactOnes{\Gamma} as being \classFPT or \classWone-hard
parameterized by $k$, later extended to general non-Boolean
languages with Bulatov~\cite{BulatovM14SICOMP}.
Marx also notes that \MinOnes{\Gamma}
is \classFPT by a simple branching procedure for every finite language
$\Gamma$~\cite{Marx05CSP}. However, the existence of so-called
\emph{polynomial kernels} for \MinOnes{\Gamma} problems is a
non-trivial question; a characterization for this was given by Kratsch
and Wahlstr\"om~\cite{KratschW10ICALP}, and follow-up work mopped
up the questions of FPT algorithms and polynomial kernels
parameterized by $k$ for all three variants \textsc{Min/Max/Exact Ones}$(\Gamma)$~\cite{KratschMW16TOCT}

Polynomial kernels for \MinSAT{\Gamma} have been considered; most
notably, there are polynomial kernels for \textsc{Edge Bipartization}
and \textsc{Almost 2-SAT}~\cite{KratschW20}. We are not aware of any
significant obstacles towards a kernelizability dichotomy of
\MinSAT{\Gamma}; however, a comparable result for more general
\MinCSP{\Gamma} appears difficult already for a domain of size 3
(due to \MinCSP{\Gamma} over domain size 3 capturing the problem
\textsc{Skew Multicut} with two terminal pairs~\cite{ChenLLOR08},
as well as the less-examined but still open \textsc{Unique Label Cover}
for an alphabet of size 3~\cite{ChitnisCHPP16}).

Another direct predecessor result of the current dichotomy is
the characterization of fixed-parameter constant-factor approximation
algorithms for \MinSAT{\Gamma} of Bonnet et al.~\cite{BonnetEM16ESA}.
Finally, we recall
that the problem \textsc{$\ell$-Chain SAT} was first published
by Chitnis et al.~\cite{ChitnisEM17}, who related its status
to a conjectured complexity dichotomy for the \textsc{Vertex Deletion List $H$-Coloring} problem.
This conjecture was subsequently confirmed with the FPT algorithm for
\textsc{$\ell$-Chain SAT} in the first paper in this series~\cite{dfl-arxiv,dfl-stoc}.

A much more ambitious optimization variant of CSPs are \emph{Valued CSP}, VCSP.
In this setting, instead of a constraint language one fixes a finite
set~$S$ of \emph{cost functions}, and considers the problem \VCSP{S},
of minimizing the value of a sum of cost functions from $S$. 
The cost functions can be either finite-valued or \emph{general},
taking values from $\mathbb{Q} \cup \{\infty\}$ to simulate
\emph{crisp}, unbreakable constraints. 
Both \MinOnes{\Gamma} and \MinCSP{\Gamma} (and, indeed, \textsc{Vertex
  Deletion List $H$-Coloring}) are special cases of VCSPs,
as are many other problems. The classical (P-vs-NP) complexity of VCSPs
admits a remarkably clean dichotomy: There is a canonical
LP-relaxation, the \emph{basic LP}, such that for any finite-valued $S$,
the problem \VCSP{S} is in P if and only if the basic LP is integral~\cite{ThapperZ16JACM}.
A similar, complete characterization for general-valued \VCSP{S}
problems is also known~\cite{KolmogorovKR17}.





\subsection{Technical overview}

The technical work of the paper is divided into three parts.
Theorem~\ref{ithm:gaifman}, i.e., the FPT algorithm for bijunctive
languages $\Gamma$ where every relation $R \in \Gamma$ has a
$2K_2$-free Gaifman graph, is proven in Section~\ref{sec:gaifman};
Theorem~\ref{ithm:arrow}, i.e., the FPT
algorithm for IHS-B languages $\Gamma$ where every relation
$R \in \Gamma$ has a $2K_2$-free arrow graph is proven in Section~\ref{sec:arrow};
and the completion of the
dichotomy, where we prove that all other cases are trivial or hard, is presented
in Section~\ref{sec:dichotomy}.
We begin the overview with the algorithmic results.




\paragraph{Graph problem.}
In both our algorithmic results, we cast the problem at hand as a graph separation problem that
   we call \textsc{Generalized Bundled Cut}.
An instance consists of 
\begin{itemize}[itemsep=0px]
\item a directed multigraph $G$ with distinguished vertices $s,t \in V(G)$;
\item a multiset $\pairs$ of subsets of $V(G)$, called \emph{clauses};
\item a family $\bundles$ of pairwise disjoint subsets of $E(G) \cup \pairs$,
  called \emph{bundles}, such that one bundle does not contain two copies of the same
  arc or clause;
\item a parameter $k$;
\item in the weighted variant, additionally a weight function $\weight \colon \bundles \to \mathbb{Z}_+$
and a weight budget $W \in \mathbb{Z}_+$.
\end{itemize}
We seek a set $Z \subseteq E(G)$ that is an $st$-cut (i.e., cuts all paths from $s$ to $t$).
An edge $e$ is \emph{violated} by $Z$ if $e \in Z$ and a clause $C \in \pairs$
is \emph{violated} by $Z$ if all elements of $C$ are reachable from $s$ in $G-Z$
(i.e., a clause is a request to separate at least one of the elements from $s$).
A bundle is \emph{violated} if at least one of its elements is violated.
An edge, a clause, or a bundle is \emph{satisfied} if it is not violated. 
An edge or a clause is \emph{soft} if it is in a bundle and \emph{crisp} otherwise.
We seek an $st$-cut $Z$ that satisfies all crisp edges and clauses,
and violates at most $k$ bundles
(and whose total weight is at most $W$ in the weighted variant).
An instance is \emph{$\maxarity$-bounded} if every clause is of size at most $\maxarity$ and, for every $B \in \bundles$,
   the set of vertices involved in the elements of $B$ is of size at most $\maxarity$.

\textsc{Generalized Bundled Cut}, in full generality,
 can be easily seen to be \classWone-hard when parameterized
by $k$, even with $\Oh(1)$-bounded instances; see Marx and Razgon~\cite{MarxR09}
and the problem \PMstCl{} defined in Section~\ref{ss:hard-cases}
and proved to be hard in Lemma~\ref{lemma:hardness:tightpairedstcut}.
In both tractable cases, the obtained instances 
are $\maxarity$-bounded for some $\maxarity$ depending on the language
and have some additional properties, related to $2K_2$-freeness of $G_R$ or $H_R$,
that allow for fixed-parameter algorithms when parameterized by $k + \maxarity$.

More precisely, in the bijunctive case, we have the following two properties.
First, the clauses are of arity $2$; let \gdpcfull{} (for short \gdpcshort)
be the \textsc{Generalized Bundled Cut} problem, restricted to clauses of size $2$. 
Second, the considered instances of \gdpcshort{} are $2K_2$-free in the following sense:
for every $B \in \bundles$, let $G_B$ be the undirected graph with vertex set consisting of all vertices
involved in an element of $B$ except for $s$ and $t$, and an edge $uv$ belongs to $E(G_B)$ if there is a clause $\{u,v\} \in B$, an arc $(u,v) \in B$, or an arc $(v,u) \in B$;
we assume that $G_B$ is $2K_2$-free for every $B \in \bundles$. 
In Section~\ref{ss:id2-to-graph} we prove the following (cf. Lemma~\ref{lem:compgdpc}).

\begin{lemma}\label{lem:intro:compgdpc}
There is a randomized polynomial-time algorithm that, 
given on input an instance $\inst$ to \MinSAT{\Gamma} with weights
where $\Gamma$ is bijunctive, of maximum arity $\maxarity$, and for every relation $R \in \Gamma$, the Gaifman graph of $R$ is $2K_2$-free, 
together with a constraint deletion set $Y$ for $\inst$,
outputs an instance $\inst'$ to \gdpcfull{} with not larger $k$ nor the weight budget larger than in $\inst$, 
such that $\inst'$ is $2K_2$-free and $\maxarity$-bounded, and 
\begin{itemize}
\item if $\inst'$ is a \yes-instance, then $\inst$ is a \yes-instance, and
\item if $\inst$ is a \yes-instance, then $\inst'$ is a \yes-instance with probability at least $2^{-\maxarity|Y|}$.
\end{itemize}
Furthermore, there exists a deterministic counterpart of the procedure that, with the same input, runs in time $2^{\maxarity|Y|} \cdot \mathrm{poly}(|\inst|)$
and outputs at most $2^{\maxarity|Y|}$ instances to \gdpcfull{} as above, such that the input instance is a \yes-instance if and only if one of the output instances is a \yes-instance.
\end{lemma}

Recall the bijunctive case covers the \textsc{Almost 2-SAT} problem, generalizing
the algorithm of~\cite{RazgonO09}. We remark that, in Lemma~\ref{lem:intro:compgdpc},
the case of \textsc{Almost 2-SAT} corresponds to the resulting \gdpcfull{} instance 
containing only bundles being singletons. 

\medskip

In the IHS-B case, the clauses can be larger, but the $2K_2$-free assumption
now applies to the arrow graph. A \textsc{Generalized Bundled Cut} instance
is a \emph{$\maxarity$-bounded $2K_2$-free \textsc{Clause Cut}} instance if 
every bundle or clause involves at most $\maxarity$ vertices and
for every $B \in \bundles$,
the following undirected graph $G_B'$ is $2K_2$-free:
$V(G_B')$ consists of all vertices involved in an element of $B$ 
except for $s$ and $t$, and $uv \in E(G_B')$ if there is
an arc $(u,v) \in B$ or an arc $(v,u) \in B$.
In Section~\ref{ss:isd-to-graph}, we prove the following (cf. Lemma~\ref{lem:isd-to-graph}):

\begin{lemma}\label{lem:intro:isd-to-graph}
Given an instance $\inst$ to \MinSAT{\Gamma} (without weights)
  where $\Gamma$ is IHS-B-, of maximum arity $\maxarity$, and for every
  relation $R \in \Gamma$, the arrow graph of $R$ is $2K_2$-free, 
one can in polynomial time compute an equivalent instance of $\maxarity$-bounded
$2K_2$-free \textsc{Clause Cut} with not larger value of $k$.
\end{lemma}

In the reduction from \MinSAT{\Gamma} to \textsc{Generalized Bundled Cut}
(e.g., of Lemmata~\ref{lem:intro:compgdpc} and~\ref{lem:intro:isd-to-graph}),
 the source vertex $s$ should be interpreted
as ``true'' and the sink vertex $t$ as ``false''; other vertices
are in 1-1 correspondence with the variables of the input instance.
Furthermore, arcs are implications that 
correspond to parts of the constraints of the input instance.
The sought $st$-cut $Z$ corresponds
to implications violated by the sought assignment in the CSP instance; a vertex
is assigned $1$ in the sought solution if and only if it is reachable from $s$ in $G-Z$.

Thus, in terms of constraints, arcs $(u,v)$ correspond to implications $(u \to v)$ in the input formula, 
and clauses $\{v_1,\ldots,v_r\}$ correspond to negative clauses $(\neg v_1 \lor \ldots \lor \neg v_r)$
in the input formula. An arc $(s \to v)$ corresponds to a positive 1-clause $(v)$.
Thereby, each bundle naturally encodes an IHS-B- constraint.
Capturing bijunctive constraints requires a little bit more work, since they can also involve
positive 2-clauses $(u \lor v)$, but this can be reduced to the IHS-B- case with clauses of arity 2
via standard methods (e.g., iterative compression followed by variable renaming~\cite{KratschW20}).

Thus, we proceed with the corresponding graph problem: \gdpcfull{} for the bijunctive case
and \textsc{Clause Cut} for the IHS-B- case (and note that the IHS-B+ case follows by symmetry). 
The bulk of Sections~\ref{sec:gaifman} and~\ref{sec:arrow} are devoted to the proofs
of the following theorems.

\begin{restatable}{theorem}{alggdpc}\label{thm:alg-gdpc}
For every integer $\maxarity \geq 2$, 
for the \gdpcfull{} problem restricted to $\maxarity$-bounded $2K_2$-free instances,    
    there exists 
\begin{itemize}
\item a randomized polynomial-time algorithm that never accepts a \no-instance
and accepts a \yes-instance 
with probability $2^{-\mathrm{poly}(k,\maxarity)}$;
\item a deterministic algorithm 
with running time bound $2^{\mathrm{poly}(k,\maxarity)} n^{\Oh(1)}$.
\end{itemize}
\end{restatable}

\begin{restatable}{theorem}{algclausecut}\label{thm:alg-clause-cut}
For every integer $\maxarity \geq 2$, 
for the \textsc{Clause Cut} problem restricted to $\maxarity$-bounded $2K_2$-free instances,    
    there exists 
\begin{itemize}
\item a randomized polynomial-time algorithm that never accepts a \no-instance
and accepts a \yes-instance 
with probability
$2^{-\Oh(k^6 \maxarity^{10} \log(k\maxarity))}$;
\item a deterministic algorithm 
with running time bound
$2^{\Oh(k^6 \maxarity^{10} \log(k\maxarity))} n^{\Oh(1)}$.
\end{itemize}
\end{restatable}

\paragraph{Flow-augmentation.}
In both cases, the first step is to apply \emph{flow-augmentation}~\cite{dfl-arxiv,dfl-stoc}.
Recall that in \textsc{Generalized Bundled Cut},
we are interested in a deletion set $Z \subseteq E(G)$ that separates $t$ and possibly some more vertices of $G$ from $s$.
Formally, $Z$ is a \emph{star $st$-cut} if it is an $st$-cut and additionally
for every $(u,v) \in Z$, $u$ is reachable from $s$ in $G-Z$ but $v$ is not.
That is, $Z$ cuts all paths from $s$ to $t$ and every edge of $Z$ is essential to separate some vertex of $G$ from $s$.
For a star $st$-cut $Z$, its \emph{core}, denoted $\corecutG{Z}{G}$, is the set of those edges $(u,v) \in Z$ such that $t$
is reachable from $v$ in $G-Z$. That is, $\corecutG{Z}{G}$ is the unique inclusion-wise minimal subset of $Z$ that is an $st$-cut.
A simple but crucial observation is that in \textsc{Generalized Bundled Cut}
any inclusion-wise minimal solution $Z$ is a star $st$-cut.

Considered restrictions of 
\textsc{Generalized Bundled Cut} turn out to be significantly simpler if the sought
star $st$-cut $Z$ satisfies the following additional property:
$\corecutG{Z}{G}$ is actually an $st$-cut of minimum possible cardinality. 
This is exactly the property that the flow-augmentation technique provides.

\begin{theorem}[directed flow-augmentation, Theorem~3.1 of~\cite{dfl-arxiv}]\label{thm:dir-flow-augmentation-intro}
There exists a polynomial-time algorithm that, given a directed graph $G$, vertices $s,t \in V(G)$,
and an integer $k$, returns a set $A \subseteq V(G) \times V(G)$
and a maximum flow $\witnessflow$ from $s$ to $t$ in $G+A$ such that 
for every star $st$-cut $Z$ in $G$ of size at most $k$,
with probability $2^{-\Oh(k^4 \log k)}$, the sets of vertices
reachable from $s$ in $G-Z$ and $(G+A)-Z$ are equal (in particular, $Z$ remains
a star $st$-cut in $G+A$),
$\corecutG{Z}{G+A}$ is an $st$-cut of minimum possible cardinality, and
every flow path of $\witnessflow$ contains exactly one edge of $Z$.
\end{theorem}

Furthermore,~\cite{dfl-arxiv} provides a deterministic counterpart of Theorem~\ref{thm:dir-flow-augmentation-intro}
that outputs a family of $2^{\Oh(k^4 \log k)} (\log n)^{\Oh(k^3)}$ candidates for the set $A$ with the guarantee that for every star $st$-cut $Z$, at least one output candidate set $A$ is as in the statement.

We call such a flow $\witnessflow$ a \emph{witnessing flow}. 
Note that each path of $\witnessflow$ is obliged to contain exactly one edge
of $\corecutG{Z}{G+A}$ as $\corecutG{Z}{G+A}$ is an $st$-cut of minimum cardinality
and $\witnessflow$ is a maximum flow. However, we additionally guarantee
that $\witnessflow$ does not use any edge of $Z \setminus \corecutG{Z}{G+A}$.

\paragraph{Bijunctive case.}
The algorithm for the bijunctive case, i.e., Theorem~\ref{thm:alg-gdpc},
can be seen as a wide generalization of the algorithm for \textsc{Weighted Almost 2-SAT}
that was announced in the conference version~\cite{dfl-stoc},
although we will not assume familiarity with that algorithm in this exposition.

Suppose that the given \gdpcshort{} instance $\inst=(G,s,t,\pairs,\bundles,\weight,k,W)$ 
that 
is a \yes-instance with a solution $Z$ of weight at most $W$ and let $\karc:=\abs{Z}$, i.e. the number of edges violated by $Z$, 
and let $\kclause$ be the number of clauses violated by $Z$.
Since $Z$ violates at most $k$ bundles, and every bundle is $\maxarity$-bounded 
and contains at most two arcs and one 
clause connecting the same pair of vertices, we have 
\[ \karc + \kclause \leq k \cdot \binom{\maxarity}{2} \cdot 3 \leq 2k\maxarity^2. \]
As previously observed, we can assume that 
$Z$ is a star $st$-cut and via~\cref{thm:dir-flow-augmentation-intro}, we can further assume (with good enough probability) that $\corecutG{Z}{G}$ 
is an $st$-cut with minimum cardinality and additionally that an $st$-maxflow $\witnessflow$ (with $\lambda_G(s,t) = |\witnessflow|$ flow paths) that is a witnessing flow for $Z$ is given. 
The ethos of the entire algorithm is that  we use the witnessing flow $\witnessflow$ at hand as a guide in search for $Z$ and impose more structure on the search space. The algorithm is randomized and at each step, the 
success probability is at least $2^{-\mathrm{poly}(b,k)}$. Below, we assume that all the guesses up to that point are successful.
(While we describe here the algorithm in the more natural ``random choice'' language, it is straightforward to make it deterministic via branching and the deterministic counterpart of Theorem~\ref{thm:dir-flow-augmentation-intro}.)

The algorithm has three stages, with three important ideas that we want to highlight.
\begin{enumerate}
\item First, we reduce the instance so that $Z$ is an $st$-mincut, rather than just a star $st$-cut.
  Perhaps counterintuitively, we do so by adding additional $st$-paths to $G$, increasing both $\lambda_G(s,t)$ and $|Z|$, but we do so in a way that more of $Z$ is contained in $\corecutG{Z}{G}$ afterwards. 
\item Once $Z$ is assumed to be an $st$-mincut, we introduce a \emph{bipartite structure} to the instance,
  where the paths of $\witnessflow$ are expected to partition into two sets, so that
  clauses only exist between the two sets, and paths consisting of edges not from $\witnessflow$
  only exist between paths in the same set. We show that if the instance does not have this bipartite form,
  then we can branch (with success probability $\Omega(1/\mathrm{poly}(\maxarity,k))$) to simplify the instance. 
\item Finally, once the instance has such a bipartite structure, we show that the problem reduces
  to a weighted, bundled $st$-mincut problem (i.e., without any clauses), which
  can be solved using results from previous work~\cite{dfl-stoc}.
\end{enumerate}

The first stage above is encapsulated in the following lemma, proven in Section~\ref{subsec:2mincut}.
\begin{restatable}{lemma}{tomincut} \label{lem:2mincut}
There is a randomized polynomial-time algorithm which takes as input an  instance $\inst=(G,s,t,\pairs,\bundles,\weight,k,W)$ of \gdpcshort\ 
and outputs an instance $\inst'=(G',s,t,\pairs,\bundles',\weight',k',W')$ with $k'\leq 2k\maxarity^2$ and $W' < 2^{2k\maxarity^2} (W+1)$
such that 
\begin{itemize}
\item if $\inst$ is $2K_2$-free, then $\inst'$ is $2K_2$-free, too,
\item if $\inst'$ is a \yes-instance, then $\inst$ is a \yes-instance, and 
\item if $\inst$ is a \yes-instance, then $\inst'$ has a solution $Z$ with $\weight'(Z')\leq W'$ 
which is an $st$-mincut of $G'$ with probability at least $2^{-\Oh((2k\maxarity^2)^5 \log (2k\maxarity^2))}.$
\end{itemize}
\end{restatable}

Let us sketch the proof of Lemma~\ref{lem:2mincut}.
As $\karc+\kclause \leq 2k\maxarity^2$, each of these integers can be 
correctly guessed with high probability. We may assume that $\karc >\lambda_G(s,t)=\abs{\corecutG{Z}{G}}$; if not, we may either output a trivial \no-instance 
or proceed with the current instance $\inst$ as the desired instance.
When $Z\setminus \corecutG{Z}{G}\neq \emptyset$, this is because there is a clause $p\in \pairs$ that is violated by $\corecutG{Z}{G}$ and some of the extra edges in $Z\setminus \corecutG{Z}{G}$ are used to separate an endpoint  (called an \emph{active vertex}) $v$ of $p$ from $s$ to satisfy $p.$
Note that if we added a path $P$ from $v$ to $t$ while staying on the $t$-side of $Z$ then any arc of $Z$ that serves to separate $v$ from $s$
would now (by definition) belong to $\corecutG{G+P}{Z}$. 
Even though we cannot do this directly, since we cannot sample an active vertex 
with high enough probability, i.e. $1/f(k,b)$ for some $f$ (which would be necessary for an FPT overall success probability and running time bound), we are able to sample a {\sl monotone} sequence of vertices $u_1,\ldots , u_\ell$ with high probability in the following sense: 
There exists a unique active vertex $u_a$ among the sequence, and all vertices before $u_a$ are reachable from $s$ in $G-Z$ and all others are unreachable from $s$ in $G-Z$. 
Once such a sequence is sampled, the $st$-path visiting (only) $s,u_1,\ldots, u_\ell,t$ in order is added as soft arcs, each  forming a singleton bundle of weight $W+1$, 
and we increase the budgets $k$ and $W$ to $k+1$ and $2W+1$ respectively. Let $\inst'$ be the new instance.
Then $Z':=Z\cup \{(u_{a-1},u_a)\}$ (where $u_0=s$) is a solution to $\inst'$ violating at most $k+1$ bundles with weight at most $2W+1$. The key improvement here as above is that $u_a$ is connected to $t$ with a directed path in $G-Z'$, implying that at least one of the extra edges in $Z\setminus \corecutG{Z}{G}$ used to cut $u_a$ from $s$ (in addition to $(u_{a-1},u_a)$) 
is now incorporated into $\corecutG{Z'}{G'}$, thus 
$\abs{Z\setminus \corecutG{Z}{G}} > \abs{Z' \setminus \corecutG{Z'}{G'}}$.
After performing this procedure of ``sampling a (monotone) sequence then adding it as an $st$-path" a bounded number of times, we get 
a \yes-instance that has an $st$-mincut as an optimal solution with high enough probability.

To properly describe the method of sampling such a sequence would go into too much depth for this overview,
but we can describe one critical ingredient. Let $\witnessflow = \{P_1,\ldots,P_\lambda\}$, and for
$i \in [\lambda]$ and a vertex $v \in V(G)$ define the \emph{projection} $\pi_i(v)$ of $v$ to $P_i$ 
as the earliest (i.e., closest to $s$) vertex $u$ on $P_i$ such that there is a $uv$-path $P$ in $G$ where the only
intersection of $P$ and $V(\witnessflow)$ is the starting vertex $u$.
Refer to $P$ as an \emph{attachment path} from $P_i$ to $v$, and let $\pi_i(v)=t$ if no such vertex exists.
We can then define an order on $V(G)$ for every $i \in [\lambda]$, where $u <_i v$ if and only if $\pi_i(u)$
is closer to $s$ on $P_i$ than $\pi_i(v)$. We make two crucial observations.
First, assume that $v <_i v' <_i t$ for $v, v' \in V(G) \setminus V(\witnessflow)$
and let $P$ and $P'$ be the respective attachment paths. Then $P$ and $P'$ are vertex-disjoint.
Second, let $v_1 <_i v_2 <_i \ldots <_i v_\ell <_i t$ be a sequence of vertices 
and assume that $Z$ cuts $P_i$ after $\pi_i(v_j)$ for some $j \in [\ell]$.
Then for all but at most $\karc-\lambda$ vertices $v_a$, $1 \leq a \leq j$,
we have that $v_a$ is in the $s$-side of $Z$ (by the first observation).
Thus, given a sequence with one (unknown) active vertex we can sample
a subsequence such that the prefix up to the active vertex is entirely
in the $s$-side of $Z$.
The full sampling of a monotone sequence now follows by considering a pair of paths $(P_i,P_j)$
and a set of clauses $p=\{u,v\} \in \pairs$ with their projections onto $P_i$ and $P_j$,
chosen under a minimality condition so that the clauses form an \emph{antichain}
--- if $\{u,v\}$ and $\{u',v'\}$ are chosen clauses, with some
orientations $(u,v)$ and $(u',v')$, then $u <_i u'$ if and only if $v' <_j v$. 
From all this it follows that if the antichain contains a clause $p$
with an active endpoint $u$ as above, then with good success probability
the entire antichain forms a monotone sequence in $Z$, as required. 

\medskip

With Lemma~\ref{lem:2mincut} in hand, it suffices to prove the following statement,
encapsulating the second and third phase above. (Theorem~\ref{thm:ID2-mincut}
is proven formally in Section~\ref{ss:ID2-mincut}.)

\begin{restatable}{theorem}{algidmincut}\label{thm:ID2-mincut}
There exists a polynomial-time algorithm that, given a $2K_2$-free \gdpcshort{} instance $\inst = (G,s,t,\pairs,\bundles,\weight,k,W)$, never accepts a \no-instance 
and accepts a \yes-instance that admits a solution $Z$ that is an $st$-mincut with probability $2^{-\Oh(k^{4} \maxarity^{8} \log(k\maxarity))}$.
\end{restatable}

In the second stage of the algorithm (being the first part of the proof of Theorem~\ref{thm:ID2-mincut}), we branch into $f(k)$ instances each of which is either an outright \no-instance 
or {\sl ultimately bipartite} in the following sense. Suppose that an instance $\inst$ admits a vertex bipartition $V(G)\setminus \{s,t\}=V_0\uplus V_1$ such that 
(i) for every flow path $P_i$ there exists $\iota \in \{0,1\}$ such that the vertices of $P_i$
are fully contained in $V_\iota \cup \{s,t\}$,
(ii) for every edge $(u,v)\in E(G)$ there exists $\iota \in \{0,1\}$ such that 
$u,v\in V_\iota\cup \{s,t\}$,
and (iii) for every clause $\{u,v\}\in \pairs$ there is no $\iota \in \{0,1\}$ such that 
 $\{u,v\}\subseteq V_\iota \cup\{s,t\}$.
If the instance at hand is ultimately bipartite, we can eliminate the clauses altogether
(see Figure~\ref{fig:intro-reverse}): 
\begin{enumerate}
\item Break the graph into $G_0 := G[V_0 \cup \{s,t\}]$ and $G_1 = G[V_1 \cup \{s,t\}]$.
\item Reverse the orientations of all edges in $G_1$ and swap the labels
of $s$ and $t$ in this reversed graph (so that the flow paths in $G_1$ still go from $s$ to $t$).
\item Merge back $G_0$ and (reversed) $G_1$ by identifying the two copies of $s$ and 
identifying the two copies of $t$. 
\item For any clause $\{u,v\}$ with $u\in V_0$ and $v\in V_1,$ replace the clause $\{u,v\}$ by the arc $(u,v)$. 
\end{enumerate}
Note that 
the reversal of $G_1$ converts the property (in $G-Z$) of ``$v$ being reachable / unreachable from $s$" to ``$v$ reaching / not reaching $t$" for all $v\in V(G_1)$.
Hence, converting a clause $\{u,v\}$ to an arc $(u,v)$ preserves the same set of violated vertex pairs, except that the violated clauses now become violated edges.
(Another way of viewing the same property is to recall that clauses
correspond to 2-clauses $(\neg u \lor \neg v)$ and arcs to
2-clauses $(\neg u \lor v)$. It is then clear that negating both
variables of an arc, or precisely one endpoint of a clause, 
yields an arc.)

\begin{figure}[tbh]
\begin{center}
\includegraphics{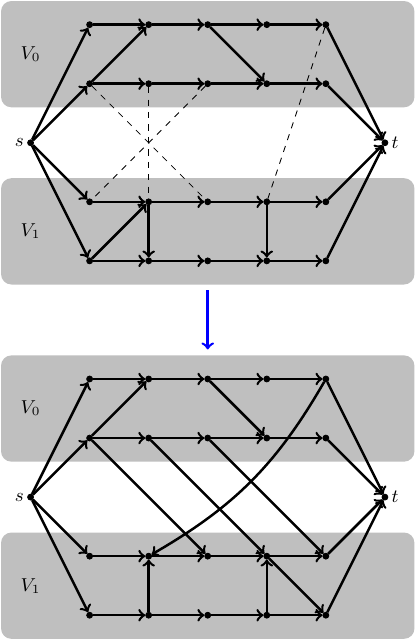}
\caption{An example of the reversing process of the ultimately bipartite 
  instance, where we reverse
the arcs in $V_1$ (bottom half of the graph). Clauses going across become
arcs.}\label{fig:intro-reverse}
\end{center}
\end{figure}

In~\cref{subsub:clean} up to~\cref{subsub:nonzero}, we present the branching strategies and preprocessing steps to reach ultimately bipartite instances. 
In short, we study obstacles to being ultimately bipartite and show that
any such obstacle can be used as a pivot for a branching step.

After obtaining an ultimately bipartite instance and after 
the aforementioned transformation into an instance without clauses (see~\cref{sss:ID2-reversal}), the bundles maintain $2K_2$-freeness.
Now, we observe that 
$2K_2$-freeness is a special case of the property called \emph{pairwise linked deletable edges},
  introduced in~\cite{dfl-stoc},
which makes \textsc{Generalized Bundled Cut} without clauses tractable. 
Hence, the third phase of the algorithm boils down to an application of the corresponding
algorithmic result of~\cite{dfl-stoc}.

\paragraph{IHS-B case.}
The algorithm for the graph counterpart of the IHS-B- case,
   i.e., the algorithm of Theorem~\ref{thm:alg-clause-cut},
   can be seen as a mix of the \textsc{Chain SAT} algorithm
of~\cite{dfl-arxiv,dfl-stoc} and the
\textsc{Digraph Pair Cut} algorithm of~\cite{KratschW20}.
Recall that here the clauses can be larger than $2$ (but of size at most $\maxarity$),
but in a bundle all arcs neither incident with $s$ nor $t$ form a $2K_2$-free graph.
By applying flow-augmentation, we can assume that for the sought solution $Z$ it holds that $\corecutG{Z}{G}$ is an $st$-cut of minimum cardinality
and we have access to a witnessing flow $\witnessflow$.
Every path $P \in \witnessflow$ contains a unique edge $e_P \in E(P) \cap Z$.

We make a number of branching and color-coding steps that make the instance 
(and the sought solution) more regular. 

We first perform the following color-coding step.
For every bundle $B$ and every $P \in \witnessflow$ we randomly guess a value $e(B,P) \in \{\bot\} \cup (E(P) \cap B)$.
We aim for the following: For every $B$ violated by $Z$, and every $P \in \witnessflow$, we want $e(B,P)$ to be the unique edge of $B \cap E(P) \cap Z$, or $\bot$ if there is no such edge.
We also branch into how the edges $e_P$ are partitioned into bundles. That is, for every two distinct $P,P' \in \witnessflow$, we guess if $e_P$ and $e_{P'}$ are from the same bundle.
Note that this guess determines the number of bundles that contain an edge of $\corecutG{Z}{G}$; we reject a guess if this number is larger than $k$.

Assume that we guessed that $e_P,e_{P'} \in B$ for some $P,P' \in \witnessflow$ and a bundle $B$ violated by $Z$. 
Then, as $G_B'$ is $2K_2$-free, either $e_P$ or $e_{P'}$ is incident with $s$ or $t$, or they have a common endpoint,
or there is an arc of $B$ from an endpoint of $e_P$ to an endpoint of $e_{P'}$, 
or there is an arc of $B$ from an endpoint of $e_{P'}$ to an endpoint of $e_P$. 
We guess which cases apply. 
If $e_P$ or $e_{P'}$ is incident with $s$ or $t$, there is only a constant number of candidates for $B$: we guess which bundle is $B$, delete $B$, decrease $k$ by one, and restart the algorithm.
All later cases are very similar to each other; let us describe here the case that $B$ contains an arc $f$ from an endpoint of $e_P$ to an endpoint of $e_{P'}$.

Let $B_1$ and $B_2$ be two arbitrarily chosen bundles that are candidates for bundle $B$,
i.e., they are consistent with the guesses so far.
That is, for $i=1,2$ we have $e_i := e(B_i,P) \neq \bot$, $e_i' := e(B_i,P') \neq \bot$,
and $B_i$ contains an arc $f_i$ that is a candidate for $f$: has its tail in $e_i$ and its head in $e_i'$.
Assume that $e_1$ is before $e_2$ on $P$ but $e_1'$ is after $e_2'$ on $P'$.
The crucial observation (and present also in the algorithm for \textsc{Chain SAT} of~\cite{dfl-stoc}) is that it cannot hold that $B = B_2$, as if we cut $e_2$ and $e_2'$, the
edge $f_1 \in B_1$ will provide a shortcut from a vertex on $P$ before the cut to a vertex on $P'$ after the cut. 
Thus, we may not consider $B_2$ as a candidate for the bundle $B$ violated by $Z$. 
Furthermore, the remaining candidates for the bundle $B$ are linearly ordered along all flow paths $P$ where $B \cap Z \cap E(P) \neq \emptyset$.

This allows the following filtering step: For every bundle $B$, we expect that $\{ P\in \witnessflow~|~e(B,P) \neq \bot\}$ is consistent with the guessing
and there is no other bundle $B'$ with the same set $\{P \in \witnessflow~|~e(B,P) \neq \bot\}$ that proves that $B$ is not violated by $Z$ as in the previous paragraph.
We also expect that $\{e(B,P)~|~P \in \witnessflow\}$ extends to a minimum $st$-cut. 
If $B$ does not satisfy these conditions, we delete it from $\bundles$ (making all its edges and clauses crisp).

Now, a simple submodularity argument shows that the first (closest to $s$) minimum $st$-cut $Z_0$ (where only edges $e(B,P)$ for some $B \in \bundles$ and $P \in \witnessflow$ are deletable)
has the correct structure: its edges are partitioned among bundles as guessed. If $Z_0$ is a solution, we return \yes; note that this is the step where we crucially rely on the instance
being unweighted. 
Otherwise, there is a clause $C \in \pairs$ violated by $Z_0$. 
It is either indeed violated by the sought solution $Z$
or there is $v \in C$ that is not reachable from $s$ in $G-Z$.
We guess which option happens: In the first case, we delete the bundle containing $C$, decrease $k$ by one, and restart the algorithm.
In the second case, we guess $v$ and add a crisp arc $(v,t)$, increasing the size of a minimum $st$-cut, and restart the algorithm. 
This concludes the overview of the algorithm; note that the last branching step is an analog of the core branching step
of the \textsc{Digraph Pair Cut} algorithm of~\cite{KratschW20}.

\paragraph{Dichotomy completion.}
Finally, in order to complete the dichotomy we need to show that the
above two algorithms cover all interesting cases.
In this task, we have significant help from the structure of
\emph{Post's lattice}~\cite{PostsLattice41}.
This is a structural result that precisely characterizes 
the expressive power of every Boolean language under a notion of
expressive power known as \emph{pp-definitions}. While this structure
is too coarse to preserve the FPT status of \MinSAT{\Gamma},
it is very useful as a starting point. Indeed, previous work on CSP
dichotomies for Boolean languages, such as Bonnet et al.~\cite{BonnetELM-arXiv}, 
have frequently used this tool; see Creignou et al.~\cite{CreignouKV08post}.

Section~\ref{sec:dichotomy} contains full proofs for completeness, 
but for the purposes of this overview we can start from a result from
the extended preprint version of a paper of Bonnet et al.~\cite{BonnetELM-arXiv}.
Specifically, they show that for any language $\Gamma$ that is not
IHS-B or bijunctive, \MinSAT{\Gamma} does not even admit a
constant-factor approximation in FPT time, parameterized by $k$.
Clearly, there in particular cannot exist exact FPT algorithms for
such languages, hence we may assume that $\Gamma$ is bijunctive or
IHS-B. By a structural observation, we show that either $\Gamma$
implements positive and negative assignments, i.e., constraints $(x=1)$ and $(x=0)$, or
\WeightedMinSAT{\Gamma} is trivial in the sense that setting all
variables to 1 (respectively to 0) is always optimal. 
Hence we assume that $\Gamma$ implements assignments.
Furthermore, recall that our basic \classWone-hard constraints
are $R_4(a,b,c,d) \equiv (a=b) \land (c=d)$ or its variants
with one or both equalities replaced by implications. 

First assume that $\Gamma$ is bijunctive and not IHS-B. In particular,
every relation $R \in \Gamma$ can be expressed as a conjunction of
1- and 2-clauses, but it does not suffice to use only conjunctions
over $\{(x \lor y), (x \to y), (\neg x)\}$ or
over $\{(\neg x \lor \neg y), (x \to y), (x)\}$ (because $\Gamma$ is not IHS-B). 
It follows from Post's lattice that $\Gamma$ can express the crisp relation $(x \neq y)$.
Furthermore, we assume that there is a relation
$R \in \Gamma$ such that the Gaifman graph $G_R$ contains a $2K_2$.
In fact, assume for simplicity that $R$ is 4-ary and that $G_R$ has
edges $\{1,2\}$ and $\{3,4\}$. Then $R$ must be a ``product''
$
R(a,b,c,d) \equiv R_1(a,b) \land R_2(c,d),
$
where furthermore neither $R_1$ nor $R_2$ implies an assignment,
as such a case would imply further edges of the Gaifman graph.
It is now easy to check that each of $R_1$ and $R_2$ is either
$R_i(x,y) \equiv (\sim\! x = \sim\! y)$
or $R_i(x,y) \equiv (\sim\! x \to \sim\! y)$,
where $\sim\! v$ represents either $v$ or $\neg v$.
It is now not difficult to use $R$ in combination with $\neq$-constraints
to implement a hard relation such as $R_4$, implying \classWone-hardness.

Next, assume that $\Gamma$ is IHS-B, say IHS-B-, but not bijunctive.
Then, again via Post's lattice, we have access to negative 3-clauses $(\neg x \lor \neg y \lor \neg z)$.
We first show that either $\Gamma$ implements equality constraints $(x=y)$,
or \WeightedMinSAT{\Gamma} has a trivial FPT branching algorithm.
We then need to show that
\WeightedMinSAT{(x=1),(x=0),(x=y),(\neg x \lor \neg y \lor \neg z)} is \classWone-hard,
and that \MinSAT{\Gamma} is \classWone-hard if Theorem~\ref{ithm:arrow}
does not apply. 
To describe these  \classWone-hardness proofs, we need a more
careful review of the hardness reduction for \MinSAT{(x=1),(x=0),R_4}.
As is hopefully clear from our discussions, this problem corresponds
to finding an $st$-cut in an auxiliary multigraph $G$, where the edges
of $G$ come in pairs and the cut may use edges of at most $k$
different pairs. We show \classWone-hardness of a further restricted
problem, \textsc{Paired Minimum $st$-cut}, where furthermore the edges
of the graph $G$ are partitioned into $2k$ $st$-paths, i.e., the
$st$-flow in $G$ is precisely $2k$ and any solution needs to cut every
path in precisely one edge.
The remaining hardness proofs now all use the same basic idea. Say that a
pair of $st$-paths $P$ being $(s=x_1=\ldots=x_n=t)$
and $P'$ being $(s=x_n'=\ldots=x_1'=t)$ in a formula $\F$ over $\Gamma$
are \emph{complementary} if any min-cost solution to $\F$
cuts between $x_i$ and $x_j$ if and only if it cuts between $x_j'$ and $x_i'$.
In other words, for a min-cost assignment $\alpha$, we have
$\alpha(x_i) \neq \alpha(x_i')$ for every $i \in [n]$. 
This way, for the purposes of a hardness
reduction from \textsc{Paired Minimum $st$-Cut} only, we can
act as if we have access to $\neq$-constraints
by implementing every path in the input instance
as a pair of complementary paths over two sets of variables $x_v$,
$x_v'$ in the output formula. Indeed, consider a pair $\{\{u,v\},\{p,q\}\}$
of edges in the input instance, placed on two distinct $st$-paths.
To force that the pair is cut simultaneously, we wish to use 
crisp clauses such as $(u \land \neg v \to p)$
and $(u \land \neg v \to \neg q)$, enforcing that if $\{u,v\}$ is cut,
i.e., $\alpha(u)=1$ and $\alpha(v)=0$ for the corresponding assignment
$\alpha$, then $\alpha(p)=1$ and $\alpha(q)=0$ as well.
This is now equivalent to the negative 3-clauses
$(\neg u \lor \neg v' \lor \neg p')$ and
$(\neg u \lor \neg v' \lor \neg q)$.

We can implement such complementary path pairs in two ways,
either with equality, negative 2-clauses, and carefully chosen
constraint weights, for hardness of \WeightedMinSAT{\Gamma},
or with a relation $R$ such that the arrow graph $H_R$ contains a $2K_2$,
for the unweighted case. Here, although the truth is a bit more
complex, we can think of such a relation $R$ as representing
either $R_4$ or a constraint such as the coupled min-cut constraint
$R(a,b,c,d) \equiv (a \to b) \land (c \to d) \land (\neg a \lor \neg c)$.
Note that the construction of complementary path pairs using such a
constraint is straightforward. 

The final case, when $\Gamma$ is both bijunctive and IHS-B-,
works via a similar case distinction, but somewhat more complex
since we need to consider the interaction of
Theorem~\ref{ithm:gaifman} and Theorem~\ref{ithm:arrow}. 
The full characterization of  \MinSAT{\Gamma}
and \WeightedMinSAT{\Gamma} as \classFPT or \classWone-hard,
including explicit lists of the \classFPT cases,
is found in Lemma~\ref{lemma:dich:full-list}
at the end of the paper. 

\paragraph{Structure of the paper.} We review some technical
background in Section~\ref{sec:prel}. We prove
Theorem~\ref{ithm:gaifman} in Section~\ref{sec:gaifman}
and Theorem~\ref{ithm:arrow} in Section~\ref{sec:arrow}.
We then complete the dichotomy with Theorem~\ref{ithm:dichotomy} in Section~\ref{sec:dichotomy}.

\section{Preliminaries}
\label{sec:prel}
\subsection{Cuts and directed flow augmentation}

All graphs considered in this paper are multigraphs unless stated otherwise.
We follow standard graph terminology, e.g., see~\cite{Diestel_book}. 
In this subsection, we introduce some terminology related to cuts and flows.

Let $G$ be a directed graph with two prescribed vertices $s$ and $t.$ 
Each arc is either $\emph{soft}$ or $\emph{crisp}$ (exclusively). An \emph{$st$-path} of $G$ is a directed path from $s$ to $t.$ 
An \emph{$st$-flow} $\witnessflow$ in $G$ is a collection of $st$-paths such that no two paths of $\witnessflow$ 
share a soft arc. The \emph{value} of an $st$-flow $\witnessflow$ is defined to be $\infty$ if it contains an $st$-path consisting of crisp arcs only, 
and $\abs{\witnessflow}$ otherwise. Therefore, one may construe a soft arc as an arc of  capacity 1 and a crisp arc as equivalent to 
infinite copies of capacity 1 arcs. The maximum value of an $st$-flow in $G$ is denoted as $\lambda_G(s,t)$ and 
an $st$-flow $\witnessflow$ whose value equals $\lambda_G(s,t)$ is called an \emph{$st$-maxflow} in $G.$ 
We often refer to a member of an $st$-maxflow $\witnessflow$ as a \emph{flow path} $P_i$ of $\witnessflow$, 
where the index $i$ ranges between 1 and $\lambda_G(s,t).$ For a flow path $P\in \witnessflow$, 
the set $E(P)$ (resp.\ $V(P)$) denotes the set of arcs (resp.\ vertices) which appear in the $st$-path $P$ as an alternating sequence of vertices and arcs.
The sets $E(\witnessflow)$ and $V(\witnessflow)$ are defined as $\bigcup_{i\in [\lambda_G(s,t)]}E(P_i)$ and $\bigcup_{i\in [\lambda_G(s,t)]}V(P_i)$ respectively.

A set $Z\subseteq E(G)$ of {\sl soft} arcs is an \emph{$st$-cut} if it hits every 
$st$-path in $G.$ Note that if $\lambda_G(s,t)=\infty$, there is no $st$-cut in $G.$ 
When $\lambda_G(s,t)$ is finite, the minimum cardinality of an $st$-cut equals $\lambda_G(s,t)$ by Menger's theorem and 
an $st$-cut attaining $\lambda_G(s,t)$ is called an \emph{$st$-mincut}. 
For an $st$-cut $Z$, a vertex $v$ of $G$ said to be \emph{in the $s$-side of $Z$} if it is reachable from $s$ in $G-Z$, and 
\emph{in the $t$-side of $Z$} otherwise. 
We will need the following standard corollary of submodularity of cuts.
\begin{lemma}\label{lem:intersect-cuts}
Let $G$ be a directed graph, let $s,t \in V(G)$, and let $X_1,\ldots,X_\ell$
be a sequence of $st$-mincuts. For $i \in [\ell]$, let $A_i$ be the set of vertices
on the $s$-side of $X_i$. Then the following set is an $st$-mincut as well:
   \[ X := \{(u,v) \in E(G)~|~u \in \bigcap_{i \in [\ell]} A_i \wedge v \in V(G) \setminus \bigcap_{i \in [\ell]} A_i\}.\]
\end{lemma}

Two other types of $st$-cuts play an important role in this paper. An $st$-cut $Z$ 
is said to be \emph{minimal} if for every $e=(u,v)\in Z$, the set $Z-e$ is not an $st$-cut. Note that 
for any minimal $st$-cut $Z$ and $e\in Z$, the tail of $e$ is reachable from $s$ in $G-Z$, i.e. there is a directed 
path from $s$ to the tail of $e$ in $G-Z$, and $t$ is reachable from the head of $e$ in $G-Z.$ Conversely, if  
every edge $e$ of an $st$-cut $Z$ satisfies this property, namely the tail of $e$ is reachable from $s$ and $t$ is reachable from the head of $e$ in $G-Z$, then
$Z$ is a minimal $st$-cut. 

We say that an $st$-cut $Z$ is a \emph{star $st$-cut} if for every edge $e\in Z$, the tail of $e$ is reachable from $s$ in $G-Z$ while 
the head of $e$ is not reachable from $s$ in $G-Z.$ It is easy to see that by discarding arcs of $Z$ from whose head $t$ is not reachable in $G-Z,$ 
one can obtain a (unique) minimal $st$-cut of $G$ as a subset of $Z$. Specifically, the set $\corecutG{Z}{G}\subseteq Z$ of a star $st$-cut $Z$ 
is defined as the subset of arcs $(u,v)\in Z$ such that there is a directed path from $v$ to $t$ in $G-Z.$ Observe that $\corecutG{Z}{G}$ is a minimal $st$-cut of $G.$ 
We frequently omit the subscript $G$ if it is clear from context.

Consider an $st$-maxflow $\witnessflow$ in $G$ and a star $st$-cut $Z\subseteq E(G)$.  We say that 
$\witnessflow$ is a \emph{witnessing flow for $Z$} if $\corecut{Z}=E(\witnessflow)\cap Z$ and $\abs{Z\cap E(P_i)}=1$ for every $i\in [\lambda_G(s,t)]$. 
Notice that the two conditions combined imply that $\corecut{Z}$ is an $st$-mincut whenever there is a witnessing flow for $Z.$
However, in the other direction, even if $\corecut{Z}$ is an $st$-mincut for a star $st$-cut $Z$, a witnessing flow may not exist.
For example, consider the graph $G$ on $V(G)=\{s,a,b,c,t\}$ with the edge set 
$\{(s,a),(a,t), (s,b), (b,c), (c,t), (a,c)\}$ and $Z=\{(s,b),(a,t),(c,t)\}$.
Note that in this example $Z$ is a star $st$-cut, but not a minimal $st$-cut.

The \emph{flow-augmentation technique} proposed in~\cite{ufl-soda} is a machinery 
that delivers a set $A\subseteq V(G)\times V(G)$ of vertex pairs, seen as fresh crisp arcs, that does \emph{not disturb} $Z$ and the augmentation of $G$ with $A$ allows 
a witnessing flow for $Z.$ Formally, let $A\subseteq V(G)\times V(G)$ and we write $G+A$ to denote the graph obtained from $G$ by adding each $(u,v)\in A$ as a crisp arc 
with $u$ as the tail and $v$ as the head.
We say that $A$ is \emph{compatible} with a star $st$-cut $Z$ if every $v\in V(G)$ is in  the same side of $Z$ in both $G$ and $G+A.$ Notice that $A$ 
is  compatible with $Z$ if and only if there is no pair $(u,v)\in A$ such that $u$ is in the $s$-side of $Z$ and $v$ is in the $t$-side of $Z.$ It is easy to verify that, 
for any $A$ compatible with a star $st$-cut $Z$ of $G$, $Z$ remains a star $st$-cut in $G+A$ and $\corecutG{Z}{G}\subseteq \corecutG{Z}{G+A}.$

Trivially $A=\emptyset$ is compatible with $Z$, but this is not a useful ``augmentation.'' The flow-augmentation technique, presented in a black-box manner in the next two statements, 
provides a set $A$ that not only lifts the value $\lambda_{G+A}(s,t)$ to match  $\abs{\corecutG{Z}{G+A}}$ 
but further, $G+A$ now admits a witnessing flow for $Z.$ 
This technique is a fundamental engine behind the positive results obtained in this work. 
(The first, randomized statement is a rephrasing of Theorem~\ref{thm:dir-flow-augmentation-intro}
using the notation already introduced in this section.)

\begin{theorem}[Theorem~3.1 of~\cite{dfl-arxiv}]\label{thm:dir-flow-augmentation}
There exists a polynomial-time algorithm that, given a directed graph $G$, vertices $s,t \in V(G)$,
     and an integer $k$, returns a set $A \subseteq V(G) \times V(G)$
and an $st$-maxflow $\witnessflow$ in $G+A$ such that 
for every star $st$-cut $Z$ of size at most $k$,
with probability $2^{-\Oh(k^4 \log k)}$ $A$ is compatible with $Z$ and 
$\witnessflow$ is a witnessing flow for $Z$ in $G+A$.
\end{theorem}
\begin{theorem}[Theorem~3.2 of~\cite{dfl-arxiv}]\label{thm:dir-flow-augmentation-det}
There exists an algorithm that, given a directed graph $G$, vertices $s,t \in V(G)$,
     and an integer $k$, runs in $2^{\Oh(k^4 \log k)} n^{\Oh(1)}$ time
     and returns a set $\mathcal{F}$ of $2^{\Oh(k^4 \log k)} (\log n)^{\Oh(k^3)}$ pairs $(A,\witnessflow)$ where $A \subseteq V(G) \times V(G)$
and $\witnessflow$ is an $st$-maxflow in $G+A$ such that 
for every star $st$-cut $Z$ of size at most $k$,
 there exists $(A,\witnessflow) \in \mathcal{F}$ such that $A$ is compatible with $Z$ and $\witnessflow$ is a witnessing flow for $Z$ in $G+A$.
\end{theorem}
We recall that the $(\log n)^{\Oh(k^3)}$ factor of Theorem~\ref{thm:dir-flow-augmentation-det}
will not create any troubles for the running time analysis of the algorithms, as for any function $f$ we have
\[ (\log n)^{f(k)} = 2^{f(k) \log \log n} \leq 2^{(f(k))^2 + (\log \log n)^2} = 2^{(f(k))^2} n^{o(1)}. \]

\subsection[$2K_2$-freeness]{$\boldsymbol{2K_2}$-freeness}

For a graph $G$ and two vertices $u,v\in V(G)$, the operation of \emph{identification (or contraction) of $u$ and $v$} replaces $u$ and $v$ with a new vertex adjacent to $(N_G(u) \cup N_G(v)) \setminus \{u,v\}$. 

A graph $G$ is \emph{$2K_2$-free} if it does not contain $2K_2$ as an induced subgraph,
that is, it does not contain two edges $e$ and $e'$ with distinct endpoints and has no other edge 
with one endpoint in $e$ and one endpoint in $e'$.
A directed graph $G$ is $2K_2$-free if and only if its underlying undirected graph is $2K_2$-free.
An important property of $2K_2$-free graphs is that they are closed not only under
vertex deletion, but also under vertex identification (cf. Lemma~6.1 of~\cite{ufl-arxiv}).

\begin{lemma}\label{lem:2k2-free-contraction}
For every positive integer $r$, the class of $rK_2$-free graphs is closed under vertex
identification.
\end{lemma}

\subsection{Constraint Satisfaction Problems}

We recall the basic definitions of CSPs from the introduction.
A constraint language $\Gamma$ over a domain $D$ is a set of
finite-arity relations $R \subseteq D^{r(R)}$ (where $r(R)$ is the arity of
$R$). A \emph{constraint} over a constraint language $\Gamma$ 
is formally a pair $(X,R)$, where $R \in \Gamma$ is a relation from the language,
say of arity $r$, and $X=(x_1,\ldots,x_r)$ is an $r$-tuple of variables 
called the \emph{scope} of the constraint. We typically write our
constraints as $R(X)$ instead of $(X,R)$, or $R(x_1,\ldots,x_r)$ when
the individual participating variables $x_i$ need to be highlighted.
Let $\alpha \colon X \to D$ be an assignment. Then $\alpha$
\emph{satisfies} the constraint $R(X)$ if $(\alpha(x_1),\ldots,\alpha(x_r)) \in R$,
and we say that $\alpha$ \emph{violates} the constraint otherwise.
A \emph{formula over $\Gamma$} is then a conjunction of constraints
over $\Gamma$, and the problem \CSP{\Gamma} is to decide, given a
formula $\F$ over $\Gamma$, whether $\F$ is satisfiable, i.e., if
there is an assignment to all variables that satisfies all constraints of $\F$.
The \emph{cost} of an assignment is the number of violated constraints.
If additionally a \CSP{\Gamma} formula $\F$ is equipped with a weight
function $\weight$ that assigns to every constraint $(X,R)$ a weight $\weight((X,R)) \in \mathbb{Z}_+$, then the \emph{weight} of an assignment is the total weight of all violated
constraints.

A constraint language is \emph{Boolean} if its domain is $D = \{0,1\}$. 
For a Boolean constraint language $\Gamma$, the \MinSAT{\Gamma} problem asks for an assignment of a given formula $\cF$ over $\Gamma$
 that violates at most $k$ constraints, where $k$ is given on input. We are considering the parameterized complexity of \MinSAT{\Gamma}
parameterized by $k$.
In the weighted variant, \WeightedMinSAT{\Gamma},
every constraint has a positive integer weight and we ask for an assignment
of cost at most $k$ and weight at most $W$, where $W \in \mathbb{Z}_+$ is also on the input.

Let us formally define the problem \WeightedMinSAT{\Gamma} for a Boolean constraint language $\Gamma$.
A set of constraints $Z\subseteq \cF$ is a \emph{constraint deletion set}, or simply a \emph{deletion set} if $\cF - Z$ is satisfiable.

\begin{quote}
  \textbf{Problem:} \WeightedMinSAT{\Gamma},
  \\ \noindent \textbf{Input:} A formula $\cF$ over $\Gamma$, a weight function $\weight \colon \cF \to \mathbb{Z}_+$, and integers $k$ and $W$.
  \\ \noindent \textbf{Parameter:} $k$.
  \\ \noindent \textbf{Question:} Is there a constraint deletion set $Z\subseteq \cF$ with $\abs{Z}\leq k$ and $\weight(Z)\leq W$?
  \end{quote}

  It is often convenient to annotate some constraints as not possible to violate.
This can be modelled in \MinSAT{\Gamma} as having $k+1$ copies of the same constraint
or in \WeightedMinSAT{\Gamma} as giving the constraint a prohibitively large weight.
Hence, in what follows we allow the instances to have \emph{crisp} constraints;
an assignment is forbidden to violate a crisp constraint. A constraint that is
not crisp is \emph{soft}.

The next two definitions capture the classes of languages $\Gamma$ that cover our two main
tractable cases.

\begin{definition}
 The \emph{Gaifman graph} of a Boolean relation $R$ of arity $r$ is the undirected graph $G_R$ with vertex set $V(G_R)=[r]$ and an edge $\{i,j\}$ for all $1\leq i<j\leq r$ for which there exist $a_i,a_j\in\{0,1\}$ such that $R$ contains no tuple $(t_1,\ldots,t_r)$ with $t_i=a_i$ and $t_j=a_j$. (In other words, there is an edge $\{i,j\}$ in $G_R$ if and only if constraints $R(x_1,\ldots,x_r)$ have no satisfying assignment $\alpha $ with $\alpha(x_i)=a_i$ and $\alpha(x_j)=a_j$ for some $a_i,a_j\in\{0,1\}$.)
 \end{definition}
 For each integer $\maxarity\geq 2$, the language $\idtwopositive$ consists of all Boolean relations $R$ of arity at most $b$ that are equivalent to a conjunction of two-ary clauses and assignments, and whose Gaifman graph $G_R$ is $2K_2$-free.

%

\begin{definition}
 The \emph{arrow graph} of a Boolean relation $R$ of arity $r$ is the directed graph $H_R$ with vertex set $V(H_R)=[r]$ and an edge $(i,j)$ for all $1\leq i,j\leq r$ with $i\neq j$ such that $R$ contains no tuple $(t_1,\ldots,t_r)$ with $t_i=1$ and $t_j=0$ but does contain such tuples with $t_i=t_j=0$ respectively $t_i=t_j=1$. (In other words, there is an edge $(i,j)$ in $H_R$ if and only if constraints $R(x_1,\ldots,x_r)$ have no satisfying assignment $\alpha$ with $\alpha(x_i)=1$ and $\alpha(x_j)=0$ but do have satisfying assignments $\alpha$ with $\alpha(x_i)=\alpha(x_j)=0$ respectively $\alpha(x_i)=\alpha(x_j)=1$.)
 
\end{definition}
 For each integer $\maxarity\geq 2$, the language $\isdpositive$ consists of all Boolean relations $R$ of arity at most $b$ that are equivalent to a conjunction of implications, negative clauses (of arity up to $\maxarity$), and assignments, and whose arrow graph $H_R$ is $2K_2$-free, i.e., the underlying undirected graph is $2K_2$-free.

%

\section{Bijunctive positive case}
\label{sec:gaifman}

Recall that the constraint language $\idtwopositive$ consists of all Boolean relations $R$ of arity at most $b$ such that the Gaifman graph $G_R$ is $2K_2$-free and $R$ can be expressed as a conjunction of 2-clauses, i.e. clauses on two variables, and  the unary relations $\{(x=1),(x=0)\}$. In other words, each constraint $R(x_1,\ldots ,x_r)$ of a formula $\cF$ over $\idtwopositive$ can be written as a 2-CNF formula $\phi_R$ on the 
variables $\{x_1,\ldots, x_r\}$.
This section is devoted to establishing our first new tractable case, captured by the following theorem.

\begin{theorem}\label{thm:alg-ID2}
For every integer $\maxarity \geq 2$, there exists
\begin{itemize}
\item a randomized polynomial-time algorithm for \WeightedMinSAT{\idtwopositive}
that never accepts a \no-instance and accepts a \yes-instance with probability $2^{-\mathrm{poly}(k,b)};$
\item a deterministic algorithm solving \WeightedMinSAT{\idtwopositive} in time $2^{\mathrm{poly}(k,b)} n^{\Oh(1)}$.
\end{itemize}
\end{theorem}
Throughout the whole section, the main narrative is in the first, more natural, randomized setting. 
The deterministic counterpart follows via straightforward modification: replace every random choice step by branching and every call to Theorem~\ref{thm:dir-flow-augmentation} with a call to Theorem~\ref{thm:dir-flow-augmentation-det}.

\subsection{Reducing to a graph problem}\label{ss:id2-to-graph}
The first step in the proof of Theorem~\ref{thm:alg-ID2} is to reduce to a cut problem in digraphs:
a variant of \textsc{Generalized Bundled Cut} with a few extra properties,
  which we call \gdpcfull{} (in short \gdpcshort).

An instance of \gdpcfull{} consists of 
\begin{itemize}
\item a directed multigraph $G$ with two distinguished vertices $s$ and $t$, 
\item a multiset $\pairs \subseteq V(G) \times V(G)$ of unordered vertex pairs called the \emph{clauses},
\item a collection $\bundles$ of pairwise disjoint subsets, called \emph{bundles}, of $E(G)\cup \pairs$, 
\item a weight function $\weight \colon \bundles \to \mathbb{Z}_+$, and
\item integers $k$ and $W$.
\end{itemize}

We remark that the disjointness of bundles can be always assumed even if it is not part of the problem description. Indeed, 
if the same copy of an arc $(u,v)$ or a clause $\{u,v\}$ appears in multiple bundles, then treating them as distinct copies does not change the solution. 
This is because if the original copy is satisfied (resp. violated) by $Z$ before duplication, then after duplication all the copies in distinct bundles will be satisfied (resp. must be violated) by $Z'$ 
where all the copies of violated arcs of $Z$ are included, and vice versa. 
Therefore, we may always assume that the instance $\inst$ keeps multiple copies of the same arc or clause to maintain the disjointness and it is convenient to add 
this assumption as part of the problem definition.
Note also that if a bundle contains multiple copies of the same arc or the same clause then we can remove all copies except for one. 

For a bundle $B$, we define a graph $G_B$ as the (undirected, simple) graph on the vertex set $V(B)\setminus \{s,t\}$ 
whose edge set consists of all pairs $uv$, with $u,v \in V(G) \setminus \{s,t\}$ such that $B$ contains an arc $(u,v)$, an arc $(v,u)$, or a clause $\{u,v\}$.
An instance $\inst$ is said to be \emph{$2K_2$-free} if $G_B$ does not contain $2K_2$ as an induced subgraph for every bundle $B$. 
Similarly, we say that $\inst$ is \emph{$\maxarity$-bounded} for an integer $\maxarity$ if $\abs{V(B)}\leq \maxarity$ for every bundle $B$.
A (copy of an) arc or a clause $p$ is \emph{soft} if $p \in \bigcup_{B \in \bundles}B$, i.e., if it is contained in some bundle $B\in\bundles$, and \emph{crisp} otherwise.
 
An $st$-cut $Z\subseteq E(G)$ of $G$ is said to \emph{violate} an arc $(u,v) \in E(G)$ if $Z$ contains $(u,v)$, and to \emph{satisfy} $(u,v)$ otherwise. 
Moreover, $Z$ \emph{violates} a clause $p\in \pairs$ if both endpoints of $p$ are in the $s$-side of $Z$, and \emph{satisfies} $p$ otherwise. A bundle $B\in \bundles$ 
is said to be \emph{violated} by an $st$-cut $Z$ if $Z$ violates at least one member of $B,$ otherwise it is satisfied by $Z$. 
A set $Z\subseteq E(G)$ is a \emph{solution} to $\inst$ if it is an $st$-cut that satisfies all crisp arcs and clauses
and violates at most $k$ bundles. The \emph{weight} $\weight(Z)$ of a solution $Z$ is the sum of the weights of bundles violated by $Z$.

\begin{quote}
  \textbf{Problem:} \gdpcfull\ (\gdpcshort\ in short)
  \\ \noindent \textbf{Input:} A (directed) graph $G$ with two vertices $s,t$, a multiset $\pairs \subseteq \binom{V(G)}{2}$ of unordered vertex pairs, 
  a collection $\bundles$ of pairwise disjoint subsets of $E(G)\cup \pairs$, a weight function $\weight \colon \bundles \to \mathbb{Z}_+$, integers $k$ and $W$,
  \\ \noindent \textbf{Parameter:} $k$
  \\ \noindent \textbf{Question:} Is there a solution $Z$ with $\weight(Z)\leq W$? 
  \end{quote}

We point out that the \textsc{Digraph Pair Cut} problem, introduced and shown to be fixed-parameter tractable in~\cite{KratschW20},
is a special case of \gdpcshort\ in which bundles are the singletons corresponding to the arcs of $G$.
The following lemma (being a precise statement of Lemma~\ref{lem:intro:compgdpc}) 
  shows how to reduce \WeightedMinSAT{\idtwopositive} equipped with a constraint deletion set to \gdpcshort.

\begin{lemma}\label{lem:compgdpc}
There is a randomized polynomial-time algorithm that, 
      given on input an instance $\inst=(\cF,\weight,k,W)$ to \WeightedMinSAT{\idtwopositive}
together with a constraint deletion set $Y$ for $\cF$, outputs an instance $\inst'=(G,s,t,\pairs,\bundles,\weight',k',W')$ to \gdpcfull{} with $k' \leq k$, $W' \leq W$
that is $2K_2$-free and $\maxarity$-bounded and 
such that 
\begin{itemize}
\item if $\inst'$ is a \yes-instance, then $\inst$ is a \yes-instance, and
\item if $\inst$ is a \yes-instance, then $\inst'$ is a \yes-instance with probability at least $2^{-\maxarity|Y|}$.
\end{itemize}
Furthermore, there exists a deterministic counterpart of the procedure that, with the same input, runs in time $2^{\maxarity|Y|} \cdot \mathrm{poly}(|\inst|)$
and outputs at most $2^{\maxarity|Y|}$ instances to \gdpcfull{} as above, such that the input instance is a \yes-instance if and only if one of the output instances is a \yes-instance.
\end{lemma}
\begin{proof}
We start with the following preprocessing on $\inst$.
\begin{claim}\label{cl:preprocess-all-0}
Given $\inst$ and $Y$ as in the lemma statement, one
can in polynomial time compute an equivalent \WeightedMinSAT{\idtwopositive}
instance $\inst_0$ together with a constraint deletion set $Y_0$ of size $|Y|$
such that the assignment 
that assigns $0$ to every variable satisfies all constaints of $\inst_0$ outside $Y_0$.
\end{claim}
\begin{proofofclaim}
We find an assignment $\alpha$ that satisfies all constraints of $\inst$ except for $Y$; note that this is possible
in polynomial time as every constraint in $\inst$ can be expressed as a conjunction of 2-CNF clauses and assignments, and hence it boils down to solving a 2-CNF SAT instance.
Observe that the set $\idtwopositive$ is closed under negating variables:
Given an $r$-ary relation $R \in \idtwopositive$ and an index $1 \leq i \leq r$,
there is a relation $R' \in \idtwopositive$ that is $R$ with ``negated $i$-th variable'', i.e., $(a_1,\ldots,a_r) \in R \Leftrightarrow (a_1,\ldots,a_{i-1},1-a_i, a_{i+1},\ldots,a_r) \in R'$. 
Indeed, if $\phi_R$ is a definition of $R$ as a conjunction of assignments and binary clauses, one can construct a formula $\phi_{R'}$ witnessing $R' \in \idtwopositive$
by negating all occurences of the $i$-th coordinate of $R$ (i.e., negating a literal in a binary clause and swapping assingment $x=0$ with $x=1$).
Hence, we can construct the instance $\inst_0$ from $\inst$ by replacing every variable $x$ that is assigned to $1$ by $\alpha$ by its negation.
That is, for every variable $x$ with $\alpha(x) = 1$, we replace every constraint involving $x$ with its version with negated $x$.
After this operation, an assignment that assigns $0$ to every variable satisfies all constraints outside $Y$. We return the modified instance and the modified set $Y$ as $\inst_0$ and $Y_0$.
\end{proofofclaim}

By somewhat abusing the notation, we denote the instance and constraint deletion set
produced by Claim~\ref{cl:preprocess-all-0}
again by $\inst = (\cF, \weight,k, W)$ and $Y$. That is,
in what follows we assume that an assignment $\alpha$ that
sets $\alpha(x) = 0$ for each variable $x$ satisfies all constraints of $\inst$ outside $Y$.

Let $X$ be the set of variables that appear in constraints of $Y$; note that $|X| \leq \maxarity|Y|$. 
We randomly sample an assignment $\beta \colon X \to \{0,1\}$, aiming for the following: If there is an assignment $\beta'$ of all variables of $\inst$ that violates at most $k$ constraints of $\inst$
of total weight at most $W$, then there is one that extends $\beta$. Note that if $\inst$ is a \yes-instance, then we are successful with probability at least $2^{-|X|} \geq 2^{-\maxarity|Y|}$. 
(In the deterministic counterpart, this step is replaced with branching.)
Let $Y_0$ be the set of constraints of $\cF$ that are violated by all extensions of $\beta$ and let $Y_1$ be the set of all constraints that are satisfied by all extensions of $\beta$. (Note that both sets can be efficiently computed because each constraint over $\idtwopositive$ is equivalent to a 2-SAT formula; furthermore, this formula can be hardcoded in the algorithm or computed by brute-force as $\maxarity$ is a constant.)
By choice of $X$, every constraint of $Y$ has all its variables in the domain of
$\beta$ and thus $\beta$ fully determines if it is satisfied.
Consequently, we have $Y \subseteq Y_0 \cup Y_1$.
Let $\cF' = \cF \setminus (Y_0 \cup Y_1)$, $k' = k - |Y_0|$, and $W' = W - \weight(Y_0)$. 
If $k' < 0$ or $W' < 0$, then we return a trivial \no-instance.

We create a \gdpcfull{} instance $\inst'=(G,s,t,\pairs,\bundles,\weight',k',W')$ as follows.
The values $k'$ and $W'$ are as defined in the previous paragraph.
The vertex set $V(G)$ consists of special vertices $s$ and $t$ and all variables of $\inst$.
For every variable $x \in X$, we add a crisp arc $(s,x)$ if $\beta(x) = 1$ and a crisp arc $(x,t)$ if $\beta(x) = 0$.

For each constraint $\constraint=R(x_1,\ldots,x_r)$ of $\cF'$, we proceed as follows.
As $R\in\idtwopositive$, the constraint $\constraint$ can be expressed as a conjunction of assignments and two-ary clauses. Let $\phi_{\constraint}$ be the formula over variables $x_1,\ldots,x_r$ that is the conjunction of all assignments and binary clauses that hold for every satisfying assignment of $\constraint$.
Since $R$ is itself definable as a conjunction of assignments and two-ary clauses, we have
\[
\forall_{a_1,\ldots,a_r \in \{0,1\}} \left((a_1,\ldots,a_r) \in R \iff \phi_{\constraint}(a_1,\ldots,a_r) = 1\right).
\]
Recall that every constraint of $\cF'$ is satisfied by the all-0 assignment $\alpha$.
Therefore, each $\phi_{\constraint}$ is in fact a conjunction of:
\begin{itemize}
\item 0-assignments, i.e., $(x=0)$, which can be seen as $(x \to 0)$,
\item negative $2$-clauses, i.e., $(\neg x \vee \neg y)$, and
\item implications, i.e., $(\neg x \vee y)$, which we rewrite as $(x \rightarrow y)$.
\end{itemize}
For every constraint $\constraint$ of $\cF'$
we add to $G$
every implication $(x \to y)$ in $\phi_{\constraint}$ as an arc $(x,y)$, every negative $2$-clause $(\neg x \vee \neg y)$ as a clause $\{x,y\}$,
and every assignment $(x=0)$ as an arc $(x,t)$. Furthermore, 
we create a new bundle consisting of all arcs and clauses added for \constraint; the weight of the bundle is set to the weight of \constraint.
This finishes the description of the output \gdpcshort{} instance $\inst'=(G,s,t,\pairs,\bundles,\weight',k',W')$.

Clearly, the output instance is $\maxarity$-bounded, as every constraint in the input instance has arity at most $\maxarity$.
To see that the output instance is $2K_2$-free, note that for every output bundle $B$ created for a constraint $\constraint=R(x_1,\ldots,x_r)$, the graph $G_B$ can be obtained
from the Gaifman graph $G_R$ of $R$ through a number of vertex deletions and vertex identifications
(i.e., whenever $x_i$ and $x_j$ for some $i,j \in [r]$ is the same variable in $\constraint$,
 the vertices of $G_R$ corresponding to $x_i$ and $x_j$  are identified in $G_B$, and whenever $x_i$ is a constant $0$ or $1$ in $\constraint$, the corresponding vertex of $G_R$ is deleted in $G_B$),
     so $2K_2$-freeness follows from Lemma~\ref{lem:2k2-free-contraction}.

In the following two claims we check the promised equivalence of instances $\inst$ and $\inst'$.
\begin{claim}
If $\inst'$ is a \yes-instance, then so is $\inst$.
\end{claim}
\begin{proofofclaim}
Let $Z$ be an $st$-cut in $\inst'$ that violates at most $k'$ bundles of total weight at most $W'$ and does not violate any crisp arc or clause. 
Define $\beta'(x) = 1$ if $x$ is reachable from $s$ in $G-Z$ and $\beta'(x) = 0$ otherwise. 
Due to the crisp arcs $(s,x)$ or $(x,t)$ added for every variable $x \in X$, and since $t$ is not reachable from $s$ in $G-Z$, 
we have that $\beta'$ extends $\beta$. Hence, $\beta$ satisfies $Y_1$ and violates $Y_0$ and it remains to show that if $\beta'$
violates some constraint $\constraint$ in $\cF'$, then $Z$ violates the corresponding bundle. 
To see this, note that if $\constraint$ is violated by $\beta'$, then one of the clauses or assignments of $\phi_{\constraint}$ is violated by $\beta'$,
   which in turn corresponds to an arc or a clause violated by $Z$ in $\inst'$, implying violation of the corresponding bundle.
\end{proofofclaim}

\begin{claim}
If $\inst$ is a \yes-instance that admits
an assignment $\beta'$ extending $\beta$ that violates at most $k$ constraints of total weight at most $W$, then $\inst'$ is a \yes-instance.
\end{claim}
\begin{proofofclaim}
Let $\beta'$ be as in the statement.
We define a set $Z$ of arcs and $C$ of pairs in $\inst'$ as follows.
For every constraint $\constraint$ in $\cF'$ that is violated by $\beta'$, 
we look at the formula $\phi_{\constraint}$ and add to $Z$ and $C$ all arcs and pairs, respectively,
corresponding to the assignments and clauses of $\phi_{\constraint}$ that are violated by $\beta'$. 
It suffices to show that $Z$ is an $st$-cut and that every clause violated by $Z$ is contained in $C$.

Recall that every arc $(x,y)$ of $G$ corresponds to an implication $(x \to y)$,
to an assignment $(x=0)$
in the formula $\phi_{\constraint}$ for some constraint $\constraint$, to $y \in X$ with $\beta(y) = 1$ while $x=s$, or to $x \in X$ with $\beta(x) = 0$ while $y=t$. 
In particular, for every arc $(s,x) \notin Z$ we have $\beta(x)=1$ and for every arc $(x,t) \notin Z$
we have $\beta(x)=0$.
Consequently, for every vertex $x$ reachable from $s$ in $G-Z$, we have $\beta(x)=1$,
  and for every vertex $x$ from which $t$ is reachable in $G-Z$, we have $\beta(x)=0$.
Since $G$ does not contain an arc $(s,t)$, this implies that
$t$ is not reachable from $s$ in $G-Z$. 
Furthermore, if a pair $\{x,y\}$ is violated by $Z$, then both $x$ and $y$ are reachable from $s$ in $G-Z$, and hence $\beta(x)=\beta(y)=1$, and the corresponding
clause is in $C$. This concludes the proof of the claim.
\end{proofofclaim}

This finishes the proof of the lemma.
\end{proof}

Lemma~\ref{lem:compgdpc} allows us to focus on \gdpcfull.
The proof of Theorem~\ref{thm:alg-gdpc}, restated below,
    spans the remainder of this section.

\alggdpc*

Here, we formally check that~\cref{lem:compgdpc} and~\cref{thm:alg-gdpc} together yield the claimed algorithm for \WeightedMinSAT{\idtwopositive}.

\begin{proof}[Proof of \cref{thm:alg-ID2}]
Let $\inst=(\cF,\weight,k,W)$ be an instance to \textsc{Weighted Min SAT$(\idtwopositive)$}. 

The \textsc{Almost 2-SAT} problem is the \textsc{Min SAT$(\Gamma)$} problem where
$\Gamma$ is the set of all unary and binary clauses.
We create an instance $\inst_0 = (\Phi, k')$ of \textsc{Almost 2-SAT} as follows.
The set of variables is the same as in $\inst$.
For every constraint $\constraint$ used in $\inst$, we consider its 2-CNF definition $\phi_{\constraint}$, and 
we add all clauses of $\phi_{\constraint}$ to $\Phi$. Note that as $\constraint$ has arity at most $\maxarity$, the formula $\phi_{\constraint}$ has at most
$4 \cdot \binom{\maxarity}{2} + 2\maxarity = 2\maxarity^2$ clauses. We set $k' = 2k\maxarity^2$. 
Note that if $\inst$ is a \yes-instance, so is $\inst_0$.

We solve $\inst_0$ with the (randomized) polynomial-time $\Oh(\sqrt{\log \mathrm{OPT}})$-approximation algorithm
of~\cite[Theorem~1.7]{KratschW20} or the fixed-parameter algorithm of~\cite[Theorem~7]{RazgonO09} 
with running time bound $\Oh(15^{k'} \cdot k' \cdot |\Phi|^3)$ in the deterministic case.
We expect a set $Y_0$ of size $\Oh(k' \sqrt{\log k'})$ of clauses
whose deletion makes $\Phi$ satisfiable, as otherwise 
$\inst_0$ is a \no-instance and we can report that $\inst$ is a \no-instance, too.
(Note that with exponentially small probability we can reject a \yes-instance here, because the approximation algorithm is randomized.)

Hence, the set $Y$ of $\Oh(k' \sqrt{\log k'})$ constraints of $\inst$ that contain
the clauses of $Y_0$ is a constraint deletion set of $\inst$.
We plug it into Lemma~\ref{lem:compgdpc} and pass the resulting
\gdpcshort{} instance to Theorem~\ref{thm:alg-gdpc}.
%
\end{proof}

%
%
%
%

%
%
%
In the narrative that follows, for clarity we focus only on the randomized part of Theorem~\ref{thm:alg-gdpc}.
For the deterministic counterpart, the proof is the same, but every randomized step needs to be replaced
either by branching in the straightforward manner, or is in fact a color-coding step and can be derandomized
using the standard framework of \emph{splitters}~\cite{NaorSS95} (see also~\cite{the-book}),
or requires the usage of Theorem~\ref{thm:dir-flow-augmentation-det} instead of Theorem~\ref{thm:dir-flow-augmentation}.

\subsection{Reducing to a mincut instance}\label{subsec:2mincut}

The first step in the algorithm of Theorem~\ref{thm:alg-gdpc} is to reduce to the case
where the sought solution $Z$ is an $st$-mincut. 
Formally, in this section we prove Lemma~\ref{lem:2mincut}.

\tomincut*

A solution $Z\subseteq E(G)$ is a \emph{minimal solution} to a \gdpcshort{} instance $\inst=(G,s,t,\pairs,\bundles,\weight,k,W)$ if for every arc $e\in Z$, 
either $Z-e$ is not an $st$-cut or there is a bundle violated by $Z-e$ that is not violated by $Z$.
It is easy to see that a minimal solution $Z$ is a star $st$-cut. Indeed, suppose that there is an arc $e=(u,v)$ of $Z$ such that 
$u$ is in the $t$-side of $Z$ or $v$ is in the $s$-side of $Z$. Then $Z-e$ is an $st$-cut and any vertex which is in the $t$-side of $Z$ is also in the $t$-side of $Z-e$, 
that is, $Z-e$ is a solution to $\inst,$ contradicting the minimality of $Z$.

An \emph{optimal solution} to $\inst$ is a solution of minimum weight of violated bundles that is additionally minimal.
Clearly, if $\inst$ is a \yes-instance, then there exists a solution $Z$ that is optimal.

The fact that a minimal solution is a star $st$-cut allows us to apply Theorem~\ref{thm:dir-flow-augmentation} and augment the value of an $st$-maxflow. 
Let $\karc:=\abs{Z}$ and let $\kclause$ be the number of clauses violated by $Z$. 
Notice that $Z$ as well as the clauses violated by $Z$ are multisets, and hence the multiplicities of arcs and 
clauses are counted for the values of $\karc$ and $\kclause$.
Since a single bundle does not contain repeated arcs nor repeated clauses, each violated
bundle contributes at most $2\binom{\maxarity}{2}$ to $\karc$ and $\binom{\maxarity}{2}$ to $\kclause$.
Consequently, 
\begin{equation}\label{eq:karc-bound}
\karc \leq 2k\binom{\maxarity}{2},\qquad \kclause \leq k\binom{\maxarity}{2}, \qquad \karc + \kclause \leq 3k\binom{\maxarity}{2} \leq 2k\maxarity^2.
\end{equation}

The next lemma is a restatement of Theorem~\ref{thm:dir-flow-augmentation} for the problem \gdpcshort, 
with $\karc$ to be the parameter $k$ in the theorem.

\begin{lemma} \label{lem:mincutcore}
There is a polynomial-time algorithm which, given an instance $\inst=(G,s,t,\pairs,\bundles,\weight,k,W)$ of \gdpcshort{}
and an integer $\karc$,
outputs a set $A \subseteq V(G) \times V(G)$ and an $st$-maxflow $\witnessflow$ in $G+A$ such that 
for every optimal solution $Z$ to $\inst$ with $|Z| \leq \karc$,
 with probability $2^{-\Oh(\karc^4 \log \karc)},$ $Z$ is an optimal solution to $\inst'=(G+A,s,t,\pairs,\bundles,\weight,k,W)$, 
$\corecutG{Z}{G+A}$ is an $st$-mincut of $G+A$ and $\witnessflow$ is a witnessing flow for $Z$ in $G+A$. 
\end{lemma}
Notice that the procedure of Lemma~\ref{lem:mincutcore} does not create a new solution which is not a solution to the initial instance. 
We will only use the promise of Lemma~\ref{lem:mincutcore} for optimal solutions
$Z$ of size exactly $\karc$. 


The next lemma formalizes one step of the augmentation procedure we use in this section.
\begin{lemma}\label{lem:absorb}
There is a polynomial-time algorithm which takes as input an instance $\inst=(G,s,t,\pairs,\bundles,\weight,k,W)$ of \gdpcshort, an $st$-maxflow $\witnessflow$ of $G$ 
and an integer $\karc$ such that $\karc-\lambda_G(s,t)>0$ and $\lambda_G(s,t)>0$
and outputs an instance $\inst^+=(G^+,s,t,\pairs,\bundles^+,\weight^+,k+1,W^+)$ with $W^+=2W+1$, and $st$-maxflow $\witnessflow^+$ of $G^+$ and an integer $\karc^+:=\karc+1$ such that 
\begin{itemize}
\item if $\inst$ is $2K_2$-free, then $\inst^+$ is $2K_2$-free, too,
\item if $\inst^+$ is a \yes-instance, then $\inst$ is a \yes-instance, and 
\item if $\inst$ is a \yes-instance with an optimal solution $Z$ with $\abs{Z}= \karc$ for which $\witnessflow$ is a witnessing flow, then 
$\inst^+$ is a \yes-instance with an optimal solution $Z^+$ such that $\abs{Z^+}= \karc^+$ and $\abs{Z^+}-\lambda_{G^+}(s,t) < \abs{Z}-\lambda_G(s,t)$, 
and for which $\witnessflow^+$ is a witnessing flow
with probability at least $2^{-\Oh(\karc^4 \log \karc)}$.
\end{itemize}
\end{lemma}

Before we prove Lemma~\ref{lem:absorb}, let us show how to use it 
to prove Lemma~\ref{lem:2mincut}.
\begin{proof}[Proof of Lemma~\ref{lem:2mincut}.]
Suppose that $\inst$ is a \yes-instance with an optimal solution $Z.$ 
As already announced earlier in this section,
we start with guessing two numbers $\karc,\kclause\in [2k\maxarity^2]$.
With probability at least $1/(2k\maxarity^2)^2$, $Z$ violates exactly $\karc$ arcs and $\kclause$ clauses of $\inst.$ 
Note that $\karc+\kclause\leq 2k\maxarity^2$.
We fix the initial values of $\karc$, 
and $\kclause$ as $\karc^{(0)}$ and $\kclause^{(0)}$ respectively throughout the proof. 

As an initialisation step, assume that $\lambda_G(s,t)=0$. 
We bootstrap the flow as follows. 
Let $\mathcal{C}$ be the set of all clauses $C$ such that both endpoints are reachable from $s$.
If $\mathcal{C}$ contains a crisp clause, let $C$ be such a clause. Otherwise, we select a clause 
$C \in \mathcal{C}$ uniformly at random. Note that $C$ is satisfied by $Z$ with probability
$\Omega(1/\kclause)$. We assume that this happens.  Let $C=\{u,v\}$. 
Now, randomly add either a crisp arc $(u,t)$ or a crisp arc $(v,t)$.
This arc is compatible with $Z$ with probability at least $1/2$. Assume that this happens,
and w.l.o.g.\ that the added arc was $(u,t)$. Then the path from $s$ to $u$ plus the arc $(u,t)$ witness 
that $\lambda_G(s,t)>0$. Note that none of the guesses in this step affect $Z$ or the values $\kclause$ and $\karc$.

Having guaranteed $\lambda_G(s,t)>0$, 
we apply the algorithm of Lemma~\ref{lem:mincutcore} with the value $\karc^{(0)}$ and obtain an instance $\inst_0$ with a non-empty $st$-maxflow $\witnessflow_0.$
The nonemptiness of $\witnessflow_0$ is required for the later usage of Lemma~\ref{lem:absorb}.

For $i\geq 0,$ given an instance $\inst_i=(G_i,s,t,\pairs,\bundles_i,\weight_i,k_i,W_i)$ of \gdpcshort, an $st$-maxflow $\witnessflow_i$ of $G_i$ and an integer $\karc_i$ 
we recursively perform the following.
\begin{itemize}
\item If $\karc_i-\lambda_{G_i}(s,t)<0$, then reject the current guess (formally, return a trivial \no-instance).
\item If $\karc_i-\lambda_{G_i}(s,t)=0$, then output the current instance $\inst_i$ and $\karc_i$ as $\inst'$ and $\karc'$. 
\item If $\karc_i-\lambda_{G_i}(s,t)>0,$ we apply the algorithm of Lemma~\ref{lem:absorb} and proceed with the resulting output 
$\inst_{i+1}=(G_{i+1},s,t,\pairs,\bundles_{i+1},\weight_{i+1},k_{i+1},W_{i+1})$ with $k_{i+1}=k_i+1$, $W_{i+1}=2W_i+1$, $st$-maxflow $\witnessflow_{i+1}$ of $G_{i+1}$ and $\karc_{i+1}:=\karc_i+1.$
\end{itemize}

Assume that the initial call for the flow-augmentation in Lemma~\ref{lem:mincutcore} and the subsequent calls to the algorithm of Lemma~\ref{lem:absorb} 
make correct guesses and let $Z_i$ be a solution which satisfies the condition of Lemma~\ref{lem:absorb} for each $\inst_i.$ 
Because the value of  $\abs{Z_i}-\lambda_{G_i}(s,t)$ strictly drops at each iteration due to Lemma~\ref{lem:absorb} and $\abs{Z_0}=\karc^{(0)},$ 
 the recursive calls return an instance $\inst'$ after at most $\karc^{(0)}$ iterations (otherwise we reject the current guess). 
Consequently, by~\eqref{eq:karc-bound} it holds that
\[
k'\leq k+\karc^{(0)} \leq k+3k{\maxarity \choose 2}\leq 2k\maxarity^2, \qquad W' < 2^{2k\maxarity^2} (W+1), \qquad \text{and} \qquad \abs{Z_i}\leq 2\karc^{(0)} \text{ for every $i$}
\]
and the probability that the at most $\karc^{(0)}$ recursive calls to Lemma~\ref{lem:absorb} are successful is at least 
\[
2^{-\Oh(\karc^{(0)}\cdot (2\karc^{(0)})^4 \log (2\karc^{(0)}))}=2^{-\Oh((2k\maxarity^2)^5 \log (2k\maxarity^2))}.
\]
That the first two statements hold is clear from the construction and Lemma~\ref{lem:absorb}. This completes the proof.
\end{proof}

It remains to prove Lemma~\ref{lem:absorb}.

\begin{proof}[Proof of Lemma~\ref{lem:absorb}.]
We will make many subsequent random guesses about a fixed hypothetical
optimal solution $Z$
and, often implicitly, proceed forward with the assumption that all guesses so far were correct.

Hereinafter, we assume that the instance $\inst=(G,s,t,\pairs,\bundles,\weight,k,W)$ and the $st$-maxflow $\witnessflow$ 
are at hand, and $Z$ is an optimal solution to $\inst$ such that $Z_{s,t}:=\corecutG{Z}{G}$ is an $st$-mincut of $G$, and 
$\witnessflow$ is a witnessing flow for $Z.$ 

Let  $\lambda:=\lambda_G(s,t),$ and $\karcout:=\karc-\lambda.$
Note that we have $\lambda=\abs{Z_{s,t}} = \abs{\witnessflow}$
and $\karcout=\abs{Z\setminus Z_{s,t}}.$ 
The assumptions of the lemma imply that $\karcout,\lambda > 0$. 

A vertex $u$ is \emph{active} in $Z$ if it is in the $s$-side of $Z_{s,t}$ and $t$-side of $Z.$ 
We say that a clause $\{u,v\}$ is \emph{active} in $Z$ if it is violated by $Z_{s,t}$ and satisfied by $Z$. 

\begin{claim}\label{claim:activeclause}
There exists a clause $\{u,v\}$ active in $Z.$ 
Moreover, in an active clause $\{u,v\}$ at least one of $u$ and $v$ is active in $Z$ and is contained in $V(G)\setminus V(\witnessflow)$. 
\end{claim}
\begin{proofofclaim}
Observe that $\karcout>0$ means that $Z\setminus Z_{s,t}$ contains at least 
one arc, say $e$. If there is no active clause in $Z$, then $Z-e$ is a solution to $\inst$ as 
$Z-e$ is an $st$-cut and $Z-e$ satisfies the same set of bundles satisfied by $Z$, and possibly more. Therefore, there exists 
an active clause $\{u,v\}$ in $Z$. 
Recall that a clause is satisfied if at least one of its vertices is not reachable from $s$.
Hence, 
at least one of $u$ and $v$ must be active in $Z$ and such a vertex does not lie on a flow path; 
indeed if $v$ is active in $Z$ and lies on some $P_i$, then $v$ is in the $s$-side of $Z_{s,t}\cap E(P_i)$. 
As $\witnessflow$ is a witnessing flow for $Z$, $Z_{s,t} \cap E(P_i) = Z \cap E(P_i)$, $\abs{Z_{s,t}\cap E(P_i)}=1$ and $v$ is still reachable from $s$ in $G-Z.$ 
This contradicts the assumption that $v$ is active in $Z$ and thus in the $t$-side of $Z.$
\end{proofofclaim}

For a vertex $v$ and a path $P_i \in \witnessflow$, a path $Q$ is an \emph{attachment path from $P_i$ to $v$} if 
$Q$ leads from a vertex $w \in V(P_i) \setminus \{t\}$ to $v$ and $w$ is the only vertex of $\witnessflow$ on $Q$.
The \emph{projection $\pi_i(v)$ of $v$ onto $P_i$} is the earliest (closest to $s$) vertex on $P_i$
such that there is an attachment path from $P_i$ to $v$ that starts in $\pi_i(v)$.
We denote $\pi_i(v) = t$ if there is no such vertex. 

For a vertex $u$ which is in the $s$-side of $Z_{s,t},$ 
note that there exists a flow path $P_i$ such that the projection of $u$ onto $P_i$ is in the $s$-side of $Z_{s,t}$. 
Such a flow path is said to \emph{activate} $u$ in $Z_{s,t}$. Notice that if
$\pi_i(u)=s$ for some flow path $P_i$, then for every flow path $P_j$ we have $\pi_j(u)=s$ 
and thus every flow path in $\witnessflow$ activates $u.$
For an active clause $\{u,v\}$ there exists a pair of flow paths $P_i,P_j\in \witnessflow$ such that 
$P_i$ activates $u$ and $P_j$ activates $v$ in $Z_{s,t}$. 
We shall often omit to mention $Z$ or $Z_{s,t}$ for a vertex, a clause active in $Z$ 
or a flow path activating a vertex in $Z_{s,t}$.

The next observation is immediate.
\begin{claim}\label{claim:activevtx}
Let $u$ be an active vertex in $Z$
activated by flow path $P_i$. 
Then $Z\setminus Z_{s,t}$ hits every attachment path from $\pi_i(u)$ to $u.$ 
\end{claim}

For a vertex $u$ as in Claim~\ref{claim:activevtx}, there is no path from $u$ to $t$ in $G-Z$ since otherwise 
there is an arc $e$ of $Z\setminus Z_{s,t}$ and a directed path from the head of $e$ to $t$ in $G-Z$ traversing $u$, thus rendering $e$ to be included in $Z_{s,t}(=\corecutG{Z}{G}).$ 
The key idea to elevate $\lambda$ is to add to the current instance such a missing path from an active vertex to $t$, thus absorbing at least one arc $e$ of 
$Z\setminus Z_{s,t}$ in the $\corecutG{Z}{G'}$. Claim~\ref{claim:activeclause} suggests that one can use the active clauses to build such a missing path. 

\begin{claim}\label{claim:almostmonotone}
Let $x_1, x_2,\ldots , x_\ell$ be a sequence of vertices in $V(P_i)\setminus \{t\}$ for some $P_i\in \witnessflow$ such that 
$x_j$ strictly precedes $x_{j+1}$ on $P_i$ for every $j\in [\ell-1]$. 
Then there are at most $\karcout$ indices $b\in [\ell]$ such that 
$\pi_i^{-1}(x_b) = \{v \in V(G)~|~\pi_i(v) = x_b\}$ contains an active vertex.
\end{claim}
\begin{proofofclaim}
Suppose that there is a sequence of vertices $u_0,\ldots , u_{\karcout}$ in $V(G)\setminus V(\witnessflow)$ and indices $1 \leq a_0 < a_1 < \ldots a_{\karcout} \leq \ell$ with $\pi_i(u_j)=x_{a_j}$ 
and let $Q_j$ be an attachment path from $x_{a_j}$ to $u_j$ for each $0 \leq j \leq \karcout.$ 
We first argue that $Q_j$'s are pairwise vertex disjoint. 
If a pair of attachment paths, say, $Q_j$ and $Q_{j'}$ intersect for some $j<j'$ 
then the prefix of $Q_j$ and suffix of $Q_{j'}$ can be combined into a directed path $Q$ from $x_{a_j}$ to $u_{j'}$.
Observe that $Q$ is an attachment path from $P_j$ to $u_{j'}$ starting in $x_{a_j}$, 
contradicting that $\pi_i(u_{j'})=x_{a_{j'}}$.

Now the claim follows because otherwise, $Z\setminus Z_{s,t}$ needs to hit at least $\karcout+1$ attachment paths by Claim~\ref{claim:activevtx}, 
and these paths are pairwise vertex-disjoint due to the observation in the previous paragraph. This leads to a contradiction as 
we deduce $\abs{Z\setminus Z_{s,t}}\geq \karcout+1>\abs{Z\setminus Z_{s,t}}.$ 
\end{proofofclaim}

\begin{claim}\label{claim:detectactive}
We can in polynomial-time sample a sequence of vertices $u_1,\ldots , u_\ell$ and a path $P_i \in \witnessflow$,
such that $\pi_i(u_j) \neq t$ for every $j \in [\ell]$, 
$\pi_i(u_j)$ strictly precedes $\pi_i(u_{j+1})$ on $P_i$ for every $j\in [\ell-1]$, 
and with probability $\Omega(1/(\lambda^2\kclause \karcout))$
\begin{itemize}
\item $\ell > 0$, and
\item there exists $a \in [\ell]$ such that 
\begin{itemize}
\item $u_a$ is active in $Z$, and 
\item for each $b\neq a$, $u_b$ is in the $s$-side of $Z$ if and only if  $b<a$.
\end{itemize}
\end{itemize}
\end{claim}
\begin{proofofclaim}
Note that we are allowed to sometimes return an empty sequence (i.e., $\ell = 0$),
and the algorithm below actually uses this option in moments when it discovers
that some of the random choices made so far were incorrect.

Recall  $\lambda = \abs{\witnessflow}>0$.
We fix an active clause $\{u^*,v^*\}$.

Recall that $\kclause$ is the number of clauses violated by $Z$. 
We sample a subset $\pairs'$ of clauses in $\pairs$; take a clause $p\in  \pairs$ into $\pairs'$
with probability $1/(2\kclause)$.
We have $\{u^*,v^*\} \in \pairs'$ with probability $1/(2\kclause)$ and, independently,
a fixed clause that is not satisfied by $Z$ is not in $\pairs'$ with probability $(1-1/(2\kclause))$.
Since there are $\kclause$ clauses not satisfied by $Z$
(and $(1-1/x)^x = \Omega(1)$ for $x \geq 2$), we have 
\[
{\sf Prob}[\{u^*,v^*\}\in \pairs' \text{ and every clause in $\pairs'$ is satisfied by $Z$}]\geq (1/(2\kclause))\cdot (1-1/(2\kclause))^{\kclause}=\Omega(1/\kclause).
\]

Henceforth, we proceed with the selected clauses $\pairs'$ and assume that all of them are satisfied by $Z$.

For every path $P_i \in \witnessflow$, by $\leq_i$ we denote the relative order of vertices on $P_i$ with $s$ being the minimum element and $t$ being the maximum element.

Second, choose an ordered pair $(i,j)\in [\lambda]\times[\lambda]$ uniformly at random. There are two possibilities.

\smallskip

\noindent {\bf Case 1.} Suppose that that there exists a flow path in $\witnessflow$ which activates $u^*$ and $v^*$ simultaneously.
In particular, this happens if one of $u^*$ and $v^*$ is reachable from $s$ avoiding $V(\witnessflow)$. 
Then the probability that for the chosen pair $(i,j)\in [\lambda]\times [\lambda]$, $i=j$ and $P_i$ activates both $u^*$ and $v^*$ is at least $1/\lambda^2$;
we expect this to happen.
Let
\[ \preim_{i,i} = \{(u,v)~|~\{u,v\} \in \pairs'\wedge\pi_i(u)\neq t\wedge\pi_i(v)\neq t\wedge\pi_i(u) \leq_i \pi_i(v) \}.\]
We expect $\preim_{i,i} \neq \emptyset$ and that $(u^*,v^*)$ or $(v^*,u^*)$ is in $\preim_{i,i}$; by symmetry, assume $(u^*,v^*) \in \preim_{i,i}$.
(Formally, if $\preim_{i,i} = \emptyset$, we return an empty sequence.)

Let $w_1$ be the earliest vertex on $P_i$ such that $w_1=\pi_i(v)$ for some $(u,v) \in \preim_{i,i}$. 
Let $(u_1,v_1) \in \preim_{i,i}$ be an arbitrary clause with $w_1 = \pi_i(v_1)$. 
We claim that $\{u_1,v_1\}$ is an active clause in $Z$.
Indeed, the success event indicates  $(u^*,v^*)\in \preim_{i,i}$ and thus $w_1 \leq_i \pi_i(v^*)$.
It follows that $\pi_i(u_1) \leq_i \pi_i(v_1) = w_1 \leq_1 \pi_i(v^*)$ on $P_i$. As $\pi_i(v^*)$ is in the $s$-side of $Z_{s,t}$, 
both $\pi_i(u_1)$ and $\pi_i(v_1)$ are in the $s$-side of $Z_{s,t}.$ As $(u_1,v_1)\in \preim_{i,i}$ is satisfied by $Z$ in a successful sampling of $\pairs'$, 
at least one of $u_1$ and $v_1$ is in the $t$-side of $Z$ and in particular at least one of them is in $V(G)\setminus V(\witnessflow)$ by Claim~\ref{claim:activeclause}; 
otherwise, we return an empty sequence.
By choosing a random vertex from $\{u_1,v_1\}\setminus V(\witnessflow)$, we output a singleton sequence $u_1$. 
It is clear that $u_1$  meets the conditions  in the statement. The probability of success is $\Omega(\kclause^{-1} \lambda^{-2})$

\medskip

\noindent {\bf Case 2.} Suppose that Case 1 does not apply.
Then, there exist distinct flow paths $P_i, P_j \in \witnessflow$ where $P_i$ activates $u^*$ and $P_j$ activates $v^*$ or vice versa.
The probability that we sample such a pair $(i,j)\in [\lambda]\times [\lambda]$ is at least $2/\lambda^2$ (both pairs $(i,j)$ and $(j,i)$ would work).

We define the set of clauses as ordered pairs as
\[ \preim_{i,j} = \{(u,v)~|~\{u,v\}\in \pairs'\wedge \pi_i(u)\neq s,t\wedge\pi_j(v)\neq s,t\} .\]
Note that at least one of $(u^*,v^*)$ and $(v^*,u^*)$ belongs to $\preim_{i,j}.$ By symmetry, assume $(u^*,v^*)\in \preim_{i,j}.$
We also remark that as Case 1 does not apply, there is no clause $\{u,v\}\in \pairs'$ with $\pi_i(u)=s$ or $\pi_j(v)=s$.

For two ordered pairs $(u_1,v_1),(u_2,v_2) \in \preim_{i,j}$,
we say that $(u_1,v_1)$ \emph{dominates} $(u_2,v_2)$ if 
$\pi_i(u_1) \leq_i \pi_i(u_2)$ and $\pi_j(v_1) \leq_i \pi_j(v_2)$ on $P_j$,
  and \emph{strictly dominated} if at least one of these inequalities is strict. 
Let $\preim''_{i,j}$ be the set of clauses in $\preim_{i,j}$ which are not strictly dominated by any other clause in $\preim_{i,j}$
and let $\preim'_{i,j}$ be a maximal set of clauses of $\preim''_{i,j}$ such that no pair of $\preim'_{i,j}$ dominates
another pair in $\preim'_{i,j}$. That is, for every $(x,y)$ such that there is $\{u,v\} \in \preim''_{i,j}$ with $\pi_i(u)=x$ and
    $\pi_i(v) = y$, we choose exactly one such pair $\{u,v\}$ into $\preim'_{i,j}$ and we choose arbitrarily. 
It is clear that $\preim_{i,j}$, $\preim''_{i,j}$, and $\preim'_{i,j}$ can be computed in polynomial time. 

We argue that there exists an active clause in the set $\preim'_{i,j}.$ Indeed, if $(u^*,v^*) \notin \preim'_{i,j}$ 
then there exists a clause $\{u^o,v^o\}$ with $(u^o,v^o)\in \preim'_{i,j}$ such that $(u^o,v^o)$ dominates $(u^*,v^*)$.
Because $\pi_i(u^*)$ is in the $s$-side of $Z_{s,t}$ and $\pi_i(u^o) \leq_i \pi_i(u^*)$,
$\pi_i(u^o)$ is in the $s$-side of $Z_{s,t}.$ Likewise, $\pi_j(v^o)$ is in the $s$-side of $Z_{s,t}$ on $P_j$. 
Since $Z$ satisfies all clauses of $\pairs'$, $\{u^o,v^o\}$ is active.
Henceforth, we fix an arbitrary active pair contained in $\preim'_{i,j}$ and denote it by $\{u^o,v^o\}.$

By the construction of $\preim'_{i,j}$ as the set of pairwise non-dominating clauses, 
the set of projections of clauses in $\preim'_{i,j}$
\[
\mathcal{D} := \{(\pi_i(u),\pi_j(v))~|~(u,v)\in \preim'_{i,j}\}
\]
forms an antichain; that is, 
$\mathcal{D}$ can be enumerated as $(x'_1,y'_1),\ldots , (x'_{\ell'},y'_{\ell'})$ 
such that $x'_1 <_i x'_2 <_i \ldots <_i x'_{\ell'}$ and
$y'_1 >_j y'_2 >_j \ldots >_j y'_{\ell'}$. See also Figure~\ref{fig:3.2}.
The sequences $(x_b')_{b=1}^{\ell'}$ and $(y_b')_{b=1}^{\ell'}$ can be computed in polynomial time from $\preim'_{i,j}.$
Also there exists some index ${a'}\in [\ell']$ 
such that $x'_{a'}=\pi_i(u^o)$ and $y'_{a'}=\pi_j(v^o).$ 

An index $b\neq a'$ is a \emph{bad} index if there exists $(u,v) \in \preim'_{i,j}$ 
with $(x'_b,y'_b)=(\pi_i(u),\pi_j(v))$ such that at least one of $u$ and $v$ is active in $Z.$ 
Due to Claim~\ref{claim:almostmonotone}, 
there are at most $2\karcout$ bad indices in $[\ell'].$ Indeed, if there are 
more than  $2\karcout$ bad indices, then w.l.o.g.
there are more than $\karcout$ indices $b$ such that $\pi_i^{-1}(x'_b)$ contains an active vertex,
contradicting Claim~\ref{claim:almostmonotone}.

Now, sample   
each index of $[\ell']$ with probability $1/(2\karcout)$ and let $J\subseteq [\ell']$ be the output indices. 
From the observation in the previous paragraph, 
\[{\sf Prob}[\text{$a' \in J$ and there is no bad index in $J$}]
\geq \frac{1}{2\karcout}\cdot \left(1-\frac{1}{2\karcout}\right)^{2\karcout}=\Omega(1/{2\karcout}). 
\]

Assume that the event above happens. 
Let $\ell:=\abs{J}$ and we rewrite the sequence in $\{x'_b~\mid~b\in J\}$ as $x_1,\ldots , x_\ell$
so that  $x_1 <_i x_2 <_i \ldots <_i x_{\ell}$  on $P_i$. 
Likewise $\{y_b~\mid~b\in J\}$ is rewritten as $y_1,\ldots , y_\ell$
so that $y_\ell <_j \ldots <_j y_2 <_j y_1$ on $P_j$. 
Let $a\in [\ell]$ be the index such that $x_a=\pi_i(u^o)$ (thus $y_a=\pi_j(v^o)$). 

Let $\{u_b,v_b\} \in \preim'_{i,j}$ be the clause with $(x_b,y_b)=(\pi_i(u_b),\pi_j(v_b))$ for each $b\in [\ell]$.
Thanks to the assumption on the sampled set $J$, namely that there is no bad index in $J$, 
for every $b \neq a$, neither $u_b$ nor $v_b$ is active in $Z$, and hence
the clause $\{u_b, v_b\}$ is satisfied by $Z_{s,t}$. 

We claim that the sequence $u_b$ is in the $s$-side of $Z$ if and only if $b<a.$ 
Consider $u_b$ with $b<a.$ As $x_b=\pi_i(u_b)$ strictly precedes $x_a=\pi_i(u_a)$ on $P_i$, 
$x_a=\pi_i(u^o)$ is in the $s$-side of $Z_{s,t}$ and $u_b$ is {\sl not active} in $Z$,  
it is deduced that $u_b$ is in the $s$-side of $Z.$ On the other side, for any $b>a$, the vertex $y_b=\pi_j(v_b)$ strictly precedes $y_a=\pi_j(v^o)$ on $P_j$.
As $y_a=\pi_j(v^o)$ is in the $s$-side of $Z_{s,t}$ and $v_b$ is not active in $Z$, this implies that $v_b$ is in the $s$-side of $Z$. 
The assumption that $b$ is not a bad index means that $\{u_b,v_b\}$ is satisfied by $Z$ while $u_b$ is not active in $Z$, that is, $u_b$ is  in the $t$-side of $Z.$ 
Hence the claim holds.

Finally, we claim that $u_a$ is an active vertex with probability at least $1/4.$ Indeed, notice that for every clause $\{u,v\}\in \preim''_{i,j}$ with $\pi_i(u)=x_a$ and $\pi_j(v)=y_a$,
  the clause $\{u,v\}$ is satisfied by $Z.$ Furthermore, as $x_a=\pi_i(u^o)$ and $y_a=\pi_j(v^o)$ and $\{u^o,v^o\}$ is active in $Z,$ and thus in the $s$-side of $Z_{s,t}$,
 any such clauses is active in $Z.$ In particular, for each such $\{u,v\}$ at least one of $u$ and $v$ is active. 
If for at least half of such pairs $\{u,v\}$, the vertex $u$ is active, then the random choice of $\{u_a,v_a\}$ gives active $u_a$ with probability $1/2$.
Otherwise, for at least half of such pairs $\{u,v\}$, the vertex $v$ is active;
note that in the probability space there is a symmetric event where all the choices were made the same, but we considered pair $(j,i)$ instead of $(i,j)$.
Hence, overall $u_a$ is an active vertex with probability at least $1/4$.

All in all, with probability $\Omega(\lambda^{-2} \kclause^{-1} \karcout^{-1})$ we output a desired sequence $u_1,\ldots , u_\ell$.
\end{proofofclaim}

\begin{figure}[tb]
\begin{center}
\includegraphics{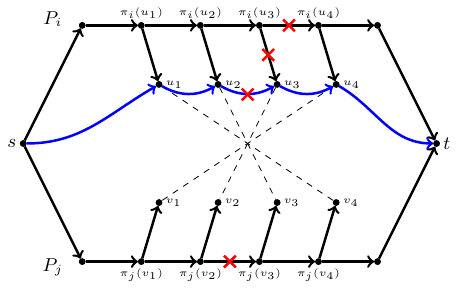}
\caption{An exemplary situation as in Case 2 of Claim~\ref{claim:detectactive} and
 its later usage. 
The dashed lines depict (crisp) clauses, whose projections on $P_i$ and $P_j$
form an antichain. Red crosses depict a solution $Z$. 
The clause $\{u_3,v_2\}$ is an active clause, with $u_3$ being an active vertex. 
Adding the blue arcs to the graph increases the size of a solution by one 
(the arc $(u_2,u_3)$ needs to be cut as well), but moves the arc $(\pi_i(u_3),u_3)$
from $Z \setminus Z_{s,t}$ to $Z_{s,t}$, decreasing $\abs{Z \setminus Z_{s,t}}$.}\label{fig:3.2}
\end{center}
\end{figure}

Claim~\ref{claim:detectactive} immediately turns into a procedure to lift the value $\lambda$ so as to match $\abs{Z}$ eventually, allowing us to wrap up the proof of Lemma~\ref{lem:absorb}.

First, note that there exists an active clause in $Z$ by Claim~\ref{claim:activeclause}.
Then the sampling of Claim~\ref{claim:detectactive} 
is applicable and the output sequence $u_1,\ldots , u_\ell$ on some $P_i\in \witnessflow$ satisfies the following with probability $\Omega(\lambda^{-2} \kclause^{-1} \karcout^{-1})$, 
where $\lambda:=\lambda_G(s,t)$ and $\karcout:=\karc-\lambda.$
\begin{itemize}
\item $u_a$ is active in $Z$ for some $a\in [\ell]$, and 
\item for each $b\neq a$, $u_b$ is in the $s$-side of $Z$ if and only if  $b<a$.
\end{itemize}

Now construct the graph $G^+$ by adding $(u_{j-1},u_j)$ for every $j\in [\ell]$ (using $u_0=s$)
as well as $(u_\ell,t)$ as fresh arcs. 
For every arc $(u_{j-1},u_j)$ for $j \in [\ell]$ we create a new bundle of weight $W+1$,
consisting of this arc only. (The arc $(u_\ell,t)$ is not contained in any bundle.)
Then 
for the instance $\inst^+=(G^+,s,t,\pairs,\bundles^+,\weight^+,k+1,2W+1)$ and $\karc^+:=\karc+1$, 
we apply the flow-augmentation of Lemma~\ref{lem:mincutcore}; let $A\subseteq V(G^+)\times V(G^+)$ and $\witnessflow^+$ be the output. 
We return the instance $\inst^+=(G^++A, s,t, \pairs, \bundles^+, \weight^+, k+1,2W+1)$, the $st$-maxflow $\witnessflow^+$ of $G^++A$ and $\karc^+$. 

The first two statements of Lemma~\ref{lem:absorb} are clear from the construction. To see the last statement, 
consider $Z^+\subseteq E(G^+)$, where $Z^+:=Z\cup \{(u_{a-1},u_a)\}$ (recall $u_0=s$).
We argue that $Z^+$ is an optimal solution to $\inst^+$ with $\abs{Z^+}-\corecutG{Z^+}{G^+} < \abs{Z}-\lambda_G(s,t)$. 
Indeed, it is an optimal solution because any solution to $\inst^+$ with less weight must hit the new path $s,u_1,\ldots , u_\ell,t$ which costs $W+1$ 
and the remaining part is a solution to $\inst$ with lesser weight than $Z$, a contradiction. 

To verify $\abs{Z^+}-\corecutG{Z^+}{G^+} < \abs{Z}-\lambda_G(s,t)$, observe that by Claim~\ref{claim:activevtx} 
$Z\setminus Z_{s,t}$ hits all attachment paths from $\pi_i(u_a)$ to $u_a$, and there is at least one such path. Moreover, at least 
one such arc, in particular the arc of $Z\setminus Z_{s,t}$ which is closest to $u_a$ on some attachment path 
is included in $\corecutG{Z^+}{G^+}$. Note that $(u_{a-1},u_a)$ is also included in $\corecutG{Z^+}{G^+}.$ The inequality follows. 

After performing a flow-augmentation on $\inst^+$ with $\karc^+$, we have $\corecutG{Z^+}{G^+}=\lambda_{G^+}(s,t)$ 
with the claimed probability. This completes the proof of Lemma~\ref{lem:absorb}.
\end{proof}


\subsection{Solving a mincut instance}\label{ss:ID2-mincut}

Lemma~\ref{lem:2mincut} allows us to focus on the following case. 
We are given a $2K_2$-free \gdpcshort{} instance $\inst = (G,s,t,\pairs,\bundles,\weight,k,W)$ together
with a witnessing flow $\witnessflow$. Our goal is to find any solution to $\inst$ if there exists a solution
$Z$ that is an $st$-mincut.

Formally, in this section we prove Theorem~\ref{thm:ID2-mincut}.

\algidmincut*

Before we proceed, let us check that Theorem~\ref{thm:ID2-mincut} indeed concludes the proof
of Theorem~\ref{thm:alg-gdpc} (and thus also of Theorem~\ref{thm:alg-ID2}).
\begin{proof}[Proof of Theorem~\ref{thm:alg-gdpc}.]
Given an instance $\inst=(G,s,t,\pairs,\bundles,\weight,k,W)$ of \gdpcshort{},
we first apply the algorithm of Lemma~\ref{lem:2mincut},
obtaining an instance $\inst' = (G',s,t,\pairs,\bundles',\weight',k',W')$
such that if $\inst$ is a \yes-instance, then (with good probability) 
$\inst'$ is a \yes-instance that admits a solution that is an $st$-mincut of $G'$. 
Then, Theorem~\ref{thm:ID2-mincut} solves correctly such an instance with good probability. 
\end{proof}

For Theorem~\ref{thm:ID2-mincut},
we provide an elaborate guessing scheme whose aim is to get rid of all clauses from the input instance $\inst$.
Once the instance does not contain any clauses,
we argue in Section~\ref{ss:ID2-no-clauses} that the instance reduces
to a \textsc{Weighted Bundles Cut} instance with the property called \emph{pairwise linked deletable edges} that
was shown to be FPT in the first paper of this series, namely~\cite{dfl-arxiv}.

We assume $\lambda \leq k \maxarity^2$, as otherwise we can safely return a negative answer.
Let $\witnessflow$ be any $st$-maxflow of $G$; note that any $st$-maxflow is a witnessing flow
for a solution $Z$ that is an $st$-mincut.

We say that an arc $e=(u,v)$ is \emph{undeletable} if it is crisp or there is a crisp parallel arc $(u,v)$, and \emph{deletable} otherwise. 
Note that any solution $Z$ that is an $st$-mincut contains exactly one deletable edge per flow path in $\witnessflow$ and nothing more.
This allows the following cleanup step: for every deletable arc $e$ that is not on $\witnessflow$, we add a crisp copy of $e$, making it undeletable.

We perform the following randomized step.
For every $P_i \in \witnessflow$ and for every bundle $B \in \bundles$, if $B$ contains at least one deletable edge of $P_i$, randomly guess either $\bot$ or one deletable edge
of $E(P_i) \cap B$. Add a crisp copy of every edge of $E(P_i) \cap B$ that has not been guessed. 
In this guessing step, we aim that for every $B \in \bundles$ violated by $Z$ precisely the edges of $Z \cap B$ have been guessed and no others; i.e., if $Z \cap B \cap E(P_i) = \emptyset$, then we guess $\bot$ for $P_i$ in $B$, and otherwise $Z \cap B \cap E(P_i) = \{e\}$ for some edge $e$ and we guess $e$ for $P_i$ in $B$. 
Note that the success probability of this step is at least $\maxarity^{-\lambda k} = 2^{-\Oh(k^2 \maxarity \log \maxarity)}$. 
As a result, every bundle $B$ contains at most one deletable edge on every flow path $P_i \in \witnessflow$
and we are now looking for a solution $Z$ that is an $st$-mincut such that, for every $B$ violated by $Z$, $Z$ contains all deletable edges of $B$. We call such a solution $Z$ \emph{proper}.

\subsubsection{Operations and cleaning}\label{subsub:clean}
We will need the following operations on $\inst$. By \emph{breaking the bundle $B$} we mean just deleting $B$ from $\bundles$; 
all its clauses and arcs become crisp. Clearly, this operation cannot turn a \no-instance into a \yes-instance and if $Z$ is a proper solution that satisfies $B$,
$Z$ remains a proper solution after breaking $B$.

By \emph{making a deletable edge $e$ undeletable} we mean the operation of breaking
the bundle containing $e$. Again, if $Z$ is a proper solution and $e \notin Z$,
then $Z$ remains a proper solution.

For a crisp edge $(s,v)$, by \emph{contracting $(s,v)$} we mean the operation of identifying
$s$ and $v$, keeping the name $s$ for the resulting vertex, and removing possible multiple
crisp arcs or clauses. If $v \in V(P_i)$ for some $P_i \in \witnessflow$, we shorten $P_i$
to start from the edge succeeding $v$; we make all deletable arcs on $P_i$
between $s$ and $v$ undeletable.
For a crisp edge $(v,t)$, the operation of contracting $(v,t)$ is defined analogously.
Note that these operations do not break the assumption of $2K_2$-freeness of bundles
and any proper solution $Z$ remains a proper solution.

For a path $P_i \in \witnessflow$ and $v \in V(P_i) \setminus \{t\}$, by \emph{contracting $P_i$ up to $v$ onto $s$} we mean the following operation: we break all bundles that 
contain deletable edges of $P_i$ that lie between $s$ and $v$ and then exhaustively
contract crisp edges with tails in $s$. Note that, as a result, the entire subpath of $P_i$
from $s$ to $v$ will be contracted onto $s$. 
The operation of \emph{contracting $P_i$ from $v$ onto $t$} is defined analogously.
We observe that for a proper solution $Z$, if the unique edge of $Z \cap E(P_i)$ is after $v$,
then contracting $P_i$ up to $v$ onto $s$ retains $Z$ as a proper solution,
while if the unique edge of $Z \cap E(P_i)$ is before $v$
then contracting $P_i$ from $v$ onto $t$ retains $Z$ as a proper solution.

For a bundle $B$, by \emph{deleting $B$} we mean the following operation.
We delete all arcs and clauses of $B$, decrease $k$ by one, and decrease $W$
by the weight of $B$. Recall that for any undeletable arc or clause removed this way,
a crisp parallel copy remains in the instance.
Furthermore, for every path $P_i$ that contains a deletable arc of $B$, say an arc $(u,v)$,
we contract $P_i$ up to $u$ onto $s$, contract $P_i$ from $v$ onto $t$, and delete $P_i$
from $\witnessflow.$ 
Let $\inst'$ be the resulting instance and $\witnessflow'$ be a resulting flow.
Clearly, if $Z$ is a proper solution that violates $B$, then $Z \setminus B$
remains a proper solution in the resulting pair $(\inst',\witnessflow')$. 
In the other direction, if $Z'$ is a solution to $\inst'$, then
$Z'$, together will all deletable arcs of $B$, is a solution to $\inst$. 

A \emph{cleanup} consists of exhaustively performing the following steps:
\begin{enumerate}
\item Contract any crisp arc $(s,v)$ onto $s$ or $(v,t)$ onto $t$.
\item Delete any vertex not reachable from $s$ in $G-\{t\}$.
\item Delete any clause containing $t$, an arc with a tail in $t$, or an arc with a head in $s$. 
\item If $k < 0$, $W < 0$, or 
$\inst$ contains a crisp clause $\{s,s\}$ or a crisp arc $(s,t)$, return a negative answer.
\item If $\witnessflow = \emptyset$, then return a positive answer
if $Z = \emptyset$ is a solution and a negative answer otherwise.
\item If $\inst$ contains a soft clause $\{s,s\}$ or a soft arc $(s,t)$, delete
the bundle containing this clause/arc.
\end{enumerate}

We start by performing a cleanup operation on the given instance.

\subsubsection{Projections}

Since every arc not contained in $\witnessflow$ is assumed to be undeletable,
we can compute a \emph{projection} of the clauses and crisp paths of the instance
down to $V(\witnessflow)$. Let us define carefully how this is done.

Recall the notion of a projection, used also in the previous section.
For a vertex $v$ and a path $P_i \in \witnessflow$, let $\pi_i(v)$ be the earliest (closest
to $s$) vertex on $P_i$ such that there is a path from $\pi_i(v)$ to $v$
whose only vertex on $\witnessflow$ is $\pi_i(v)$
(in particular, such a path consists of crisp arcs only); 
we denote $\pi_i(v) = t$ if there is no such vertex.
We often refer to $\pi_i(v)$ as \emph{the projection of $v$ onto $P_i$}. 
For $1 \leq i,j \leq \lambda$, the projection of a pair $(u,v) \in V(G) \times V(G)$ onto $(P_i,P_j)$
is the ordered pair $(\pi_i(u),\pi_j(v))$.
Note that if $\{u,v\}$ is a clause
and an $st$-mincut $Z$ violates $\{\pi_i(u), \pi_j(v)\}$ treated as a clause,
then it also violates the clause $\{u,v\}$.

For $1 \leq i \leq \lambda$, we define
\[ \pairs_{i,i} = \{(\pi_i(u),\pi_i(v))~|~\{u,v\} \in \pairs\}
\cap \left((V(G) \setminus \{t\}) \times (V(G) \setminus \{t\})\right). \]
(That is, we exclude pairs that contain the vertex $t$, but keep the ones containing $s$.)
For $1 \leq i \neq j \leq \lambda$, we define
\[ \pairs_{i,j} = \{(\pi_i(u),\pi_j(v)), (\pi_i(v),\pi_j(u))~|~\{u,v\} \in \pairs\}
\cap \left((V(G) \setminus \{s,t\}) \times (V(G) \setminus \{s,t\})\right). \]
(That is, this time we exclude pairs that contain either $s$ or $t$.)
Note that the notion of $\pairs_{i,i}$ and $\pairs_{i,j}$ are here a bit different than the ones used inside the proof of Claim~\ref{claim:detectactive}.

Fix $1 \leq i,j \leq \lambda$.
We say that an $st$-mincut $Z$ \emph{satisfies} a pair $(x,y) \in \pairs_{i,j}$
if it satisfies $\{x,y\}$ treated as a clause (equivalently, cuts $P_i$ before $x$
    or cuts $P_j$ before $y$) and \emph{violates} $(x,y)$ otherwise
(equivalently, cuts $P_i$ after $x$ and cuts $P_j$ after $y$).
For an element $(x,y) \in \pairs_{i,j}$, the \emph{preimage}
of $(x,y)$ is the set of all clauses $\{u,v\}$ such that 
$(x,y) \in \{(\pi_i(u),\pi_j(v)), (\pi_i(v),\pi_j(u))\}$. 
An element $(x,y)$ is \emph{crisp} if its preimage contains a crisp clause, and \emph{soft}
otherwise. 
Observe that $\pairs_{i,j}$ and $\pairs_{j,i}$ consist of the same pairs,
but reversed, that is, $\pairs_{j,i} = \{(y,x)~|~(x,y) \in \pairs_{i,j}\}$.

Next, for $1 \leq i \neq j \leq \lambda$, we define
$E_{i,j}$ to be the set of all pairs $(x,y)$ such that $x \in V(P_i) \setminus \{s,t\}$, $y \in V(P_j) \setminus \{s,t\}$,
and $G$ contains a path $Q$ from $x$ to $y$ such that $V(\witnessflow) \cap V(Q) = \{x,y\}$, $V(P_i) \cap V(Q) = \{x\}$ and $V(P_j) \cap V(Q) = \{y\}$.
Observe that if $x \in V(P_i) \cap V(P_j) \setminus \{s,t\}$, then $(x,x) \in E_{i,j}$, as witnessed by a zero-length path.
Let $Z'$ be a set consisting of exactly one deletable edge on each flow path of $\witnessflow$.
We say that $Z'$ \emph{violates} a pair $(x,y)\in E_{i,j}$ if $Z'$ contains an edge of $P_i$ after $x$
and an edge of $P_j$ before $y$, and \emph{satisfies} $(x,y)$ otherwise. 
Note that an $st$-mincut cannot violate any pair of $E_{i,j}$.

For $i \in [\lambda]$, define an order $\leq_i$ on $V(P_i)$
as $u \leq_i v$ if and only if $u$ comes before $v$ on $P_i$, i.e., $u$ is closer to $s$.
For two pairs $(x_1,y_1),(x_2,y_2)\in \pairs_{i,j}$,
we say that $(x_1,y_1)$ \emph{dominates} $(x_2,y_2)$ if 
$x_1 \leq_i x_2$ and $y_1 \leq_j y_2$. 
A  pair $(x,y) \in \pairs_{i,j}$ is \emph{minimal} if it is not dominated by any other pair
in $\pairs_{i,j}$, and 
the set of all minimal elements of $\pairs_{i,j}$ will be written as $\pairs'_{i,j}.$ 

\begin{claim}\label{claim:clausedom}
Let $Z'$ be an $st$-mincut.
Let $p_1,p_2 \in \pairs_{i,j}$. If $p_1$ dominates $p_2$ and $p_1$ is satisfied by
$Z'$, then $p_2$ is satisfied by $Z'$.
\end{claim}

Similarly, given two pairs $(x_1,y_1),(x_2,y_2)\in E_{i,j}$, the arc $(x_1,y_1)$ \emph{dominates} $(x_2,y_2)$ if $x_1 \leq_i x_2$
and $y_2 \leq_j y_1$. (Note the reversed order here!)
A pair of $E_{i,j}$ is \emph{minimal} if it is not dominated by any other pair.
The set of all minimal elements of $E_{i,j}$ will be written as $E'_{i,j}.$ 

\begin{claim}\label{claim:arcdom}
Let $Z'$ be a set consisting of exactly one deletable edge on every flow path of $\witnessflow$.
Let $a_1,a_2\in E_{i,j}$. If $a_1$ dominates $a_2$ and $a_1$ is satisfied by $Z'$, then $a_2$ is satisfied by $Z'$.
\end{claim}

Observe that $\pairs'_{i,j}$ forms an antichain and $E'_{i,j}$ forms a chain, i.e.,
for $(x_1,y_1), (x_2,y_2) \in \pairs'_{i,j}$ we have $x_1 \leq_i x_2$ if and only if $y_2 \leq_j y_1$,
and for $(x_1,y_1), (x_2,y_2) \in E'_{i,j}$ we have $x_1 \leq_i x_2$ if and only if $y_1 \leq_i y_2$
(in both cases with equality only if the pairs are identical).

\subsubsection{Auxiliary labeled graph}

Recall that our goal is to reduce to an instance without clauses. 
So far, we have reduced to the case where the instance consists of the flow $\witnessflow$
and, for every $i,j \in [\lambda]$,
a set $E'_{i,j}$ of crisp edges from $V(P_i)$ to $V(P_j)$
forming a chain, and a set of clauses $\pairs'_{i,j}$ between $V(P_i)$ and $V(P_j)$
forming an antichain.
We are looking for a solution that cuts exactly one edge per flow path
that satisfies the following constraints for every $i,j \in [\lambda]$: 
\begin{itemize}
\item for every $(u,v) \in E'_{i,j}$, $u$ is after the cut arc on $P_i$ or
$v$ is before the cut arc on $P_j$;
\item for every $\{u,v\} \in \pairs'_{i,j}$, $u$ is after the cut arc on $P_i$ 
or $v$ is after the cut arc on $P_j$. 
\end{itemize}

\begin{figure}[t] 
\begin{center}
\includegraphics{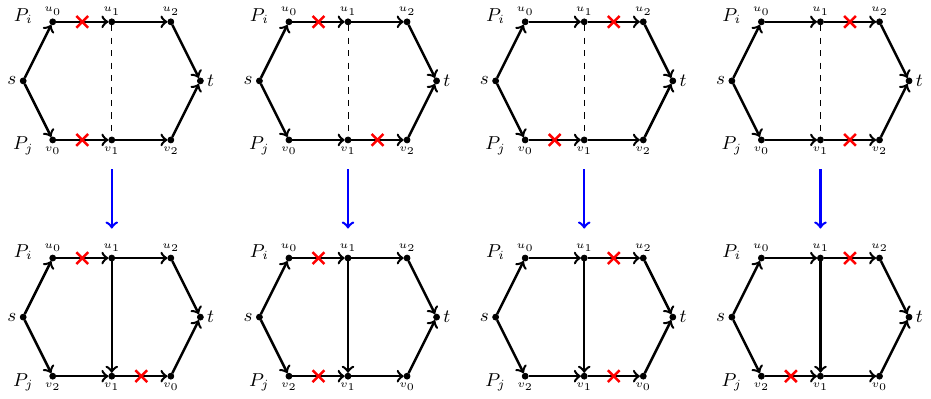}
\caption{Reversing a path $P_j$ and its influence on a pair $\{u_1,v_1\} \in \pairs'_{i,j}$.
To maintain the role of $\{u_1,v_1\}$ as a constraint meaning ``cut $P_i$ before $u_1$
or cut $P_j$ before $v_1$'', after reversing $P_j$ the pair $\{u_1,v_1\}$
should become an arc $(u_1,v_1)$ in $E'_{i,j}$. 
The four columns above correspond to four cases where the paths $P_i$ and $P_j$ are cut 
in the sought solution; in the last column the pair $\{u_1,v_1\}$ 
and the corresponding arc $(u_1,v_1)$ are violated.}\label{fig:3.3}
\end{center}
\end{figure}

We make the following observation; see also Figure~\ref{fig:3.3}.
Fix $j \in [\lambda]$ and \emph{reverse} the path $P_j$.
That is, if the vertices of $P_j$ are $s,v_1,\ldots,v_\ell,t$, replace $P_j$
with a path $s,v_\ell,v_{\ell-1},\ldots,v_1,t$. Then, if we fix a solution $Z$ in the original
graph and, after reversing $P_j$, 
we cut an arc between the same two vertices on $P_j$, the sets of vertices on the $s$-side
and $t$-side on $V(P_j) \setminus \{s,t\}$ swap with each other. 
As a result, if we want to maintain the meaning of $\pairs'_{i,j}$ as constraints,
for every $i \in [\lambda] \setminus \{j\}$, every 
clause in $\pairs'_{i,j}$ should become
an arc with the same endpoints that is put into $E_{i,j}$.
In this manner, we get rid of the clauses of $\pairs'_{i,j}$. 

At the same time, however, during such a reverse operation the arcs of $E_{i,j}$ should become
clauses again while it is not clear what to do with the arcs of $E_{j,i}$, 
so we need to apply the operation carefully. 
We observe that if we reverse both $P_i$ and $P_j$ at the same time,
then any arc in $E_{i,j}$ should change its direction and become an arc in $E_{j,i}$. 

Consequently, to use the above observations
to remove clauses from the current instance, we need to identify
a set $I \subseteq [\lambda]$ of indices of paths that we would like to reverse in such a manner that
(a) every nonempty $\pairs'_{i,j}$ satisfies $|\{i,j\} \cap I| = 1$ and
(b) every nonempty $E'_{i,j}$ satisfies $|\{i,j\} \cap I| \neq 1$. 
This, of course, is not always possible; we develop a branching strategy that aims
at removing obstacles towards the existence of such a set $I$.

\newcommand{\HedgeL}{\xi}
\newcommand{\HvtxL}{\Xi}

We observe that the existence of such a set $I$ corresponds to labeling 
$[\lambda]$ with the elements of the two-element field so that $\pairs'_{i,j} \neq \emptyset$ 
implies that the labels of $i$ and $j$ are distinct 
while $E'_{i,j} \neq \emptyset$ implies that the labels of $i$ and $j$ are equal. 
More formally, let $\mathbb Z_2=\{0,1\}$ where $1 + 1=0$. 
We construct a mixed graph $H_\inst$ (i.e., a graph containing undirected and directed edges)
on the vertex set $[\lambda]$ with a labeling function $\HedgeL \colon E(H)\rightarrow \mathbb Z_2$ as follows. 
For every $1 \leq i,j \leq \lambda$, 
    \begin{itemize}
    \item if $i \neq j$ and $E'_{i,j} \neq \emptyset$, then we add an arc $(i,j)$ to $E(H_\inst)$ with label $0$;
    \item if $\pairs'_{i,j} \neq \emptyset$ (equivalently, $\pairs'_{j,i} \neq \emptyset$),
    then we add an undirected edge $ij$ to $E(H_\inst)$ with label $1$.
    \end{itemize}
We remark that $E'_{i,j} \neq \emptyset$ if and only if $E_{i,j} \neq \emptyset$
and similarly $\pairs'_{i,j} \neq \emptyset$ if and only if $\pairs_{i,j} \neq \emptyset$.

Note that there may be a loop at a vertex $i \in V(H_\inst)$ with label $1$ if $\pairs_{i,i} \neq \emptyset$,
     and for $1 \leq i \neq j \leq \lambda$, there can be both an arc $(i,j)$ with label $0$,
     an arc $(j,i)$ with label $0$, and
     an edge $ij$ with label $1$.
     
Observe that if $V(P_i) \cap V(P_j) \setminus \{s,t\} \neq \emptyset$, then 
there is both an arc $(i,j)$ and an arc $(j,i)$ (both with labels $0$).
Note also that any clause of the form $(s,v)$ with $v\notin \{s,t\}$ imposes a loop to some vertex of $H_\inst$, as due to the cleanup operation $v$ is reachable from $s$
and $\pi_i(v) \neq t$ for some $i$. 

For a cycle $C$ of the underlying undirected multigraph of
$H$ (i.e., ignoring edge orientations), we define $\Lambda(C)=\sum_{e\in C} \HedgeL(e)$. 
Notice that the output element $\Lambda(C)$ does not depend on the order of group operation 
because $\mathbb Z_2$ is an abelian group.
A \emph{non-zero} cycle $C$ of $H_\inst$ is a cycle in
the underlying undirected graph
with the value $\Lambda(C)$ not equal the identity element (0) of $\mathbb Z_2$. 
A labeling function $\HvtxL \colon V(H_\inst)\rightarrow \mathbb Z_2$ is said to be \emph{consistent} 
if for every edge $e=(u,v)$ of $H_\inst$, it holds that $\HvtxL(v)=\HvtxL(u) + \HedgeL(e).$
It is known, for example~\cite{GUILLEMOT201161}, that 
$H_\inst$ admits a consistent labeling if and only if it does not contain a non-zero cycle. 

Going back to our intuition, our goal is to make $H_\inst$ admit a consistent labeling
by designing branching steps on non-zero cycles. When $H_\inst$ admits a consistent labeling,
we set $I$ to be the set of indices $i \in [\lambda]$ with label $1$ and reverse
the paths $P_i$ for $i \in I$.

We start with the following simple observation that ensures that $H_{\inst}$
does not become significantly richer as we branch.

\begin{claim}\label{claim:H-does-not-grow}
Assume we perform an operation on $\inst$ of breaking a bundle, contracting
a crisp $(v,t)$ edge, or contracting a path $P_i$ from $v$ onto $t$, obtaining an instance $\inst'$,
then $H_{\inst'}$ is a subgraph of $H_{\inst}$.
If we additionally perform an operation of contracting a crisp $(s,v)$ edge,
contracting a path $P_i$ up to $v$ onto $s$, or a cleanup operation that does not result
in deleting a bundle,
then $H_{\inst'}$ is a subgraph of $H_{\inst}$ with possibly some new loops with label $1$
added.
\end{claim}
\begin{proofofclaim}
Breaking a bundle turns a number of deletable edges of $\witnessflow$ into crisp edges,
which does not change the sets $E_{i,j}$ or $\pairs_{i,j}$.
If $G'$ is a result of contracting a crisp edge $(v,t)$, then any path or clause in $G'-\{t\}$ exists also in $G-\{t\}$.
Hence, the operation of contracting a crisp edge $(v,t)$ does not add any new element to neither of the sets
$E_{i,j}$ or $\pairs_{i,j}$.
The operation of contracting a path $P_i$ from $v$ onto $t$ is a sequence of operations
of breaking a bundle and contracting a crisp edge with a head in $t$.
This completes the proof of the first part of the claim.

If $G'$ is a result of contracting a crisp edge $(s,v)$, then any path or clause in $G'-\{s,t\}$ exists also in $G-\{s,t\}$.
Hence, the operation of contracting a crisp edge $(s,v)$ does not add any new element to neither of the sets
$E_{i,j}$ or $\pairs_{i,j}$ for $i \neq j$. 
The operation of contracting a path $P_i$ up to $v$ onto $s$ is a sequence of operations
of breaking a bundle and contracting a crisp edge with a tail in $s$.

We are left with the last claim about a cleanup operation that does not delete a bundle. 
To this end, let us go over steps that a cleanup operation can do that are not yet discussed.

If $v$ is not reachable from $s$ in $G-\{t\}$, then $v$ is also not reachable in $G-\{t\}$ from any vertex
of $V(\witnessflow) \setminus \{t\}$. Hence, deleting such $v$
does not change neither of the sets $E_{i,j}$ or $\pairs_{i,j}$.
Also, no edge with a tail in $t$ or head in $s$ participates in a path that is used in the definitions of projections
or $E_{i,j}$. Hence, deleting such an edge does not change either of the sets $E_{i,j}$ or $\pairs_{i,j}$.
Finally, since $\pi_i(t) = t$ for every $i \in [\lambda]$, a clause that involves $t$ is not taken into account
in any set $\pairs_{i,j}$ and can be deleted. This completes the proof of the claim.
\end{proofofclaim}

\subsubsection{Branching on a non-zero cycle}\label{subsub:nonzero}

In this section,
when the  instance $\inst$ in reference is clear in the context,
we omit the subscript in $H_\inst.$

Our goal is to provide a guessing step that, given a non-zero cycle $C$ of 
(the underlying undirected multigraph of) $H$ that
is not a single label-1 loop,
randomly chooses from a small number of subcases, in each case either deleting at least
one bundle (thus decreasing $k$ by at least one) or, while keeping $k$ the same,
decreasing the number of non-loop edges of $H$.
We remark that such a cycle $C$ can be found using standard methods
(cf.~\cite{GUILLEMOT201161}).
At the end we will deal with label-1 loops.

If $H$ has no non-zero cycle, we proceed directly to Section~\ref{sss:ID2-reversal}.
Otherwise, the algorithm performs the said guessing and recurses,
usually starting with a cleanup operation from the beginning of this subsection. 
Thus, the depth of the recursion is at most $k|E(H)| \leq 3k\lambda^2$.

Let $J=j_1,e_1,\ldots , e_{d-1},j_d$ be a path in (the underlying graph of) $H$. An ordered pair $p=(x,y)$ of vertices in $G$ 
is said to have \emph{type $e_i$} if $x$ is on the flow path $P_{j_i}$, $y$ on $P_{j_{i+1}}$ and 
$p$ is an element of $E_{j_i,j_{i+1}}\cup E_{j_{i+1},j_i}$ (respectively of $\pairs_{j_i,j_{i+1}}$) whenever $e_i$ is labeled by 0 
(respectively, 1); 
when $e_i$ is labeled by 0, the orientation of $p$ 
matches the orientation of the edge $e_i$ in $H$.
We also called the pair $p$an \emph{edge pair} if $e_i$ has label $0$
and a \emph{clause pair} if $e_i$ has label $1$.
Furthermore, a pair is \emph{minimal} if it is minimal in its corresponding set $E_{j_i,j_{i+1}}$, $E_{j_{i+1},j_i}$, or $\pairs_{j_i,j_{i+1}}$. 

 A set $f=\{p_1,\ldots , p_{d-1}\}$ of ordered vertex pairs $p_i=(x_i,y_{i+1})$ of $G$ for $i=1,\ldots , d-1$  is 
called a \emph{connection of type $J$} if 
\begin{itemize}
\item for each $i\in [d-1],$ the pair $(x_i,y_{i+1})$ is minimal and has type $e_i$.
\item for each $2\leq i \leq d-1,$ there is no minimal pair $(x'_i,y'_{i+1})$ of type $e_i$ such that $x'_i$ is between $x_i$ (exclusive) and $y_i$ (inclusive) on the flow path $P_i$;
furthermore, if $x_i$ is strictly before $y_i$ on $P_i$, then we require that
there is no minimal pair $(x'_i,y'_{i+1})$ of type $e_i$ such that $x'_i$ is not earlier
than $y_i$ on $P_i$.
\end{itemize} 
In other words, as $(x_i,y_{i+1})$ we take a minimal pair of type $e_i$ where $x_i$
is the first possible vertex on $P_i$ not earlier than $y_i$ and, if no such pair exists,
we ask for $x_i$ being the last possible vertex on $P_i$ earlier than $y_i$.
Note that the choice of $p_1$ determines the choices of the remaining 
pairs in a connection of type $J$.

We say that a path $J$ of $H$ is \emph{even} (respectively, \emph{odd})
if it contains an even (respectively, odd) number of edges labeled by $1$. 
A connection $f$ of type $J$ is even (resp. odd) if $J$ is even (resp. odd). 
We denote $x_1$ as $\vtx_1(f)$ and $y_d$ as $\vtx_d(f).$
Observe that a non-zero cycle in $H$ corresponds to an odd connection back to the same vertex.

Assume that we have a non-zero cycle $C=j_1,e_1,\ldots , e_{d-1},j_d,e_d,j_1$ 
in $H$ consisting of at least two vertices. Because it is non-zero, at least one edge has label $1$ and without loss of generality 
assume that $e_d$ has label $1$. Let $J=j_1,e_1,\ldots , e_{d-1},j_d$ be the subpath $C-e_d,$ which is clearly even. 
For brevity, for $i \in [d-1]$, by $\leq_i$ we denote the relative order from $s$ to $t$ of vertices
on the flow path $P_{j_i}$.

Let $F$ be the set of connections of type $J$, that is,
   for every minimal pair $p_1 = (x_1,y_2)$ of type $e_1$ we
 put into $F$ the connection of type $J$ starting with $p_1$.
Clearly, for any distinct connections $f,f' \in F$ 
we have $\vtx_1(f)\neq \vtx_1(f')$.
We order the elements of $F$ according to the $\leq_1$-order of their starting vertices $\vtx_1(f)$.
With the next straightforward claim, we observe that this order also sorts the endpoints $\vtx_d(f)$.

\begin{claim}
Let $f,f'\in F$ be two connections of type $J$. 
Then $\vtx_1(f)\leq_1 \vtx_1(f')$ implies $\vtx_d(f)\leq_d \vtx_d(f').$
\end{claim}

For a connection $f=\{p_1,\ldots , p_{d-1}\}$ in $F$ with $p_i=(x_i,y_{i+1})$ for each $1\leq i\leq d-1$, 
by \emph{branching on $f$} we mean randomly guessing which of the following cases applies according to the relation between $Z$ and the pairs of $f$:

\begin{description}
\item[Clause on the $s$-side.] For some \emph{clause} pair $p_i=\{x_i,y_{i+1}\}$, both $x_i$ and $y_{i+1}$ are in the $s$-side (of $Z$). Then we contract $P_{j_i}$ up to $x_i$ onto $s$
and $P_{j_{i+1}}$ up to $y_{i+1}$ onto $s$,
    perform cleanup and recurse. 
    Note that all bundles in the preimage of $p_i$
    will be deleted, as they will contain a clause $(s,s)$, so
    the budget $k$ is decreased by at least one.
\item[Clause on the $t$-side.] For some \emph{clause} pair $p_i=\{x_i,y_{i+1}\}$, both $x_i$ and $y_{i+1}$ are in the $t$-side (of $Z$). Then we contract $P_{j_i}$ from $x_i$ and $P_{j_{i+1}}$ from $y_{i+1}$ onto $t$, 
  perform cleanup and recurse.
This has the effect of replacing at least one endpoint of any pair in $\pairs_{j_i,j_{i+1}}$ by $t$, and emptying $\pairs_{j_i,j_{i+1}}$. 
Hence, the number of non-loop edges of $H$ decreases in the recursive call.
\item[Arc $t$-side to $s$-side, forward.] For some pair $p_i=\{x_i,y_{i+1}\}$ whose type $e_i$ is an \emph{arc} from $j_i$ to $j_{i+1}$, the vertex $x_i$ is in the $t$-side and $y_{i+1}$ is in the $s$-side. 
Then contract $P_{j_i}$ from $x_i$ onto $t$ and $P_{j_{i+1}}$ up to $y_{i+1}$ onto $s$.
Notice that all pairs $(x,y) \in E_{j_i,j_{i+1}}$ either have $x$ not earlier than $x_i$ on $P_{j_i}$ or $y$ not later than $y_{i+1}$ on $P_{j_{i+1}}$.
Hence, this operation will empty $E_{j_i,j_{i+1}}$, decreasing the number of non-loop edges
of $H$ in the recursive call.
\item[Arc $t$-side to $s$-side, backward.] For some pair $p_i=\{x_i,y_{i+1}\}$ whose type $e_i$ is an \emph{arc} from $j_{i+1}$ to $j_i$, the vertex $y_{i+1}$ is in the $t$-side and $x_i$ is in the $s$-side. 
Then we proceed analogously as in the previous branching.
\item[Arc $s$-side to $t$-side, impossible.] For some pair $p_i=\{x_i,y_{i+1}\}$ whose type $e_i$ is an \emph{arc}, the tail is in the $s$-side and the head is in the $t$-side. This is impossible, but noted here for completeness of
the case analysis.
\item[Cut between $x_i$ and $y_i$.]
None of the previous five cases apply and $P_{j_i}$ for some $2\leq i\leq d-1$ is the first flow path such that $x_i$ and $y_i$ are not in the same side of $Z$, i.e. 
$x_i$ is in the $s$-side and $y_i$ is in the $t$-side (or vice versa which can be handled similarly). 
Furthermore, since none of the previous five cases apply to the pair $\{x_i,y_{i+1}\}$,
$x_i$ is in the $s$-side of $Z$ and $y_i$ is in the $t$-side of $Z$, 
we can infer in which side $y_{i+1}$ lies: in the same side as $x_i$ if $e_i$ is of label $0$ and in the opposite side 
if $e_i$ is of label $1$. We contract $P_{j_i}$ up to $x_i$ onto $s$ and from $y_i$ onto $t$, contract $P_{j_{i+1}}$ up to $y_{i+1}$ onto $s$
or from $y_{i+1}$ onto $t$, depending on the inferred side of $y_{i+1}$, perform a cleanup, and recurse.

For the analysis of the progress in this case, we perform the following distinction
into subcases.
\begin{description}
\item[Clause $e_i$, $x_i$ before $y_i$.] If $e_i$ is a \emph{clause} and $x_i <_i y_i$, then
  we contract $x_i$ onto $s$ and $y_i$ and $y_{i+1}$ onto $t$ on $P_{j_i}$ and $P_{j_{i+1}}$, respectively.
  Consider any projection $p=(x,y)=(\pi_{j_i}(u),\pi_{j_{i+1}}(v))$ of a clause $\{u,v\}$ to paths $(P_{j_i}, P_{j_{i+1}})$.
  Note that $p$ is eliminated (i.e., has one endpoint contracted into $t$) if $y_i \leq_i x$ or $y_{i+1} \leq_{i+1} y$.
  Furthermore, if $x \leq_i x_i$ then $y_{i+1} \leq_{i+1} y$ by minimality of $p_i$. Hence a pair $p$ that is not
  eliminated must have $x_i <_i x <_i y_i$ and $y <_{i+1} y_{i+1}$.
  But then $p$ would have been a preferred choice over $p_i$ in the connection.
  Hence no such pair can exist, and the number of non-loop edges of $H$ decreases in the recursive call.
\item[Clause $e_i$, $y_i$ before $x_i$.] If $e_i$ is a \emph{clause} and $y_i <_i x_i$, then
  we contract $y_i$ and $y_{i+1}$ onto $s$ and $x_i$ onto $t$. As in the previous case,
  let $p=(x,y)=(\pi_{j_i}(u),\pi_{j_{i+1}}(v))$ be the projection of a clause $\{u,v\}$
  to paths $(P_{j_i}, P_{j_{i+1}})$ and assume neither endpoint of $p$ is contracted into $t$.
  Then $x <_i x_i$, which implies $y >_{i+1} y_{i+1}$.
  Let $p$ be chosen to be such a projection that is minimal
  and as close to $x_i$ as possible on $P_{j_i}$. Then we branch into two options for $p$:
  Either $y$ is in the $s$-side of $Z$ and $p$ and its bundle are deleted;
  or we contract $y$ onto $t$ on $P_{i_{j+1}}$ and by choice of $p$ we get $\pairs_{j_i,j_{i+1}}=\emptyset$,
  decreasing the number of non-loop edges in $H$.
\item[Arc $e_i$, $x_i$ before $y_i$.] If $e_i$ is an \emph{arc} (in either direction) and $x_i <_i y_i$,
  then we contract $x_i$ and $y_{i+1}$ onto $s$ and $y_i$ onto $t$.
  Observe that any crisp arc with head or tail in $s$ or $t$ is discarded or contracted in the cleanup step.
  Assume there is an arc $e$ in $E'_{j_i,j_{i+1}}$ or $E'_{j_{i+1},i}$ (as appropriate,
  depending on the direction of $e_i$) that remains after cleanup. 
  Then its head must lie after $p_i$, which by minimality of $p_i$ means that its tail lies
  after $p_i$ as well (on the respective path). Furthermore $e$ must connect to $P_{j_i}$
  strictly before $y_i$. Thus $e$ connects to $P_{j_i}$ strictly between $x_i$ and $y_i$.
  But then $e$ would have been a preferred choice over $p_i$ in the connection $f$.
  Thus all minimal arcs in the class $E_{j_i,j_{i+1}}$ or $E_{j_{i+1},j_i}$ are removed in cleanup, which implies
  that all arcs in the class are removed in cleanup and $H$ has fewer non-loop edges in the recursion.
\item[Arc $e_i$, $y_i$ before $x_i$.] If $e_i$ is an \emph{arc} (in either direction) and $y_i <_i x_i$,
  then we contract $y_i$ onto $s$ and $x_i$ and $y_{i+1}$ onto $t$.
  Assume as in the previous item that $e$ is an arc of $E'_{j_i,j_{i+1}}$
  respectively $E'_{j_{i+1},j_i}$ (according to direction of $e_i$) 
  which remains after cleanup. Then the tail of $e$ must lie before $p_i$,
  so by minimality of $p_i$ the head of $e$ lies before $p_i$ as well.
  If $e$ connects to $P_{j_i}$ between $x_i$ and $y_i$, then $e$ is a better choice than
  $p_i$ for the connection $f$. Otherwise $e$ has an endpoint contracted into $s$,
  and is either discarded or contracted in cleanup. 
\end{description}
\item[Canonical case.] None of the previous six cases applies. Observe that
\begin{itemize}
\item for every $e_i$ labeled by 1, the endpoints of $p_i$ are in the distinct sides of $Z$, 
\item for every $e_i$ labeled by 0, the endpoints of $p_i$ are in the same side of $Z$, and
\item for every $2\leq i\leq d-1$, $x_i$ and $y_i$ are in the same side of $Z$ (if $d\geq 3$).
\end{itemize}
The assignment on the endpoints of $f$ to $s$ or $t$ is called a \emph{canonical assignment on $f$} if it satisfies these conditions. 
\end{description}
We observe that a branch on $f$ makes a correct guess with probability $\Omega(d^{-1}) = \Omega(\lambda^{-1}) = \Omega((k\maxarity)^{-1})$. 

Note that there are two canonical labelings of $f$, namely the ones which assign $\vtx_1(f)$ to $s$ and $t$ respectively. The next claim is immediate 
from that $J$ contains an even number of edges labeled by 1 and the endpoints of a pair $p_i$ have distinct assignments 
under a canonical assignment if and only if $e_i$ is labeled by 1. 

\begin{claim}\label{claim:connsame}
Let $f\in F$ be a connection of even type $J$.
Then in a canonical assignment on $f$, $\vtx_1(f)$ is assigned $s$
if and only if $\vtx_d(f)$ is assigned $s.$ 
\end{claim}

Now we present the guessing step.
Let $f'\in F$ be the first (leftmost) connection such that there exists at least one pair $(u,v) \in \pairs'_{j_1,j_d}$ 
with $u\leq_1 \vtx_1(f')$ and $v\leq_d \vtx_d(f')$; in case no such $f'$ exists, let $f'$ be the vertex $t$ (with a convention $\vtx_1(t) = \vtx_d(t) = t$).
Let $f$ be the immediate predecessor of $f'$ in $F$ (or the last element of $F$, if $f'=t$, 
 or  $f=s$ if $f'$ is the first element of $F$, with again a convention
 $\vtx_1(s) = \vtx_d(s) = s$) 
and note that $\vtx_1(f) <_1 \vtx_1(f')$ and $\vtx_d(f)\leq_d \vtx_d(f').$ 

First, we branch on $f$ and $f'$ (unless they are equal $s$ or $t$, respectively) and end up in the case where both $f$ and $f'$ have canonical assignments
(we use a convention that $t$ is a canonical assignment of $f'=t$ and $s$ is a cannonical assignment
 of $f=s$). 
If $f'$ has the canonical assignment in which $\vtx_1(f')$ (and $\vtx_d(f')$ by Claim~\ref{claim:connsame}) 
is assigned with $s$, the pair $(u,v)$ is not satisfied; in particular, it must be soft.
Therefore, each bundle in the preimage of $(u,v)$ can be deleted (thus decreasing the budget $k$ by 1 each time).
We guess if this is the case, and if yes then we proceed as above, cleanup, and recurse. 

Henceforth, we consider the case when the canonical assignment on $f'$ which sets $\vtx_1(f')$ and $\vtx_d(f')$ with $t.$ 
The algorithm guesses whether the cannonical assignment of $f$ sets $\vtx_1(f)$ and $\vtx_d(f)$ to $s$ or to $t$.
We consider a few mutually exclusive subcases, that is, for each subcase it is assumed that none of the preceding subcases are applicable.

%
%
\begin{description}
\item[Case A.] $f$ has the canonical assignment in which $\vtx_1(f)$ (and $\vtx_d(f)$ as well) is assigned with $s$. 
Consider the pairs $p_1=(x_1,y_1)$ of $f$ and $p'_1=(x'_1,y'_1)$ of $f'$ and note that there is no other minimal pair 
of type $e_1$ whose endpoints are between $x_1$ and $x'_1$ on $P_{j_1}.$ This also implies that there is no other minimal pair 
of type $e_1$ whose endpoints are between $y_1$ and $y'_1$ on $P_{j_2}.$
Observe that any solution $Z$ that complies with 
the canonical assignment on (both endpoints of) $p_1,p'_1$ satisfy all arcs or clauses of type $e_1.$
Furthermore, contraction according to the canonical assignment
(i.e., we contract $P_{j_1}$ up to $x_1$ onto $s$ and from $x'_1$ onto $t$, and contract $P_{j_2}$ up to / from $y'_1$ and $y_1$, depending on the inferred sides of $y'_1$ and $y_1$)
will eliminate all arcs or clauses of type $e_1$. Hence, we cleanup and recurse,
    either decreasing $k$ (if there was a bundle deleted during the cleanup)
 or decreasing the number of non-loop edges of $H$.

\item[Case B.] There exists  a minimal pair $(u',v') \in \pairs'_{j_1,j_d}$ such that $\vtx_1(f)\leq_1 u'$ and $\vtx_d(f)\leq_d v'.$ 
As Case A is not applicable, $\vtx_1(f)$ and $\vtx_d(f)$ is assigned with $t$ in the canonical assignment; in particular, $f \neq s$.
This means that both $u'$ and $v'$ are in the $t$-side, and thus all pairs
of $\pairs_{j_1,j_d}$ 
are satisfied by the sought solution $Z$. 
We contract $P_{j_1}$ from $u'$ and $P_{j_d}$ from $v'$ onto $t$, cleanup and recurse.
In the recursive call, either the cleanup deleted a bundle, or $\pairs_{j_1,j_d}$ become empty,
   decreasing the number of non-loop edges of $H$.

\item [Case C.] We say that a pair $(u,v)\in \pairs'_{j_1,j_d}$ \emph{crosses} $f$ if 
$u <_1 \vtx_1(f)$ and $\vtx_d(f) <_d v$ (\emph{left-to-right}) or $\vtx_1(f) <_1 u$ and $v <_d\vtx_d(f)$ (\emph{right-to-left}). 
Because Case B is not applicable 
and also by the choice of $f'$ and $f$, every  pair in $\pairs'_{j_1,j_d}$ crosses $f.$ 
There are three subcases.

\begin{description}
\item[C-1.] There exist minimal pairs crossing $f$ both left-to-right and right-to-left. Let $(u_1,v_1)$ 
be the minimal pair crossing $f$ left-to-right with $u_1$ as close to $\vtx_1(f)$ as possible. Likewise, let $(u_2,v_2)$ 
be the minimal pair crossing $f$ right-to-left with $u_2$ as close to $\vtx_1(f)$ as possible. Notice that there is no other minimal pair between 
$P_{j_1}$ and $P_{j_d}$ with an endpoint between $u_1$ and $u_2$ on $P_{j_1}$. Now that 
both $v_1$ and $u_2$ are on the $t$-side on $P_{j_1}$ and $P_{j_d}$ respectively,  every minimal pair between 
$P_{j_1}$ and $P_{j_d}$ has an endpoint in the $t$-side. 
Hence, we can contract $P_{j_1}$ from $u_2$ onto $t$ and $P_{j_d}$ from $v_1$ onto $t$,
cleanup and recurse, again either deleting a bundle in the cleanup phase
or decreasing the number of non-loop edges of $H$ (as $\pairs_{j_1,j_d}$ empties).
\item[C-2.] Every minimal clause crosses $f$ left-to-right. Then, every minimal clause between 
$P_{j_1}$ and $P_{j_d}$ has an endpoint in the $t$-side of $\vtx_1(f)$, thus in the $t$-side of any solution $Z$ 
which conforms the canonical assignment on $f$.
Hence, we can contract $P_{j_1}$ from $\vtx_1(f)$ onto $t$, cleanup, and recurse.
\item[C-3.] Every minimal clause crosses $f$ right-to-left. This subcase is symmetric to the case C-2.
\end{description}
This finishes the description of the guessing step
when $H$ contains a non-zero cycle of length greater than $1$.

%
%
%
\end{description}

Consider now a case when 
$H_\inst$ contains no non-zero cycle of length greater than $1$, but possibly contains
some loops. 

Let $I \subseteq [\lambda]$ be the set of vertices where $H_\inst$ has a loop.
For every $i \in I$,
choose a pair $(x_i,y_i) \in \pairs_{i,i}$ such that the later (on $P_i$) vertex
from $\{x_i,y_i\}$ is as early (close to $s$) as possible.

If $I \neq \emptyset$, guess whether there exists an index $i \in I$
such that both $x_i$ and $y_i$ are in the $s$-side of $Z$.
Note that this is only possible if $(x_i,y_i)$ is soft and then all clauses
in the preimage of $(x_i,y_i)$ are violated by $Z$.
Hence, in this case we can contract $P_i$ up to the later vertex of $x_i$ or $y_i$
onto $s$, cleanup, and recurse; the cleanup phase will delete all bundles in the preimage
of $(x_i,y_i)$. 

In the remaining case, if no such $i \in I$ exists, for every $i \in I$
we contract $P_i$ from the later of the vertices $x_i$ or $y_i$ onto $t$, obtaining an
instance $\inst'$.
Note that this operation deletes all loops from $H_{\inst'}$, while not introducing
any new edge to $H_{\inst'}$. Thus, $H_{\inst'}$ has no non-zero cycles.
We pass the instance to a postprocessing step described in Section~\ref{sss:ID2-reversal}.

Observe that a single guessing step is successful with probability $\Omega((k\maxarity)^{-1})$.
As discussed, the depth of the recursion is $\Oh(k\lambda^2) = \Oh(k^3 \maxarity^2)$.
Hence, the success probability of the algorithm so far is $2^{-\Oh(k^3 \maxarity^2 \log(k\maxarity))}$.


\subsubsection{From a mincut instance without non-zero cycles to an instance without clauses}\label{sss:ID2-reversal}

We now formalize the aforementioned idea of reversing some of the flow paths to get rid
of the clauses.

As in the previous subsection,
we are given a $2K_2$-free \gdpcshort{} instance $\inst = (G,s,t,\pairs,\bundles,\weight,k,W)$ together
with a witnessing flow $\witnessflow$. Our goal is to find any solution to $\inst$ if there exists a proper solution $Z$: a solution that is an $st$-mincut
that, if it violates a bundle $B$, then it contains all deletable edges of $B$.

The additional assumption we have is that the auxiliary graph $H_\inst$
has no non-zero cycles.
Also, $\witnessflow \neq \emptyset$, as otherwise the algorithm from the previous section
would have returned an answer.
Hence, $[\lambda]$ can be partitioned as $I_0 \uplus I_1$ such 
that (a) for every nonempty $E_{i,j}$, both $i$ and $j$ are in the same set $I_\eta$;
(b) for every nonempty $\pairs_{i,j}$, $i \neq j$ and $i$ and $j$ are in distinct sets
$I_\eta$. See Figure~\ref{fig:3.3.5}. 

\begin{figure}[tbh]
\begin{center}
\includegraphics{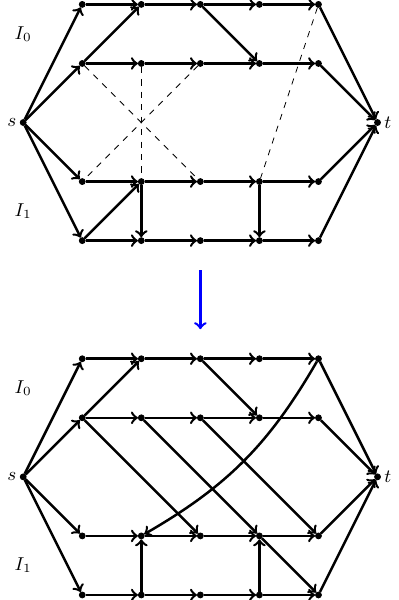}
\caption{An example of the reversing process, where we reverse
the paths in $I_1$ (bottom half of the graph). Clauses going across become
arcs and arcs in the bottom half reverse.}\label{fig:3.3.5}
\end{center}
\end{figure}

We perform the following variant of the cleanup operation. 
Most of the steps below are also present in the cleanup operation 
defined in Section~\ref{subsub:clean}. However,
we need to refrain from contracting crisp arcs $(s,v)$ where $v$
lies on a vertex of $\witnessflow$, as such an operation may add an arc to $H_\inst$.

\begin{enumerate}
\item 
We delete from $G$ any vertex $v$ where there is no path from $v$
to $t$ nor to any endpoint of a clause in $G-\{s\}$. 
\item 
We also delete from $G$ any vertex $v$ that is not reachable from $s$ in $G-\{t\}$.

Note that the two steps above do not change 
the space of all $st$-mincuts nor whether an $st$-mincut is a (proper) solution.
In particular, this operation does not delete any vertex on a flow path of $\witnessflow$.
Since $\witnessflow \neq \emptyset$, the vertex $t$ is not deleted.

\item 
We also exhaustively contract every crisp edge $(v,t)$, delete any clause that
contains $t$ or an arc that has a tail in $t$, and terminate the current branch
with a negative answer
if no $st$-mincut of size $\lambda$ exists (this in particular happens if 
    there is a crisp arc $(s,t)$).

\item \label{cleanup2:sv}
Since $H_\inst$ has no loops, there is no vertex $v$ that is on one hand an endpoint of
a clause, and on the other hand admits a path from $s$ to $v$ that does not contain
any vertex on a flow path of $\witnessflow$ except for $s$ because for a clause $\{u,v\}$ with such $v$,
$\pairs_{i,i}$ would contain $(\pi_i(u),\pi_i(v))$ for some $i$ and thus adding a loop on $H_\inst.$
Consequently, if $(s,v)$ is a crisp arc with $v$ not on $\witnessflow$, we can contract
$(s,v)$ without creating any new edge of $H_\inst$.
We perform this process exhaustively.
\end{enumerate}

After the above cleaning, we observe the following.
\begin{claim}\label{claim:01classification}
Every $v \in V(G) \setminus \{s,t\}$ satisfies exactly one of the following:
\begin{enumerate}
\item There exists $i \in I_0$ and a path $P(i,v)$ from a vertex
of $V(P_i) \setminus \{s,t\}$ to $v$ that does not contain any vertex
of $\witnessflow$ except for the starting vertex.
\item There exists $j \in I_1$ and a path $P(j,v)$ from a vertex
of $V(P_j) \setminus \{s,t\}$ to $v$ that does not contain any vertex
of $\witnessflow$ except for the starting vertex.
\end{enumerate}
\end{claim}
\begin{proofofclaim}
We first show that at least one of the above options hold. 
By the cleaning step, $v$ is reachable from $s$ in $G-\{t\}$; let $P$ be a path from $s$ to $v$
that avoids $t$.
If $P$ contains a vertex of $\witnessflow$ distinct from $s$, then the subpath of $P$
from the last such vertex to $v$ witnesses that $v$ falls into one of the above cases.
Otherwise, the first edge of $P$ would be contracted in Point~\ref{cleanup2:sv}
of the cleaning process described above, a contradiction.

Assume now that both paths $P(i,v)$ and $P(j,v)$ exist for some $i \in I_0$ and $j \in I_1$.
If $v$ lies on a flow path $P_\zeta \in \witnessflow$, then either if $\zeta \in I_1$
and then $P(i,v)$ witnesses $E_{i,\zeta} \neq \emptyset$,
or $\zeta \in I_0$ and $P(j,v)$ witnesses $E_{j,\zeta} \neq \emptyset$, both 
being a contradiction with the properties of the partition $I_0 \uplus I_1$.
So $v$ does not lie on a flow path of $\witnessflow$.
Since $v$ is not reduced in the cleaning phase, there is a path $Q$ in $G-\{s\}$ 
from $v$ to either $t$ or to an endpoint of a clause.

In the first case, let $Q'$ be the prefix of $Q$ from $v$ to the first intersection with
a path of $\witnessflow$. The endpoint of $Q'$ is not $t$, as otherwise the last edge of
$Q'$ should have been contracted in the cleaning phase. 
Thus, the endpoint of $Q'$ is in $V(P_\zeta) \setminus \{s,t\}$ for some $\zeta \in [\lambda]$.
The concatenation of $P(i,v)$ and $Q'$ implies $E_{i,\zeta} \neq \emptyset$
and the concatenation of $P(j,v)$ and $Q'$ implies $E_{j,\zeta} \neq \emptyset$.
This is a contradiction with the fact that $i$ and $j$ lie in different sets $I_0,I_1$.

In the second case, let $\{u,u'\} \in \pairs$ be such that $u$ is the endpoint of $Q$.
Furthermore, we can assume that there is no path from $v$ to $t$ in $G-\{s\}$.
In particular, this implies that $Q$ has no vertex of $\witnessflow$, as otherwise we could
redirect $Q$ from the first intersection with $\witnessflow$, along a flow path to $t$.
We also observe that both $\pi_i(u)$ and $\pi_j(u)$ are distinct from $t$.
Furthermore, the assumption that $u'$ is reachable from $s$ in $G-\{t\}$ implies
that $\pi_\zeta(u') \neq t$ for at least one index $\zeta \in [\lambda]$.
This implies that either $H_\zeta$ has a loop (if any of $\pi_i(u)$, $\pi_j(u)$, or $\pi_\zeta(u')$ equals $s$) or both $\pairs_{i,\zeta}$ and $\pairs_{j,\zeta}$ are nonempty.
Both conclusions yield a contradiction, as $i$ and $j$ lie in different sets $I_0,I_1$.
\end{proofofclaim}

Claim~\ref{claim:01classification} motivates us to define a partition $V(G)\setminus \{s,t\} = V_0 \uplus V_1$; $v \in V_0$ if the first case of Claim~\ref{claim:01classification}
happens and $v \in V_1$ otherwise.
We have the following observation.
\begin{claim}\label{claim:01split}
For every clause $\{u,v\} \in \pairs$, we have $u,v \in V(G) \setminus \{s,t\}$
and $u$ and $v$ lie in two distinct sets among $V_0$ and $V_1$.
For every arc $(u,v) \in E(G-\{s,t\})$, either $u,v \in V_0$ or $u,v \in V_1$.
\end{claim}
\begin{proofofclaim}
For the statement about a clause $\{u,v\}$, first note that if $u=t$ or $v=t$, then the clause
would have been deleted in the cleaning phase, and if $u=s$ or $v=s$, then
the clause would yield a loop in $H_\inst$, a contradiction.
Furthermore, if $u,v \in V_\eta$ for $\eta \in \{0,1\}$, then 
the corresponding paths $P(i_u,u)$ and $P(i_v,v)$ for $i_u,i_v \in I_\eta$
would certify $\pairs_{i_u,i_v} \neq \emptyset$, and an edge $i_ui_v$ of label $1$
in $H_\inst$.
As $i_u,i_v \in I_\eta$, this is a contradiction
to the properties of the partition $I_0 \uplus I_1$.

Take now an edge $(u,v)$ with $u \in V_\eta$ and $v \neq \{s,t\}$ for $\eta \in \{0,1\}$.
Let $P(i,u)$ for $i \in I_\eta$ be a path witnessing $u \in V_\eta$.
If $v$ lies on $P(i,u)$, then the prefix of $P(i,u)$ until $v$ certifies $v \in V_\eta$.
Otherwise, the concatenation of $P(i,u)$ and $(u,v)$ certifies $v \in V_\eta$.
This proves the second stament.
\end{proofofclaim}

For $\eta=0,1$, let $G_\eta$ be the subgraph of $G$ induced by $V_\eta \cup \{s,t\}$,
with the following exception on edges between $s$ and $t$.
$G$ does not contain an arc $(t,s)$ (as it would be deleted in the cleaning phase)
and a crisp arc $(s,t)$ (as it would yield a negative answer in the cleaning phase).
Any soft arc $(s,t)$ necessarily belongs to a flow path of $\witnessflow$;
for every $i \in I_\eta$ such that $P_i$ consists of a soft arc $(s,t)$, we put this
edge to $G_\eta$.
In this manner, by Claim~\ref{claim:01split}, $G_0$ and $G_1$ are edge-disjoint
and each edge of $G$ belongs to exactly one of $G_0$ and $G_1$.
Also, by Claim~\ref{claim:01split}, every clause of $\inst$ has one vertex
in $G_0$ and one vertex in $G_1$.

We now perform the following operation.
First, we reverse all edges of $G_1$ and swap the labels of vertices $s$ and $t$, obtaining
a graph $G_1'$.
Second, we replace every clause $\{u,v\}$ with $u \in V_0$ and $v \in V_1$
with an arc $(u,v)$. Third, we identify the two vertices $s$ of $G_0$ and $G_1'$
into a vertex $s$ and the two vertices $t$ of $G_0$ and $G_1'$ into a vertex $t$.
Let $G'$ be the resulting graph. Every edge and clause of $\inst$
has its corresponding edge in $G'$. The operation naturally defines bundles $\bundles'$
in $G'$, now consisting of edges only; we extend the weight function $\weight$
naturally to $\bundles'$. See Figure~\ref{fig:3.3.5} for an illustration.

Clearly, every bundle of $\bundles'$ is $2K_2$-free.
Hence, we obtained a new \gdpcshort{} instance $\inst' = (G',s,t,\emptyset,\bundles',k,\omega,W)$
that is $2K_2$-free and has no clauses. We solve such instances in Section~\ref{ss:ID2-no-clauses}
via Theorem~\ref{thm:ID2-no-clauses}.
Recall that the success probability so far in the algorithm of Theorem~\ref{thm:ID2-mincut} is $2^{-\Oh(k^2 \maxarity^3 \log(k\maxarity))}$. 
Together with the success probability of Theorem~\ref{thm:ID2-no-clauses}, it gives the desired success probability promised in Theorem~\ref{thm:ID2-mincut}.
Here, it remains to formally check that it suffices to solve $\inst'$.

\begin{claim}
Let $Z$ be a proper solution to $\inst$ and let $X$ be the set of clauses violated
by $Z$. Then, $Z \cup X$ is an $st$-cut in $G'$ that violates
at most $k$ bundles of $\bundles'$ of total weight at most $W$.
\end{claim}
\begin{proofofclaim}
Let $S$ consist of the vertices reachable from $s$
in $G_0-Z$ and not reachable from $s$ in $G_1-Z$, except for $t$.
We claim that $S$ contains the set of vertices reachable from $s$ in $G'-(Z \cup X)$.
Indeed, if for a clause $\{u,v\} \in \pairs$ with $u \in V_0$ and $v \in V_1$, 
  the tail of the corresponding arc $(u,v) \in E(G')$ is in $S$, then 
  $u$ is reachable from $s$ in $G-Z$, and thus either $\{u,v\} \in X$
  or $v$ is not reachable from $s$ in $G-Z$. Hence, in the latter case $v \in S$ as well.
Hence, $Z \cup X$ is an $st$-cut in $G'$.

The claim about the cardinality and cost of violated bundles follows immediately from
the fact that if $Z$ violates a bundle $B \in \bundles$, then $Z \cup X$ contains an edge
of the corresponding bundle in $\bundles'$.
\end{proofofclaim}
\begin{claim}
Let $Z'$ be an $st$-cut in $G'$ that violates at most $k$ bundles of $\bundles'$
of total weight at most $W$.
Then the set of edges of $G$ that correspond to edges of $Z'$ is a solution
to $\inst$.
\end{claim}
\begin{proofofclaim}
Assume $Z'$ corresponds to edges $Z$ and clauses $D$ in $\inst$.
First, observe that $Z$ is an $st$-cut in $G$: any $st$-path in $G$ is either contained
in $G_0$ or in $G_1$, so either it is cut by $Z'$ (if it is in $G_0$) or its reversal
is cut in $G_1'$. 
Second, assume that a clause $\{u,v\}$ with $u \in V_0$, $v \in V_1$ is violated
by $Z$ in $\inst$. Then, $u$ is reachable from $s$ in $G_0-Z$ and $v$
is reachable from $s$ in $G_1-Z$. Let $P$ be a path from $s$ to $u$ in $G_0-Z$,
 let $Q$ be a path from $s$ to $v$ in $G_1-Z$, and let $Q'$ be the reversal
 of $Q$ in $G_1' - Z'$ (that is now a path from $v$ to $t$). 
Then, the concatenation of $P$, $(u,v)$, and $Q'$ is a path from $s$ to $t$ in $G'-Z'$
unless $(u,v) \in Z'$, that is, unless $\{u,v\} \in D$. 
Consequently, every bundle violated by $Z$ corresponds to a bundle violated by $Z'$,
  finishing the proof.
\end{proofofclaim}

\subsubsection{Solving an instance without clauses}\label{ss:ID2-no-clauses}

Thus, we are left with solving a \gdpcshort{} instance
$\inst=(G,s,t,\emptyset,\bundles,k,\omega,W)$ that has no clauses and is $2K_2$-free.
We observe that we end up with the \textsc{Weighted Bundled Cut} instance with the property called \emph{pairwise linked deletable edges} that
was shown to be FPT in the first paper of this series, namely~\cite{dfl-arxiv}.

Let $\inst=(G,s,t,\emptyset,\bundles,k,\omega,W)$ be a \gdpcshort{} instance without clauses.
Recall that a soft arc $e=(u,v)$ is \emph{deletable} if there is no parallel crisp arc $(u,v)$, and \emph{undeletable} otherwise.
We say that a bundle $B \in \bundles$ has \emph{pairwise linked deletable arcs} if for every two deletable arcs $e_1,e_2 \in B$
that are not incident with either $s$ or $t$, there is a path from an endpoint of one of the edges to an endpoint of the other
that does not use an edge of another bundle (i.e., uses only edges of $B$ and crisp edges).
Note that a $2K_2$-free bundle without clauses has pairwise linked deletable arcs, but the latter property is slightly more general. 

One of the applications of flow-augmentation of~\cite{dfl-arxiv} can be stated as follows.
\begin{theorem}[Theorem~4.1 of~\cite{dfl-arxiv}]\label{thm:ID2-no-clauses}
There exists a randomized polynomial-time algorithm that, given a \gdpcshort{} instance
$\inst=(G,s,t,\emptyset,\bundles,k,\omega,W)$ with no clauses and whose every bundle has pairwise linked deletable arcs, 
never accepts a \no-instance and accepts a \yes-instance with probability $2^{-\Oh(k^4 d^4 \log(kd))} n^{\Oh(1)}$, where $d \leq \maxarity^2$ is the maximum
number of deletable arcs in a single bundle.
\end{theorem}
As Theorem~\ref{thm:ID2-no-clauses} solves the remaining case, the proof of Theorem~\ref{thm:ID2-mincut} is completed, which implies Theorem~\ref{thm:alg-ID2}.

\section{IHS-B positive case}
\label{sec:arrow}
This section is devoted to the proof of the following theorem, which captures the second new tractable case. (This statement is a detailed version of Theorem~\ref{ithm:arrow}.)

\begin{theorem}\label{thm:alg-IS}
For every fixed $\maxarity \geq 2$, there exists 
\begin{itemize}
\item a randomized polynomial-time algorithm that, given an instance of \MinSAT{\isdpositive} with parameter $k$, never accepts a \no-instance and accepts a \yes-instance with probability at least $2^{-\Oh(k^6 \maxarity^{10} \log(k\maxarity))}$;
\item a deterministic algorithm that solves \MinSAT{\isdpositive}
in time $2^{\Oh(k^6 \maxarity^{10} \log(k\maxarity))} n^{\Oh(1)}$. 
\end{itemize}
\end{theorem}

We start with analyzing the constraints over $\isdpositive$
in Section~\ref{ss:isd-analysis} to establish a canonical way of expressing a constraint
as a conjunction of implications and negative clauses. 
This allows us in Section~\ref{ss:isd-to-graph}
to cast our problem as a cut problem in digraphs: 
a variant of \textsc{Generalized Bundled Cut} that we call \textsc{Clause Cut}.
Compared to the \gdpcfull{} problem of the previous section,
\textsc{Clause Cut} features larger negative clauses, but
the $2K_2$-freeness assumption applies now to the set of arcs (i.e., ignoring the clauses). 

We deal with \textsc{Clause Cut} in Section~\ref{ss:basic-clause-cut}
and~\ref{ss:clause-cut}.
First, in Section~\ref{ss:basic-clause-cut} we set up the scene with a number of basic
operations that simplify the instance. 
In Section~\ref{ss:clause-cut} we provide an algorithm for \textsc{Clause Cut}.
Here, we first make the instance more organized
and regular by a number of guessing and color-coding steps.
This regularization allows us to make a crucial greedy observation
(Claim~\ref{claim:clause-cut:X}), which in some sense says that there is only
one ``canonical'' candidate for a solution to the current instance.
This observation allows to either conclude (if the candidate is indeed a solution)
or perform a subsequent branching step (otherwise) and recurse.

\subsection[Constraints over $\isdpositive$]{Constraints over $\boldsymbol{\isdpositive}$}\label{ss:isd-analysis}

We start with an analysis of the structure of constraints over $\isdpositive.$ 
Consider any $R \in \isdpositive$ with $R\neq\emptyset$, i.e., such that $R$-constraints have at least one satisfying assignment.
(The never-satisfied constraints can be deleted from the input instance and the 
 parameter reduced accordingly, so we henceforth assume that in the input instance
 every constraint has at least one satisfying assignment.)
Let $\constraint=R(x_1,\ldots,x_r)$ be any $R$-constraint. We distinguish the following subsets of variables of \constraint:
\begin{itemize}
\item $\vars{\constraint}=\{x_1,\ldots,x_r\}$ is the set of all variables used in $\constraint=R(x_1,\ldots,x_r)$;
\item $\ones{\constraint}$ is the set of these variables $x_i \in \vars{\constraint}$ such that every satisfying
assignment of $R(x_1,\ldots,x_r)$ sets $x_i$ to $1$;
\item $\zeroes{\constraint}$ is the set of these variables $x_i \in \vars{\constraint}$ such that every satisfying
assignment of $R(x_1,\ldots,x_r)$ sets $x_i$ to $0$;
\item $\rest{\constraint} := \vars{\constraint} \setminus (\ones{\constraint} \cup \zeroes{\constraint})$ is the set of remaining
variables, that is, variables $x_i$ for which there exists both a satisfying assignment setting
$x_i$ to $1$ and a satisfying assignment setting $x$ to $0$. 
\end{itemize}

Our goal now is to express \constraint as a conjunction of implications and negative clauses in a canonical way. 
A fixed $R \in \isdpositive$ can have many equivalent representations as a conjunction
of implications and negative clauses; we choose to use the ``most rich one'', consisting
of all possible implications and negative clauses implied by $R$. 
In this manner, we will be able to translate the $2K_2$-freeness assumption of the arrow graph
onto the set of implications. 

\begin{definition}
Let $R\in\isdpositive$ of arity $r\leq b$ and with $R\neq\emptyset$.
The \emph{canonical definition} of a constraint $\constraint=R(x_1,\ldots,x_r)$ is the following set $\candef{\constraint}$ of constraints on variables $\{x_1,\ldots,x_r\}$ and constants $0$ and~$1$:
\begin{itemize}
\item For every $x_i \in \ones{\constraint}$, an implication $(1 \rightarrow x_i)$.
\item For every $x_i \in \zeroes{\constraint}$, an implication $(x_i \rightarrow 0)$.
\item For every $x_i,x_j \in \rest{\constraint}$, with $i\neq j$, such that $(i,j)$ is an edge of the arrow graph $H_R$ of $R$, an implication $(x_i \rightarrow x_j)$.
\item All negative clauses on variables of $\rest{\constraint}$ that are true in every satisfying assignment
of $\constraint$. 
\end{itemize}
\end{definition}

\begin{lemma}
Let $R\in\isdpositive$ of arity $r\leq b$ and with $R\neq\emptyset$.
An assignment $\alpha$ satisfies a constraint $\constraint=R(x_1,\ldots,x_r)$ if and only if it satisfies all constraints in the canonical definition $\candef{\constraint}$ of \constraint.
\end{lemma}

\begin{proof}
The forward implication is immediate by definition of $\candef{\constraint}$. 
For the backward definition, consider any assignment $\alpha$
that satisfies all constraints in $\candef{\constraint}$. 
Let $\phi$ be a set of implications, negative clauses, and assignments on variables $x_1,\ldots,x_r$ whose conjunction is equivalent to $\constraint=R(x_1,\ldots,x_r)$, according to the definition of $\isdpositive$.
Let $C$ be an arbitrary constraint in $\phi$; we want to show that it is satisfied by $\alpha$.
In what follows, we treat an assignment $(x = 1)$ as an implication $(1 \to x)$ and an assignment $(x = 0)$ as an implication $(x \to 0)$. 

Assume first that $C$ is an implication $(x \rightarrow y)$, where $x$ and $y$ are variables
or constants.
If $x \equiv 1$ or $x \in \ones{\constraint}$, then $y$ is a variable and $y \in \ones{\constraint}$.
Symmetrically, if $y \equiv 0$ or $y \in \zeroes{\constraint}$, then $x$ is a variable and $x \in \zeroes{\constraint}$.
If $x \in \zeroes{\constraint}$, then $x \rightarrow 0 \in \candef{\constraint}$, so $\alpha(x) = 0$ and $C$ is satisfied.
Symmetrically, $C$ is satisfied if $y \in \ones{\constraint}$. 
In the remaining case, $x,y \in \rest{\constraint}$. Then, the implication $C = (x \rightarrow y)$ 
forbids all assignments with $x=1$ and $y=0$, but as $x,y \in \rest{\constraint}$, there is an assignment
with $x=1$ (and hence $y=1)$ and an assignment with $y=0$ (and hence $x=0$). 
Consequently, $(x,y)$ corresponds to an edge of the arrow graph $H_R$, and thus $C \in \candef{\constraint}$ and
is satisfied by $\alpha$. 

Consider now the case where $C$ is a negative clause.
If $C$ contains a variable of $\zeroes{\constraint}$, then $C$ is satisfied by $\alpha$.
Otherwise, let $C'$ be the clause $C$ restricted to the literals with variables in $\rest{\constraint}$,
that is, with literals with variables in $\ones{\constraint}$ deleted. 
Then, $C'$ is also true in every satisfying assignment of $\constraint$. Hence, $C' \in \candef{\constraint}$
and thus $C$ is satisfied by $\alpha$.

Thus, $\alpha$ satisfies all constraints in $\phi$, and therefore satisfies $\constraint=R(x_1,\ldots,x_r)$, as claimed.
\end{proof}

\subsection{Graph representation}\label{ss:isd-to-graph}

Our goal now is to translate the input instance $(\F,k)$ for \MinSAT{\isdpositive} into an equivalent instance of a graph separation problem, called \textsc{Clause Cut}, using the canonical definitions of all constraints \constraint of $\F$. 

A \textsc{Clause Cut} instance consists of:
\begin{itemize}
\item A directed (multi)graph $G$ with distinguished vertices $s,t\in V(G)$. 
\item A family (possibly a multiset) $\clauses$ of subsets of vertices of $V(G)$,
  called \emph{clauses}. 
\item A family $\bundles$ of pairwise disjoint subsets of $\clauses \cup E(G)$, called 
\emph{bundles}, such that each bundle $B$ contains at most one copy of each arc and each clause.
\item An integer $k$, called the \emph{budget}. 
\end{itemize}
The main difference between \textsc{Clause Cut} and \gdpcfull{}, considered in the previous section, is that here we allow clauses of arity higher than two. On the other hands, we do not consider weights here. 

For a \textsc{Clause Cut} instance $\inst = (G,s,t,\clauses,\bundles,k)$, its \emph{arity}
is the maximum over: size of any clause and the number of vertices involved in any
bundle (a vertex is involved in a bundle if it is an endpoint of an arc of the bundle or an element of a clause of the bundle).
An arc or a clause of $\inst$ is \emph{soft} if it is contained in a bundle and \emph{crisp}
otherwise. 
A \emph{cut} in $\inst$ is an $st$-cut $Z$ in $G$ that contains only soft arcs. 
A clause $C\in\clauses$ is \emph{satisfied} by $Z$ if at least one element of $C$ is not reachable
from $s$ in $G-Z$, and \emph{violated} otherwise.
A cut $Z$ is a \emph{feasible solution} if every crisp clause is satisfied by $Z$. 
A bundle $B$ is \emph{satisfied} by a feasible solution $Z$ if no arc of $B$ is in $Z$
and no clause of $B$ is violated by $Z$, and \emph{violated} otherwise.
The \emph{cost} of a feasible solution is the number of violated bundles. 
A \textsc{Clause Cut} instance $\inst$ asks for the existence of a feasible solution $Z$ of cost
at most $k$. 
A \textsc{Clause Cut} instance $\inst$ is \emph{$2K_2$-free} if for every bundle $B$,
the arcs of $B$ in $G-\{s,t\}$ form a $2K_2$-free graph, that is, for every two arcs $e,f$ of $B$,
one of the following cases holds:
(1)~$e$ and $f$ share an endpoint, (2)~$e$ or $f$ is incident with  $s$ or $t$, (3)~$e$ or $f$ is a loop,
or (4)~there is another
arc of $B$ with one endpoint shared with $e$ and the other endpoint shared with $f$.

Note that if $\inst = (G,s,t,\clauses,\bundles,k)$ is of arity at most $\maxarity$
and if $Z$ is a feasible solution to $\inst$ of cost at most $k$, then
$|Z| \leq k\maxarity^2$, as a bundle cannot contain multiple copies of the same arc.
Canonical definitions allow us to phrase the input instance to \MinSAT{\isdpositive}
as a \textsc{Clause Cut} instance with parameter and arity bounded by $k$ and $\maxarity$.
The next lemma is a restatement of Lemma~\ref{lem:intro:isd-to-graph} with the current
notation.

\begin{lemma}\label{lem:isd-to-graph}
Given a \MinSAT{\isdpositive}
instance $(\F,k)$
one can in polynomial time output an equivalent $2K_2$-free \textsc{Clause Cut} instance
of arity at most $\maxarity$ and budget at most $k$. 
\end{lemma}
\begin{proof}
First, we delete all constraints of $\F$ that do not have any satisfying assignments and decrease
$k$ accordingly. If $k$ becomes negative as a result, we return a trivial \no-instance. 

The vertex set $V(G)$ consists of the variable set of $\F$ together with additional vertices $s$ and $t$.
Intuitively, $s$ will represent the value $1$ and $t$ will represent the value $0$. 

For every constraint $\constraint=R(x_1,\ldots,x_r)$ in $\F$, we proceed as follows; note that $R\in\isdpositive$. 
For every implication $(x_i\rightarrow x_j)$ of $\candef{\constraint}$, we add the corresponding arc $(x_i,x_j)$ to $G$, noting that this will create a separate copy of the arc each time it is applied to a new constraint but only a single arc per constraint.
For every implication $(1\rightarrow x_i)$ of $\candef{\constraint}$, we add an arc $(s,x_i)$ and for every implication $(x_i\rightarrow 0)$ of $\candef{\constraint}$ we add an arc $(x_i,t)$, again creating separate copies each time but only one arc per constraint.
For every negative clause of $\candef{\constraint}$, we add the corresponding clause to $\clauses$; that is, 
whenever there is a clause $\bigvee_{i\in J} (\neg x_i)$ in $\candef{\constraint}$ for some $J\subseteq [r]$, the vertex set $\{x_i:i\in J\}$ is created as a clause in $\clauses.$
Finally, we add to $\bundles$ a bundle consisting of all arcs and clauses created for constraint \constraint.

The budget of the output instance equals the current value of $k$ (i.e., after the deletions
of never-satisfied constraints). This completes the description of the output \textsc{Clause Cut}
instance. 
The bound on the arity of the output instance follows immediately from the bound $\maxarity$
on the arity of constraints of the input instance. 
To see that the output instance is $2K_2$-free, note that for every $\constraint=R(x_1,\ldots,x_r)$ in $\F$ (hence $R \in \isdpositive$)
the set of implications of $\candef{\constraint}$ on $\rest{\constraint}$ is $2K_2$-free because the arrow graph $H_R$ of $R$ is $2K_2$-free and, further,
the class of $2K_2$-free graphs is closed under the operation of vertex identification (Lemma~\ref{lem:2k2-free-contraction}).

It remains to prove equivalence of the instances. 
\begin{claim}
The input instance $(\F,k)$ and the constructed \textsc{Clause Cut} instance are equivalent.
\end{claim}
\begin{proofofclaim}
In one direction, consider a feasible solution $Z$ to the output instance. 
Let $\alpha$ be a variable assignment that assigns $\alpha(x)=1$ if and only if the vertex
corresponding to $x$ is reachable from $s$ in $G-Z$. 
For every constraint $\constraint=R(x_1,\ldots,x_r)$ of $\F$ (minus those removed for being never satisfiable), consider 
its canonical definition $\candef{\constraint}$. 
Observe that if an implication 
of $\candef{\constraint}$ is violated by $\alpha$, then the corresponding
arc of $G$ is in $Z$, and if a negative clause of $\candef{\constraint}$ is violated by $\alpha$, 
then the corresponding clause of $\clauses$ is violated by $Z$. 
Consequently, if a constraint of $\F$ is violated by $\alpha$, then its corresponding
bundle is violated by $Z$. Thus, $\alpha$ violates at most $k$ constraints. 

In the other direction, let $\alpha$ be a variable assignment that violates
at most $k$ constraints of $\F$. 
For each constraint $\constraint=R(x_1,\ldots,x_r)$, with $R\in\isdpositive$, 
consider its canonical definition $\candef{\constraint}$.
Define the set $Z$
to be the set of all arcs of $G$ that correspond to the implications 
of the canonical definitions
that are violated by $\alpha$ (over all constraints of $\F$). 
Observe that $Z$ is an $st$-cut, as a path from $s$ to $t$ in $G$ corresponds to a sequence
of implications from $1$ to $0$, and thus one of these implications is violated by $\alpha$.
Since there are no crisp clauses, $Z$ is a feasible solution.
Finally, observe that if $Z$ violates a bundle, then $\alpha$ violates the corresponding
constraint of $\F$. Hence, the cost of $Z$ is at most $k$, as desired. 
\end{proofofclaim}

This finishes the proof of the lemma.
\end{proof}

\subsection{Basic operations on a \textsc{Clause Cut} instance}\label{ss:basic-clause-cut}

Lemma~\ref{lem:isd-to-graph} ensures that, in order
to prove Theorem~\ref{thm:alg-IS}, it suffices to solve \textsc{Clause Cut}. 

Let $\inst = (G,s,t,\clauses,\bundles,k)$ be a \textsc{Clause Cut} instance. 
We will use the following four basic operations on $G$. 

For $x,y \in V(G)$, 
by \emph{adding a (crisp) arc $(x,y)$}, we mean adding the arc $(x,y)$ to $G$
(and not adding it to any bundle, so it becomes crisp). 
For $x,y \in V(G)$, by \emph{identifying $x$ and $y$} we mean contracting $x$ and $y$
into a single vertex of $G$ (if $x$ or $y$ equals $s$ or $t$, the new vertex is also called
    $s$ or $t$; we will never identify $s$ with $t$).
Recall that $2K_2$-freeness is preserved under vertex identification (Lemma~\ref{lem:2k2-free-contraction}).
For $B \in \bundles$, by \emph{breaking $B$} we mean deleting $B$ from $\bundles$
(thus making all arcs and clauses of $B$ crisp). 
For $B \in \bundles$, by \emph{deleting $B$} we mean deleting $B$ from $\bundles$, 
deleting all its arcs and clauses from $E(G)$ and $\clauses$, and decreasing $k$ by one. 
These operations have the following straightforward properties.

\begin{proposition}
Let $\inst = (G,s,t,\clauses,\bundles,k)$ be a $2K_2$-free \textsc{Clause Cut} instance
of arity at most $\maxarity$. 
\begin{itemize}
\item Let $x,y \in V(G)$ and let $\inst'$ be the result of adding the arc $(x,y)$ to $G$.
  If $\inst'$ is a \yes-instance, so is $\inst$.
  In the other direction,
  if $\inst$ is a \yes-instance with a feasible solution $Z$ of cost at most $k$ 
  such that either $x$ is not reachable from $s$ in $G-Z$ or $y$ is reachable from $s$ in $G-Z$,
  then $\inst'$ is a \yes-instance.
\item Let $x,y \in V(G)$, $\{x,y\} \neq \{s,t\}$,
  and let $\inst'$ be the result of identifying $x$ and $y$.
  If $\inst'$ is a \yes-instance, so is $\inst$.
  In the other direction,
  if $\inst$ is a \yes-instance with a feasible solution $Z$ of cost at most $k$ 
  such that $x$ is reachable from $G-Z$ if and only if $y$ is,
  then $\inst'$ is a \yes-instance.
\item Let $B \in \bundles$ and let $\inst'$ be the result of breaking $B$.
  If $\inst'$ is a \yes-instance, so is $\inst$.
  In the other direction,
  if $\inst$ is a \yes-instance with a feasible solution $Z$ of cost at most $k$ 
    that does not violate $B$,
  then $\inst'$ is a \yes-instance.
\item Let $B \in \bundles$ and let $\inst'$ be the result of deleting $B$.
  If $\inst'$ is a \yes-instance, so is $\inst$.
  In the other direction,
  if $\inst$ is a \yes-instance with a feasible solution $Z$ of cost at most $k$ 
    that violates $B$,
  then $\inst'$ is a \yes-instance.
\end{itemize}
\end{proposition}
\begin{proof}
The proofs are straightforward and very similar to each other; we provide here only the proof of the last point
as an example (as it is probably the most complicated).

Let $G$ ($G'$) be the graph of the instance $\inst$ (resp. $\inst'$). 

In one direction, assume that $\inst'$ is a yes-instance and let $Z'$ be a feasible solution to $\inst'$ of cost at most $k-1$. 
Let $Z_B$ be the set of those arcs $(u,v) \in B$ such that, in $G'$, the vertex $u$ is in the $s$-side of $Z'$ but $v$ is in the $t$-side of $Z'$.
Let $Z = Z' \cup Z_B$. Observe that, by the definition of $Z_B$, $Z$ is a star $st$-cut in $G$
and a vertex $u \in V(G) = V(G')$ is in the $s$-side of $Z$ in $G$ if and only if it is in the $s$-side of $Z'$ in $G'$.
Consequently, every bundle violated by $Z$ in $\inst$ is either equal to $B$ or is also violated by $Z'$ in $\inst'$. Hence, $Z$ is a feasible solution of cost at most $k$ in $\inst$, as desired.

In the other direction, assume that $\inst$ is a yes-instance that admits a feasible solution $Z$ of cost at most $k$ that violates $B$. 
Note that if $u \in V(G)$ is reachable from $s$ in $G'-(Z \setminus B)$, then it is also reachable from $s$ in $G-Z$, but the implication in the other direction may not be true.
Let $Z'$ be the set of those arcs $(u,v) \in Z \setminus B$ such that $u$ is reachable from $s$ in $G'-(Z \setminus B)$. 
Observe that, as $Z$ is a star $st$-cut in $G$, the vertex $v$ is not reachable from $s$ in $G'-(Z \setminus B)$. 
Hence, $Z'$ is a star $st$-cut in $G'$ with the following property: if $v \in V(G) = V(G')$ is in the $s$-side of $Z'$ in $G'$, then $v$ is also in the $s$-side of $Z$ in $G$.
Consequently, if $Z'$ violates a bundle $B'$ of $\inst'$, then $Z$ violates the same bundle in $\inst$. 
We infer that $Z'$ is a feasible solution to $\inst'$ of cost at most $k-1$, as desired. 
\end{proof}

\subsection{Solving \textsc{Clause Cut}}\label{ss:clause-cut}

Armed with the basic operations of the previous section, we are now ready to solve
\textsc{Clause Cut}, that is, to prove Theorem~\ref{thm:alg-clause-cut}.

\algclausecut*
\begin{proof}
As in the previous section, for clarity we focus only on the randomized part of Theorem~\ref{thm:alg-clause-cut}.
Again, the deterministic counterpart follows directly by replacing every randomized step by either simple
branching or, in the single color-coding step of the algorithm with an appropriate use of splitters~\cite{NaorSS95}. 

Let $\inst = (G,s,t,\clauses,\bundles,k)$ be a \textsc{Clause Cut} instance. 

The algorithm will make modifications to $\inst$, usually in the form of adding
(crisp) arcs to $G$ or identifying two vertices of $G$. It will be immediate that
neither of the operations can turn a \no-instance into a \yes-instance. 
The algorithm will make a number of randomized choices; we will make sure that
if the input instance is a \yes-instance, then it will stay a \yes-instance with sufficiently high
probability. Accordingly, we henceforth work under the assumption that $\inst$ is a \yes-instance
and we want to certify this fact with good probability. 

The algorithm will make a number of random choices, either finding a feasible solution of cost
at most $k$ to the instance at hand, or recursing. In the recursion, we will always guarantee
that either the budget $k$ decreases, or the budget $k$ stays while the value of the 
maximum $st$-flow $\lambda_G(s,t)$ increases.
As $\lambda_G(s,t) > k\maxarity^2$ implies that we are dealing with a \no-instance, this will 
guarantee a bound of $\Oh(k^2 \maxarity^2)$ on the depth of the recursion.
At every step of the recursion, any random choice will be correct with probability 
$2^{-\Oh(k^4 \maxarity^8 \log(k\maxarity))}$.
Together with the recursion depth bound, this will give the desired success probability of Theorem~\ref{thm:alg-clause-cut}.

We start by observing that we can restrict ourselves to feasible solutions $Z$
that are star $st$-cuts: if $Z$ is a feasible solution, $(x,y) \in Z$, but either $x$
is not reachable from $s$ in $G-Z$ or $y$ is reachable from $s$ in $G-Z$, then 
$Z \setminus \{(x,y)\}$ is a feasible solution as well, and, since
the set of vertices reachable from $s$ is the same in $G-Z$ as in $G-(Z \setminus \{(x,y)\})$, 
the cost of $Z \setminus \{(x,y)\}$ is not larger than the cost of $Z$.
Together with the bound $|Z| \leq k\maxarity^2$, this allows us to invoke directed flow-augmentation (Theorem~\ref{thm:dir-flow-augmentation}),
obtaining a pair $(A,\witnessflow)$. 
(In the deterministic setting, we use here Theorem~\ref{thm:dir-flow-augmentation-det}
 and iterate over all returned pairs $(A,\witnessflow)$.)
We add $A$ as crisp arcs to $G$. 
Observe that if $\inst$ is a \no-instance, the resulting instance is still a \no-instance,
while if $\inst$ is a \yes-instance,
with probability $2^{-\Oh(k^4 \maxarity^8 \log (k \maxarity))}$
the resulting instance has a feasible solution $Z$ of cost at most $k$
that is a star st-cut, $\corecutG{Z}{G+A}$ is an $st$-mincut, and $\witnessflow$
is a witnessing flow for $Z$.
We proceed further with the assumption that this is the case.

\paragraph{Structure of cut edges on flow paths.}
Let $\lambda = \lambda_G(s,t)$ and let $\witnessflow = \{P_1,P_2,\ldots,P_\lambda\}$. For $i \in [\lambda]$, let $e_i^Z$ be the unique edge of $Z$ on $P_i$.
Recall that $\lambda \leq |Z| \leq k\maxarity^2$. 
Randomly sample a function $f \colon [\lambda] \to [k]$, aiming for the following: For every $1 \leq i,j \leq \lambda$, 
we have $f(i) = f(j)$ if and only if $e_i^Z$ and $e_j^Z$ are in the same bundle.
Note that this happens with probability at least $k^{-\lambda} = 2^{-\Oh(k\maxarity^2 \log k)}$. 
We henceforth assume that this is the case.

Let $R_f \subseteq [k]$ be the range of $f$. 
For every $a \in R_f$ and for every two distinct $i,j \in f^{-1}(a)$, proceed as follows.
Recall that $e_i^Z$ and $e_j^Z$ are in the same bundle, denote it $B_a^Z$.
Since we are working with a $2K_2$-free instance and $Z$ is a star $st$-cut, one of the following options holds:
\begin{itemize}
\item $e_i^Z$ or $e_j^Z$ starts in $s$ or ends in $t$;
\item $e_i^Z$ and $e_j^Z$ share their tail;
\item $e_i^Z$ and $e_j^Z$ share their head;
\item there is another arc of $B_a^Z$ with one endpoint in $e_i^Z$ and second endpoint in $e_j^Z$. 
\end{itemize}
Note that the head of $e_i^Z$ cannot coincide with the tail of $e_j^Z$ (or vice versa) since both are edges of an $st$-mincut.
We randomly guess which of the four options above happens; if more than one option applies, we aim at guessing the earliest one applicable.
Additionally, in the first option we guess which of the four sub-options happen and in the last case we guess whether
the extra arc has its head in $e_i^Z$ and tail in $e_j^Z$ or vice versa. 
Our guess (for every $a$ in the range of $f$) is correct with probability $2^{-\Oh(\lambda^2)} = 2^{-\Oh(k^2 \maxarity^4)}$ and we continue under this assumption.
We henceforth call the above guess the \emph{relation between $e_i^Z$ and $e_j^Z$}. 

We also perform the following sampling: Every arc becomes red with probability $1/2$, independently of the other arcs. 
We aim for the following: For every bundle $B$ that is violated by $Z$, the red arcs of $B$ are exactly the arcs of $B$ that are in $Z$.
Since every bundle contains at most $\maxarity^2$ arcs, this happens with probability $2^{-\Oh(k \maxarity^2)}$. 
Furthermore, if our guess is correct, the following step maintains that $Z$ is a solution: For every soft arc $(u,v)$ that is not red, we add a crisp copy $(u,v)$, unless it is already
present in the graph.
(We remark in passing that the derandomization of this step 
 requires standard color-coding derandomization tools.)

\begin{definition}\label{def:candidate}
For every $a \in R_f$, we say that a bundle $B$ is a \emph{candidate for $B^Z_a$} if the following holds:
\begin{itemize}
\item for every $i \in f^{-1}(a)$, $B$ contains exactly one red arc on $P_i$, denoted henceforth $e(B,i)$;
\item for every $i \in [\lambda] \setminus f^{-1}(a)$, $B$ contains no red arc on $P_i$;
\item for every two distinct $i,j \in f^{-1}(a)$, the arcs $e(B,i)$ and $e(B,j)$ are in relation as guessed;
\item the set $\{e(B,i)~|~i \in f^{-1}(a)\}$ is a subset of some $st$-mincut in $G$ that uses only soft arcs.
\end{itemize}
\end{definition}
For every $i \in [\lambda]$, we break every bundle that contains a red edge of $P_i$ but is not a candidate for $B^Z_{f(i)}$. 
Note that if our guesses are correct, $Z$ remains a feasible solution.

We remark that, because of the last condition in the definition of a candidate for $B^Z_a$, the above step can be iterated and we perform it exhaustively.
That is, breaking a bundle $B$ makes some arcs crisp, which in turn can shrink the space of all $st$-mincuts that use only soft edges, which can break the last condition above for a different bundle. 

If for some $i \in [\lambda]$, we have guessed that $e^Z_i$ has its tail in $s$ or head in $t$, there is only one remaining candidate for
$B^Z_{f(i)}$: $B^Z_{f(i)}$ is the unique bundle that contains the first or last, respectively, edge of $P_i$. 
If this is the case, we delete the corresponding bundle and recurse. 
Henceforth we assume that this option does not happen for any $i \in [\lambda]$. 

Now, the crucial observation is that for every $a \in R_f$, the candidates for $B^Z_a$ are linearly ordered along all flow paths $P_i$, $i \in f^{-1}(a)$. 

\begin{claim}\label{claim:clause-cut:order}
Let $a \in R_f$, let $B$ and $B'$ be two candidates for $B^Z_a$, and let $i,j \in f^{-1}(a)$ be distinct.
Then $e(B,i)$ is before $e(B',i)$ on $P_i$ if and only if $e(B,j)$ is before $e(B',j)$ on $P_j$. 
\end{claim}

\begin{proofofclaim}
Assume the contrary. 
Note that the roles of $B$ and $B'$ are symmetrical and the roles of $i$ and $j$ are symmetrical. As a result, without loss of generality we can assume the following:
$e(B,i)$ is before $e(B',i)$ on $P_i$, $e(B,j)$ is after $e(B',j)$ on $P_j$,
and, furthermore, 
either $e_i^Z$ and $e_j^Z$ share a tail, $e_i^Z$ and $e_j^Z$ share a head, or there is another arc of $B_a^Z$ with a tail in $e_i^Z$ and head in $e_j^Z$. 

We claim that there is no $st$-mincut in $G$ that contains both $e(B',i)$ and $e(B',j)$, which is a contradiction to the assumption that $B'$ is a candidate for $B^Z_a$.
Indeed, the guessed relation between $e_i^Z$ and $e_j^Z$ and the fact that $B$ is a candidate for $B^Z_a$
implies that the following path from $s$ to $t$ would not be cut by such an $st$-mincut: follow $P_i$ from $s$ to $e(B,i)$, 
move to $e(B,j)$ via shared tail, head, or assumed arc with a tail in $e(B,i)$ and head in $e(B,j)$, and continue along $P_j$ to $t$. 
This finishes the proof.
\end{proofofclaim}

Claim~\ref{claim:clause-cut:order} allows us to enumerate, for every $a \in R_f$, the candidates for $B^Z_a$ as $B_{a,1}, B_{a,2},\ldots,B_{a,n_a}$, where
for every $i \in f^{-1}(a)$ and for every $1 \leq p < q \leq n_a$, the edge $e(B_{a,p},i)$ appears on $P_i$ before $e(B_{a,q},i)$. 

\paragraph{Solution closest to $\boldsymbol{s}$ and branching.}
The breaking of some of the bundles, performed in the previous paragraphs,
turned some soft arcs into crisp arcs.
Furthermore, we added crisp copies of all soft arcs that are not red. Both these operations potentially 
restricted the space of $st$-mincuts as crisp arcs are considered undeletable.
If, as a result, $\witnessflow$ is no longer an $st$-maxflow (i.e., $\lambda_G(s,t) > |\witnessflow|$ at this point), then we recurse on the current instance.
Otherwise, we observe that the $st$-mincut closes to $s$ plays an important role.
\begin{claim}\label{claim:clause-cut:X}
Let $X$ be the $st$-mincut closest to $s$. Then,
$$X = \{e(B_{f(i),1},i)~|~i \in [\lambda]\}.$$
\end{claim}
\begin{proofofclaim}
For every $a \in R_f$, let $X_a$ be an $st$-mincut that contains the edges $\{e(B_{a,1},i)~|~i \in f^{-1}(a)\}$.
(Such an $st$-mincut exists as $B_{a,1}$ is a candidate for $B^Z_a$ and, in particular, 
satisfies the last property of Definition~\ref{def:candidate}.)
Let $S_a$ be the set of vertices reachable from $s$ in $G-X_a$.
Then, by submodularity (Lemma~\ref{lem:intersect-cuts}), $X' := \{(x,y)\in E(G): x\in \bigcap_{a \in R_f} S_a, y\in V(G)\setminus \big ( \bigcap_{a \in R_f} S_a\big )\}$ is an $st$-mincut as well.

Note that, for $i \in [\lambda]$, the edge $e(B_{f(i),1},i)$ is the first soft edge on $P_i$ that does not have a crisp copy. 
Hence, the whole prefix of $P_i$ from $s$ to the tail of $e(B_{f(i),1},i)$ is in $S_a$ for every $a \in R_f$.
Furthermore, the suffix of $P_i$ from the head of $e(B_{f(i),1},i)$ to $t$ is not in $S_{f(i)}$. 
Consequently, $X' \supseteq \{e(B_{f(i),1},i)~|~i \in [\lambda]\}$. 
As $X'$ is an $st$-mincut, $|X'| = \lambda$, and thus
$X' = \{e(B_{f(i),1},i)~|~i \in [\lambda]\}$. 

Finally, as the edge $e(B_{f(i),1},i)$ is the first soft edge on $P_i$ that does not have a crisp copy for every $i \in [\lambda]$, for every $st$-mincut, its 
edge on $P_i$ is not earlier than $e(B_{f(i),1},i)$. 
Consequently, $X'$ is the $st$-mincut closest to $s$, as desired.
\end{proofofclaim}

The algorithm computes $X$, the $st$-mincut closest to $s$.

Note that the arcs of $X$ are contained in exactly $|R_f| \leq k$ bundles.
We check if $X$ is a feasible solution of cost at most $k$. If this is the case, we return that we are dealing with a \yes-instance.
Otherwise, regardless of whether $X$ is not a feasible solution (it violates a crisp clause) or is a feasible solution of cost more than $k$, 
there is a clause (crisp or soft) that is violated by $X$.

Fix one such clause $C$ and randomly branch into the following (at most $\maxarity+1$) directions.
If $C$ is soft, one option is that $C$ is violated by the sought solution $Z$; if we guess that this is the case, we delete the bundle containing $C$ and recurse.
In the remaining $|C| \leq \maxarity$ options, we guess a vertex $x \in C$ that is not reachable from $s$ in $G-Z$, add an arc $(x,t)$ and recurse. 
If such a vertex $x$ is guessed correctly, $Z$ remains a feasible solution after the arc $(x,t)$ is added. Furthermore, since $C$ is violated by $X$,
there is a path from $s$ to $x$ in $G-X$; since $X$ is the $st$-mincut closest to $s$, adding the arc $(x,t)$ strictly increases the value of maximum $st$-flow in the recusive call. 

This finishes the description of the recursive algorithm.

\medskip

In total, the success probability of all guesses made by a single recursive
call is $2^{\Oh(k^4 \maxarity^8 \log(k\maxarity))}$
(with the application of the flow-augmentation being the bottleneck). 
The depth of the recursion is bounded by $\Oh(k^2 \maxarity^2)$: 
the value of $k$ can only decrease and, between subsequent decreases of $k$,
the value of $\lambda_G(s,t)$ can increase up to the value of $k\maxarity^2$. 
The success probability promised in Theorem~\ref{thm:alg-clause-cut} follows.
This finishes the proof of Theorem~\ref{thm:alg-clause-cut}.
\end{proof}

Theorem~\ref{thm:alg-IS} follows immediately from Theorem~\ref{thm:alg-clause-cut} and
Lemma~\ref{lem:isd-to-graph}.


\section{Min SAT($\boldsymbol{\Gamma}$) dichotomy}
\label{sec:dichotomy}
\subsection{Statement of results}

We can now state the full classification of \MinSAT{\Gamma} and \WeightedMinSAT{\Gamma} into being either FPT or \classWone-hard. Throughout, and in the following, we consider parameterization by solution size $k$.
Our goal is to show Theorem~\ref{ithm:dichotomy}, restated below.

\thmdich*


\subsection{CSP preliminaries}

We need some additional terminology for CSPs.

Let $\Gamma$ be a Boolean constraint language and $R \subseteq \{0,1\}^r$ a Boolean relation.
A \emph{pp-definition}, or \emph{primitive positive definition} of $R$ over $\Gamma$
is a formula $\cF$ over $\Gamma \cup \{=\}$ on variables $X \cup Y$ such that
\[
  R(X) \equiv \exists Y \colon  \cF(X,Y).
\]
In other words, for every tuple $\alpha \in \{0,1\}^{|X|}$, $\alpha \in R$ if and only if
there exists $\beta \in \{0,1\}^{|Y|}$ such that $\cF(\alpha,\beta)$ is satisfied.
A \emph{qfpp-definition}, or \emph{quantifier-free primitive positive definition} of $R$ over $\Gamma$
is a pp-definition without existential quantification, i.e., $Y=\emptyset$.
The qfpp-definition of a relation provides fine-grained explicit information
about the relation, which we use frequently in the analysis below.

Let $\cF$ be a formula over a Boolean language $\Gamma$, on a variable set $V=V(\cF)$.
Recall that the \emph{cost} in $\cF$ of an assignment $\alpha \colon V \to \{0,1\}$ is the
number of constraints in $\cF$ that are false under $\alpha$. 
More formally, for a Boolean relation $R \subseteq \{0,1\}^r$
we define a \emph{cost function} $f_R \colon 2^r \to \mathbb{N}$ by
\[
  f_R(x_1,\ldots,x_r) =
  \begin{cases}
    0 & (x_1,\ldots,x_r) \in R \\
    1 & \textrm{otherwise}.
  \end{cases}
\]
We can then treat a (weighted or unweighted) \MinSAT{\Gamma} instance
as a sum of cost functions, i.e., for a formula $\cF$ over $\Gamma$
we define the cost function
\[
f_{\cF}(\alpha) = \sum_{R(x_1,\ldots,x_r) \in \cF} f_R(\alpha(x_1),\ldots,\alpha(x_r)).
\]
Note that for weighted problems, we distinguish between the \emph{cost} 
and the \emph{weight} of an assignment, i.e., the cost function $f_{\cF}$
is defined as above even if the input instance comes with a
set of weights $\weight$.

Let $f \colon 2^r \to \mathbb{N}$ be a cost function and $\Gamma$ a constraint language.
A \emph{proportional implementation} of $f$ over $\Gamma$ 
is a formula $\cF$ over $\Gamma$, on a variable set $V(\cF)=X \cup Y$ where $|X|=r$,
such that
\[
c \cdot f(X) = \min_Y f_{\cF}(X,Y)
\]
holds for every $X \in 2^r$, for some constant $c \in \mathbb{N}_+$. 
(The case when $c=1$ has been referred to as a \emph{strict and perfect implementation}
in earlier CSP literature~\cite{CreignouKSbook01,KhannaSTW00}.)
Informally, for a relation $R$ we will say \emph{proportional implementation of $R$} 
to mean a proportional implementation of $f_R$. 

Finally, recall that a \emph{crisp} constraint is a constraint whose
violation is forbidden. This is typically modeled in the CSP literature
as a cost function $f \colon 2^r \to \{0,\infty\}$, taking the value
$f(\alpha)=\infty$ for forbidden assignments $\alpha$. However, since
we are dealing with purely finite-valued cost functions, we will choose
another model. We say that \emph{\MinSAT{\Gamma} supports crisp constraints $R(X)$}
for a relation $R \subseteq \{0,1\}^r$ if, for every $k \in \mathbb{N}$,
there exists a formula $\cF$ over $\Gamma$ such that $f_{\cF}(\alpha)=0$
for any $\alpha \in R$ and $f_{\cF}(\alpha) > k$ otherwise. 

We summarize the immediate consequences of the above. 

\begin{proposition}\label{prop:reduction-basics}
  Let $\Gamma$ be a Boolean constraint language.  The following holds.
  \begin{enumerate}
  \item If $R$ has a pp-definition over $\Gamma$, then
    \MinSAT{\Gamma} supports crisp constraints $R(X)$
  \item If $R$ has a proportional implementation over $\Gamma$,
    then \MinSAT{\Gamma \cup \{R\}} has an FPT-reduction to \MinSAT{\Gamma}.
  \end{enumerate}
\end{proposition}
\begin{proof}
  The former is immediate, by repeating each pp-definition $k+1$
  times for a parameter value of $k$. For the latter, assume that $R$ has a proportional
  implementation with factor $c > 0$.  Replace every constraint $R(X)$ in
  the input by a proportional implementation, and duplicate every other
  constraint $R'(X)$, $R' \in \Gamma$, $c$ times in the output.
  Correctness is immediate. 
\end{proof}

For a Boolean relation $R \subseteq \{0,1\}^r$, the \emph{dual} of $R$
is the relation $\{(1-t[1], \ldots, 1-t[r]) \mid t \in R\}$
produced by taking item-wise negation of every tuple of $R$.
For a Boolean constraint language $\Gamma$,
the \emph{dual of $\Gamma$} is the language that
contains the dual of $R$ for every $R \in \Gamma$.

\begin{proposition} \label{prop:dual}
  Let $\Gamma$ be a finite Boolean constraint language
  and let $\Gamma'$ be its dual.
  Then \MinSAT{\Gamma} (respectively \WeightedMinSAT{\Gamma}) is FPT
  if and only if \MinSAT{\Gamma'} (respectively \WeightedMinSAT{\Gamma'}) is FPT.
\end{proposition}
\begin{proof}
  Since the operation of taking the dual is its own inverse,
  it suffices to show one direction. Let $\cF$ be a formula over $\Gamma$,
  and let $\cF'$ be the formula produced by replacing every constraint
  $R(X)$ in $\cF$ by the constraint $R'(X)$ where $R'$ is the dual of $R$.
  Then for every assignment $\phi \colon V(I) \to \{0,1\}$, and every
  constraint $R(X)$ in $\cF$, $\phi$ satisfies $R(X)$ if and only if
  the assignment $\phi' \colon v \mapsto 1-\phi(v)$
  satisfies $R'(X)$. Hence the problems are essentially
  equivalent and the reduction is immediate. 
\end{proof}

\subsubsection{Post's lattice of co-clones and the structure of Boolean languages}
\label{sec:post}
Our dichotomy depends on classical results on the structure of Boolean
languages, under the expressive power of pp-definitions. These have
frequently been used in complexity dichotomies over Boolean languages,
including the characterization of MinSAT problems with constant-factor
FPT approximations by Bonnet et al.~\cite{BonnetEM16ESA}. 
For more background, see a survey of Creignou and Vollmer~\cite{CreignouKV08post}.
See also surveys of Couceiro et al.~\cite{CouceiroHL22} and Barto et al.~\cite{BartoKW17}
for a broader perspective.
We review the definitions.
  
\begin{table}
  \centering
  \begin{tabular}{lll}
    Co-clone & Common name & Basis \\
    \hline
    IR$_2$ & Unary & $(x=1)$, $(x=0)$ \\
    ID$_2$ & Bijunctive & $(x \lor y)$, $(x \to y)$, $(\neg x \lor \neg y)$ \\
    IS$_{00}^n$ & IHS-B+ & $(x=0)$, $(x \to y)$, $(x_1 \lor \ldots \lor x_n)$ \\
    IS$_{10}^n$ & IHS-B- & $(x=1)$, $(x \to y)$, $(\neg x_1 \lor \ldots \lor \neg x_n)$ \\
    IS$_{00}$ & IHS+ & $(x=0)$, $(x \to y)$, $(x_! \lor \ldots \lor x_d)$ for all $d \in \N$ \\
    IS$_{10}$ & IHS- & $(x=1)$, $(x \to y)$, $(\neg x_1 \lor \ldots \lor \neg x_d)$ for all $d \in \N$ \\
    IE$_2$ & Horn & $(\neg x \lor \neg y \lor z)$, $(x=1)$, $(x=0)$\\
    IV$_2$ & Dual Horn & $(\neg x \lor y \lor z)$, $(x=1)$, $(x=0)$ \\
    IL$_3$ & Even-arity affine & $(x_1 \oplus x_2 \oplus x_3 \oplus x_4=b)$, $b=0, 1$\\
  \end{tabular}
  \caption{Explicit definition of co-clones relevant to the dichotomy result. The basis provided is up to pp-definitions.}
  \label{tab:coclones}
\end{table}

Let $\Gamma$ be a set of Boolean relations. The \emph{co-clone generated by $\Gamma$},
denoted $\langle \Gamma \rangle$, is the set of all relations that have a pp-definition over $\Gamma$. 
An explicit description of all Boolean co-clones was derived by Post~\cite{PostsLattice41},
known as \emph{Post's lattice of co-clones}.
A presentation suitable for our purposes is given by Böhler et al.~\cite{BohlerRSV05coclones}.
A \emph{plain basis} of a co-clone $C$ is a set of relations $B \subseteq C$ such that
every relation $R \in C$ has a qfpp-definition over $B$; see Creignou et al.~\cite{CreignouKZ08plain-coclone}.%
\footnote{In fact, the requirement for a plain basis is somewhat stronger: 
  There is a qfpp-definition of every relation $R \in C$ without using equality constraints $(x=y)$
  unless $=$ is present in $B$, and without repeating any variables in a constraint application.
  This makes no difference to our purposes, but it explains why, e.g., the plain basis
  of IS$_{00}^n$ contains the positive $d$-clause of every arity $d \leq n$, while
  up to qfpp-definitions we only need the positive clause of arity $n$. 
}
Not every co-clone has a finite plain basis, but the co-clones important in this paper do. 
We observe the mapping of the language classes considered in this paper to the co-clones, 
using the names used in the references above~\cite{BohlerRSV05coclones,CreignouKZ08plain-coclone};
see also~\cite[Figure~1]{BonnetEM16ESA}.
Explicit definitions of the most relevant co-clones are also given in Table~\ref{tab:coclones}.
\begin{itemize}
\item $\Gamma$ is bijunctive if and only if $\Gamma \subseteq \text{ID}_2$. 
  It has a plain basis of $\{(x=1), (x=0), (x \lor y), (x \to y), (\neg x \lor \neg y)\}$. 
\item $\Gamma$ is IHS-B+ (implicative hitting set, bounded, positive case) if and only if 
  $\Gamma \subseteq \text{IS}_{00}^n$ for some $n \in \N$.
  IS$_{00}^n$ has a plain basis of $(x=1)$, $(x=0)$, $(x \to y)$ and
    positive $d$-clauses $(x_1 \lor \ldots \lor x_d)$, $d \leq n$.
\item $\Gamma$ is IHS-B- (implicative hitting set, bounded, negative case) if and only if
  $\Gamma \subseteq \text{IS}_{10}^n$ for some $n \in \N$.
  IS$_{10}^n$ has a plain basis of $(x=1)$, $(x=0)$, $(x \to y)$ and
    negative $d$-clauses $(\neg x_1 \lor \ldots \lor \neg x_d)$, $d \leq n$.
  \end{itemize}
In addition, we need to discuss the following special cases.
\begin{itemize}
\item A language $\Gamma$ is \emph{0-valid} if every relation $R \in \Gamma$
  contains the all-0 tuple. The class of all 0-valid languages forms the co-clone II$_0$.
\item Similarly, $\Gamma$ is \emph{1-valid} if every relation $R \in \Gamma$
  contains the all-1 tuple. The class of all 1-valid languages forms the co-clone II$_1$.
\item A relation $R$ is \emph{self-dual} if $R$ equals its own dual, and a
  language $\Gamma$ is self-dual if every relation $R \in \Gamma$ is self-dual.
  The class of all self-dual languages forms the co-clone IN$_2$.
\item A language $\Gamma$ pp-defines $(x=0)$ and $(x=1)$ if and only if 
  $\text{IR}_2 \subseteq \langle \Gamma \rangle$
\end{itemize}
From this, we note the following by inspection of Post's lattice.

\begin{proposition} \label{prop:posts-lattice-relations}
  The following hold.
  \begin{enumerate}
  \item Every language that does not pp-define $(x=0)$ and $(x=1)$
    is either 0-valid, 1-valid or self-dual.
  \item Every language that is self-dual but not 0-valid or 1-valid pp-defines
    disequality constraints $(x \neq y)$.
  \item Every language that is IHS-B- but not bijunctive pp-defines negative 3-clauses
    $(\neg x \lor \neg y \lor \neg z)$.
  \item Every language that is bijunctive but not IHS pp-defines
    disequality constraints $(x \neq y)$.
  \end{enumerate}
\end{proposition}

\subsection{Initial structural observations}

We make several structural observations pertaining to the ability to implement constant assignments and to proportional implementations.

Throughout the section, to eliminate uninteresting cases we say that
\MinSAT{\Gamma} is \emph{trivial} if instances of the problem are either
always true, always false, always satisfied by setting all variables to 0 (0-valid),
or always satisfied by setting all variables to 1 (1-valid).
More precisely, \MinSAT{\Gamma} is trivial if every non-empty relation $R \in \Gamma$ is 0-valid,
or every non-empty relation $R \in \Gamma$ is 1-valid. 
The following is then a consequence of Proposition~\ref{prop:posts-lattice-relations}.

\begin{proposition} \label{prop:constants-or-neg}
  Assume that \MinSAT{\Gamma} is not trivial. 
  Then either $\Gamma$ pp-defines $(x=0)$ and $(x=1)$,
  or $\Gamma$ is self-dual and pp-defines $(x \neq y)$. 
\end{proposition}

Let us also observe that due to the special structure of these relations, a
pp-definition of $(x=0)$, $(x=1)$ or $(x \rightarrow y)$ implies
a proportional implementation. 

\begin{proposition} \label{prop:crisp-to-soft}
  Let $\Gamma$ be a Boolean constraint language. 
  For any relation $R \in \{(x=0),(x=1),(x \to y)\}$
  it holds that if $R$ has a pp-definition in $\Gamma$,
  then the corresponding relation also has a proportional  
  implementation in $\Gamma$ with factor $c=1$. 
\end{proposition}
\begin{proof}
  For each of these relations, there is only one non-satisfying
  assignment.  Hence any pp-definition is a proportional implementation
  with factor $c$, where $c$ is the cost of the
  pp-definition on the non-satisfying assignment.  
  Furthermore, if $c > 1$, then we may simply discard 
  constraints of the implementation until the minimum cost of an
  assignment extending the non-satisfying assignment to $R$
  has dropped to 1. 
\end{proof}

Finally, crisp constraints $(x \neq y)$ can be used to ``emulate''
assignment constraints $(x=0)$ and $(x=1)$ in a way compatible with FPT-reductions,
even if these relations are not pp-definable over $\Gamma$. 
We wrap up these observations in the following result.

\begin{lemma}\label{lemma:trivial:or:constants}\label{lemma:have-constants}
 Let $\Gamma$ be a Boolean language such that \MinSAT{\Gamma} is not trivial. Then \MinSAT{\Gamma\cup\{(x=0),(x=1)\}} has an FPT-reduction to \MinSAT{\Gamma},
 and \WeightedMinSAT{\Gamma \cup \{(x=0), (x=1)\}} has an FPT-reduction to \WeightedMinSAT{\Gamma}. 
\end{lemma}
\begin{proof}
By Propositions~\ref{prop:constants-or-neg} and~\ref{prop:crisp-to-soft}, for every Boolean constraint language $\Gamma$, 
either $\Gamma$ admits a proportional implementation of $(x=0)$ and $(x=1)$, or $\Gamma$ is self-dual and admits
a proportional implementation of $(x \neq y)$. In the former case we are done by Proposition~\ref{prop:reduction-basics}(2); in the latter case we can introduce two
global variables $z_0$ and $z_1$, with a crisp constraint $(z_0 \neq z_1)$, and replace every constraint $(x=0)$ (respectively $(x=1)$)
by $(x \neq z_1)$ (respectively $(x \neq z_0)$). Under self-duality we can then assume that $\alpha(z_0)=0$ and $\alpha(z_1)=1$,
which makes the replacement valid. Indeed, if $\alpha(z_0)=1$ then simply replace $\alpha(x) \gets 1-\alpha(x)$ for every variable $x$.
Since $\Gamma$ is self-dual, this change does not change the set of satisfied constraints. 
Note that every step in these claims works with constraint weights as well.
\end{proof}

\subsection{Hard cases}\label{ss:hard-cases}

The required hard cases go back to a hardness proof by Marx and Razgon~\cite{MarxR09}. To make this work self-contained and set up clear notation, we begin by proving \classWone-hardness of an auxiliary parameterized graph cut problem, defined below, which also underlies previous hardness proofs.

\begin{quote}
 \textbf{Problem:} \PMstCl\\ 
 \noindent \textbf{Input:} A directed acyclic graph $D=(U,A)$, vertices $s,t\in U$, an integer $\ell\in\N$, a pairing $\A$ of the arcs in $A$, and a partition of $A$ into a set $\P$ of $2\ell$ arc-disjoint $st$-paths.\\ 
 \noindent \textbf{Parameter:} $\ell$.\\ 
 \noindent \textbf{Question:} Is there an $st$-cut that consists of the arcs in at most $\ell$ pairs from $\A$?
\end{quote}

\PMstCl is purposefully restrictive to simplify later reductions to other problems. Hardness comes from the pairing $\A$ of arcs alone but the additional structure will be very handy. First of all, minimum $st$-cuts contain at least $2\ell$ arcs. Moreover, every minimum $st$-cut $B\subseteq A$ cuts every path of $\P$ in exactly one arc so that the resulting parts stay reachable from $s$ respectively reaching $t$ in $D-B$. Furthermore, this shows that flow augmentation cannot make the problem any easier.

\begin{lemma}\label{lemma:hardness:tightpairedstcut}
 \PMstCl is \classWone-hard.
\end{lemma}

\begin{proof}
 We give a reduction from \MCCk, which is well known to be \classWone-hard (cf.~\cite{the-book}). Let an instance $(G,k)$ of \MCCk be given where $G=(V_1,\ldots,V_k;E)$ and each $V_i$ is an independent set in $G$; we seek a $k$-clique with exactly one vertex from each set $V_i$. By padding with isolated vertices we may assume that each set $V_i$ has exactly $n$ vertices, and we rename these to $V_i=\{v_{i,1},\ldots,v_{i,n}\}$. In the following, we create an instance $(D,s,t,\ell,\A,\P)$ of \PMstCl and afterwards prove equivalence to $(G,k)$.
 
 The arc set $A$ of $D=(U,A)$ consists of two arcs $a_{p,q}$ and $a_{q,p}$ per (undirected) edge $\{p,q\}\in E$; for ease of presentation, the endpoints of $a_{p,q}$ and $a_{q,p}$ will be specified later. The pairing $\A$ of the arcs in $A$ simply contains all unordered pairs $\{a_{p,q},a_{q,p}\}$, exactly one for each $\{p,q\}\in E$. We let $\ell=\binom{k}{2}$ so that $2\ell=k(k-1)$. We now construct $2\ell=k(k-1)$ paths $P_{i,j}$, one for each choice of $i,j\in[k]$ with $i\neq j$; suitably identifying certain vertices of these paths will then complete the construction of $D=(U,A)$.
 
 For given $i,j\in[k]$ with $i\neq j$, the path $P_{i,j}$ is obtained by concatenating all arcs $a_{p,q}$ where $p\in V_i$ and $q\in V_j$. The arcs $a_{p,q}$ are ordered by first subscript according to the numbering of $V_i=\{v_{i,1},v_{i,2},\ldots,v_{i,n}\}$, i.e., $a_{v_{i,1},\ast},\ldots,a_{v_{i,1},\ast},a_{v_{i,2},\ast},\ldots,a_{v_{i,n-1},\ast},a_{v_{i,n},\ast},\ldots,a_{v_{i,n},\ast}$; the order of the second coordinates $q$ is irrelevant. Intuitively, deleting some arc $a_{p,q}$ will correspond to choosing $p$ and $q$ for the $k$-clique in $G$, though all these choices need to be enforced to be consistent of course. We let $\P$ contain all $2\ell$ paths $P_{i,j}$. Clearly, each arc in $A$ is contained in exactly one path of $\P$.
 
 Let us now identify certain vertices of different paths $P_{i,j}$ in $\P$ that have the same subscript $i$: In each path $P_{i,\ast}$ and for each $r\in[n]$, there is a first arc $a_{p,q}$ with $p=v_{i,r}$; identify the heads of all these arcs into a single vertex called $u_{i,r}$. The relevant structure of a path $P_{i,j}$ is the ordering of the vertices $u_{i,\ast}$ and the ordering of arcs by first subscript $u_{i,1},a_{v_{i,1},\ast},\ldots,a_{v_{i,1},\ast},u_{i,2},a_{v_{i,2},\ast},\ldots,a_{v_{i,n-1},\ast},u_{i,n},a_{v_{i,n},\ast},\ldots,a_{v_{i,n},\ast}$. Finally, we identify the first vertex of every path in $\P$ into a vertex called $s$, and the final vertex of each path into a vertex called $t$. Clearly, we obtain a directed graph $D=(U,A)$ with $s,t\in U$ and with $\P$ a partition of the arc set $A$ into $2\ell$ arc-disjoint paths. It is also easy to see that $D$ is acyclic, because all paths in $\P$ agree on the ordering of the shared vertices $s$, $t$, and $u_{\ast,\ast}$. This completes the construction of the instance $(D,s,t,\ell,\A,\P)$ of \PMstCl. For convenience, we will use $u_{1,1},u_{2,1},\ldots,u_{k,1}$ as aliases for $s$; these were identified into $s$. Similarly, we will use $u_{1,n+1},u_{2,n+1},\ldots,u_{k,n+1}$ as aliases for $t$.
 
 Because the identification of vertices leaves us with a set $\P$ of $2\ell$ arc-disjoint $st$-paths, the size of a minimum $st$-cut in $D$ is at least $2\ell$. Because the paths partition the edge set, the minimum $st$-cut size is in fact exactly $2\ell$ as one could e.g.\ cut the $2\ell$ edges leaving $s$. More generally, the only necessary synchronization of cutting the $2\ell$ paths comes from the shared vertices $u_{i,r}$, where $i\in[k]$ and $r\in[n]$. Concretely, for each $i\in[k]$, the paths $P_{i,\ast}$ share the vertices $s=u_{i,1},u_{i,2},\ldots,u_{i,n},u_{i,n+1}=t$ and to get a minimum $st$-cut it is necessary and sufficient to cut all of them between some choice of $u_{i,r}$ and $u_{i,r+1}$.
 
 To prove equivalence of $(G,k)$ and $(D,s,t,\ell,\A,\P)$ it suffices to check that the pairing $\A$ of arcs in $A$ restricts our choice to such minimum $st$-cuts that correspond to a $k$-clique in $G$. Concretely, assume that $(G,k)$ is a \yes-instance, and let $r_1,r_2,\ldots,r_k\in[n]$ such that $v_{1,r_1},v_{2,r_2},\ldots,v_{k,r_k}$ is a $k$-clique in $G$. It suffices to check that $B=\bigcup_{i,j}\{a_{v_{i,r_i},v_{j,r_j}}, a_{v_{j,r_j},v_{i,r_i}}\}$ cuts all paths in $\P$ in the required way (as outlined above): First of all, $B$ contains exactly one arc on each path $P_{i,j}$, namely $a_{v_{i,r_i},v_{j,r_j}}$. Second, for each $i\in[k]$, all paths $P_{i,\ast}$ are cut between $u_{i,r_i}$ and $u_{i,r_i+1}$ because the relevant arcs $a_{p,q}$ in $B$ all have the same first subscript $p=v_{i,r_i}$. It follows that $(D,s,t,\ell,\A,\P)$ is a \yes-instance.
 
 For the converse, assume that $D$ contains an $st$-cut $B\subseteq A$ that consists of the arcs in at most $\ell$ pairs from $\A$. It follows immediately that $B$ contains exactly $2\ell$ arcs, one from each path $P_{i,j}$ with $i,j\in[k]$ and $i\neq j$, and that it is a minimum $st$-cut that is the union of exactly $\ell$ pairs from $\A$. By the structure of minimum $st$-cuts in $D$, we know that for each $i\in[k]$, all paths $P_{i,\ast}$ must be cut between some $u_{i,r_i}$ and $u_{i,r_i+1}$ for some $r_i\in[n]$. (Recall that $u_{i,1}=s$ and $u_{i,n+1}=t$.) Using these values $r_1,r_2,\ldots,r_k\in[n]$, we claim that $v_{1,r_1},v_{2,r_2},\ldots,v_{k,r_k}$ is a $k$-clique in $G$. To this end, we show that for all $i,j\in[k]$ with $i\neq j$ the edge $\{v_{i,r_i},v_{j,r_j}\}$ exists in $G$: Let us consider how $B$ cuts paths $P_{i,j}$ and $P_{j,i}$. The cuts happen between $u_{i,r_i}$ and $u_{i,r_i+1}$ respectively between $u_{j,r_j}$ and $u_{j,r_j+1}$. Crucially, $P_{i,j}$ contains all arcs $a_{p,q}$ with $p\in V_i$ and $q\in V_j$, while $P_{j,i}$ contains the corresponding arcs $a_{q,p}$ and $\A$ pairs all these arcs as $\{a_{p,q},a_{q,p}\}$. Thus, $B$ must intersect the two paths exactly in the two arcs of some pair $\{a_{p^*,q^*},a_{q^*,p^*}\}$. By construction, as $a_{p^*,q^*}$ lies between $u_{i,r_i}$ and $u_{i,r_i+1}$ on $P_{i,j}$, we must have $p^*=v_{i,r_i}$. Similarly, as $a_{q^*,p^*}$ lies between $u_{j,r_j}$ and $u_{j,r_j+1}$ on $P_{j,i}$, we must have $q^*=v_{j,r_j}$. Again by construction, it follows directly that $\{p^*,q^*\}=\{v_{i,r_i},v_{j,r_j}\}$ is an edge of $G$. This completes the proof of equivalence. 
 
 Clearly this reduction can be performed in polynomial time and the parameter value $\ell=\binom{k}{2}$ is bounded by a function of $k$. This completes the proof of \classWone-hardness.
\end{proof}

We are now set up to prove the required hardness results for \MinSAT{\Gamma} and \WeightedMinSAT{\Gamma} by reductions from \PMstCl. 
Compared to the preliminary version of this paper~\cite{csp-soda} we present a single, unified proof written in more abstract terms; we then delegate the case work over various constraint languages $\Gamma$ to proving that the conditions of this hardness proof can be met. 

We introduce some definitions. Let a \emph{double implication} be a 4-ary relation $R^*$ satisfying 
\begin{align*}
  (a=b)\wedge (c=d) \implies R^*(a,b,c,d) \implies (a\rightarrow b) \wedge (c\rightarrow d),
\end{align*}
i.e., $\{(0,0,0,0),(0,0,1,1),(1,1,0,0),(1,1,1,1)\} \subseteq R^*$ and
$R^*$ contains no tuples $t$ where $(t[1],t[2])=(1,0)$ or $(t[3],t[4])=(1,0)$.
Double implications are the canonical relations that make MinSAT \classWone-hard. 
We show a condition under which they can be implemented. 

\begin{lemma} \label{lemma:neg-and-gr}
  Let $\Gamma$ be a Boolean constraint language that supports crisp
  $(x \neq y)$ and let $R \in \Gamma$ be a relation such that
  the Gaifman graph $G_R$ contains a $2K_2$. Then $\Gamma$
  proportionally implements a double implication $R^*$.
\end{lemma}
\begin{proof}
  Let $R$ be $\ell$-ary, where $\ell \geq 4$, and without loss of generality assume that $\{1,2,3,4\}$ induces a $2K_2$ in $G_R$, with edges $\{1,2\}$ and $\{3,4\}$.
  It follows that there are values $\alpha_1,\alpha_2,\alpha_3,\alpha_4\in\{0,1\}$ such that $R(x_1,x_2,x_3,x_4,\ldots,x_\ell)$ implies $(x_1,x_2)\neq(\alpha_1,\alpha_2)$ and $(x_3,x_4)\neq(\alpha_3,\alpha_4)$.
  Construct the formula
  \[
    \cF(X)=R(x_1^*,x_2^*,x_3^*,x_4^*, x_5,\ldots,x_\ell) \land (x_1 \neq x_1') \land (x_2 \neq x_2') \land (x_3 \neq x_3') \land (x_4 \neq x_4'),
  \]
  on variable set $X=\{x_1,x_1',\ldots,x_\ell\}$,
  where $x_i^*=x_i$ if $\alpha_i$ matches the $i$:th position of $(1,0,1,0)$, and $x_i^*=x_i'$ otherwise. 
  We claim that this is a proportional implementation of a double implication.
  More precisely,
  \[
    R^*(x_1,\ldots,x_4) \equiv \exists_{x_1',\ldots,x_4',x_5,\ldots,x_\ell} \cF
  \]
  is a double implication, and furthermore every assignment to $\{x_1,\ldots,x_4\}$ can be extended to an assignment to $X$
  such that every constraint of $\cF$ except $R$ is satisfied. (In particular, every use of $\neq$ in $\cF$ may be assumed to be crisp.)
  It follows that this is a proportional implementation with constant $c=1$. 
  
  We show that $R^*$ is indeed a double implication. One direction is easy: By assumption, 
  $\cF$ implies that $(x_1^*,x_2^*) \neq (\alpha_1,\alpha_2)$, which by construction implies
  that $(x_1,x_2) \neq (1,0)$, i.e., $\cF$ implies $(x_1 \to x_2)$. Similarly, $\cF$ implies $(x_3 \to x_4)$.
  Thus $R^*(a,b,c,d) \implies (a \to b) \land (c \to d)$.
  It remains to check that $(0,0,0,0)$, $(0,0,1,1)$, $(1,1,0,0)$, and $(1,1,1,1)$ are all contained in $R^*$, i.e., that $R^*(a,b,c,d)$ is implied by $(a=b)\wedge (c=d)$. For this, we observe that the Gaifman graph $G_{R^*}$ remains a $2K_2$ with edges $\{1,2\}$ and $\{3,4\}$.
  Indeed, the only possible difference between the projection of $R$ and $R^*$ to positions $\{1,2,3,4\}$ is a negation of some coordinates,
  which clearly preserves the Gaifman graph. The argument can now be completed as follows. 
  Since $G_{R^*}$ has no edge $\{1,3\}$, there is a tuple $t \in R^*$ such that $t[1]=t[3]=1$.
  Since $R^*$ implies $(x_1 \to x_2)$ and $(x_3 \to x_4)$ we must have $t=(1,1,1,1)$.
  Similarly, there is a tuple $t \in R$ such that $t[1]=1$ and $t[4]=0$, which can only be $t=(1,1,0,0)$;
  a tuple $t \in R$ such that $t[2]=0$ and $t[3]=1$, which can only be $t=(0,0,1,1)$;
  and a tuple $t \in R$ such that $t[2]=t[4]=0$, which can only be $t=(0,0,0,0)$. 
  This completes the argument that $(a=b)\wedge (c=d)\implies R^*(a,b,c,d)$.

  Finally, we get a proportional implementation of $R^*(a,b,c,d)$ with factor $c=1$, since (as noted) every infeasible assignment to $a, b, c, d$ 
  can be trivially extended to a satisfying assignment if the constraint over $R$ is deleted in $\cF$, giving cost $1$.
\end{proof}

In addition, we present a mapping of the variables along a path $P \in \P$ to Boolean variables
that allows us an implicit encoding of $(x \neq y)$-constraints. 
Let $n \in \N$. A \emph{complementary chain formula (of length $n$)} is a formula $\F$ 
on variable set $X=\{s,t,x_1,x_1',\ldots,x_n,x_n'\}$ such that the minimum-cost assignments
for $\F$ are precisely the assignments $\alpha_i$, $i \in \{0,\ldots,n\}$
defined as $\alpha_i(s)=1$, $\alpha_i(t)=0$, $\alpha_i(x_j')=1-\alpha_i(x_j)$ for $j  \in [n]$ and
\[
  \alpha_i(x_j) =
  \begin{cases}
    1 & j \leq i \\
    0 & j > i
  \end{cases}
\]
otherwise. $\F$ has \emph{bounded cost} if the cost $\min_\alpha f_{\F}(\alpha)$ is independent of $n$.
Finally, $\F$ is \emph{efficiently constructible} if it can be constructed in polynomial time in $n$.

These conditions allow us to prove hardness.

\begin{lemma} \label{lm:unified-hardness}
  Let $\Gamma$ be a Boolean language that allows efficiently constructible complementary chain formulas with bounded cost
  for every length $n \geq 0$, and let $R \in \Gamma$ be a relation such that the Gaifman graph $G_R$ contains a $2K_2$. 
  Then \MinSAT{\Gamma} is \classWone-hard. 
\end{lemma}
\begin{proof}
  We show the result by an FPT-reduction from \PMstCl. 
  To this end, let $(D,s,t,\ell,\A,\P)$ be an instance of \PMstCl with $D=(U,A)$. Recall that $\P$ partitions $A$ into $2\ell$ arc-disjoint $st$-paths.
  For each vertex $v\in U$, we create two variables $x_v$ and $x'_v$ with the intention of enforcing $x_v=1-x'_v$ for all minimum-cost assignments, so that $x'_v$ acts like the negation of $x_v$.
  Let $X=\{x_v\mid v\in U\}$ and $X'=\{x'_v\mid v\in U\}$.
  For every path $P \in \P$, say $P=(s=v_0,v_1,\ldots,v_r,t=v_{r+1})$, we create $\ell+1$ disjoint copies of a complementary chain formula $\F_P$ of length $r$,
  using variables from the common variable set $X \cup X'$, and we let $\F_0$ be the disjoint union of all such formulas.
  Let $c$ be the cost of a complementary chain formula and set the parameter to $k=2\ell (\ell+1)c + \ell$. 
  Finally, we handle the pairing $\A$. By Lemma~\ref{lemma:neg-and-gr}, there is a proportional implementation
  of a double implication $R^*$, using crisp constraints $(x \neq y)$ and a single application of $R$.
  Following Lemma~\ref{lemma:neg-and-gr}, without loss of generality let the implementation of $R^*$ be
  \[
    R^*(a,b,c,d) = \exists_{a', b', c', d', y_5, \ldots, y_t} R(a^*,b^*,c^*,d^*,y_5,\ldots,y_r) \land
    (a \neq a') \land (b \neq b') \land (c \neq c') \land (d \neq d'),
  \]
  where $v^* \in \{v,v'\}$ for $v=a,b,c,d$. For every pair $\{(p,q),(u,v) \in \A\}$, 
  we add a constraint $R(x_p^*,x_q^*,x_u^*,x_v^*,y_5, \ldots, y_r)$ to the output, where $y_i$, $i \geq 5$ is a fresh local variable
  and each variable $x_p, x_p', \ldots, x_v, x_v'$ in the constraint is taken accordingly from the variable set $X \cup X'$.
  Note that we do not output the $\neq$-constraints. 
  (For example, if the implementation is $R^*(a,b,c,d) \equiv \exists_{a',b',c',d',y} R(a,b,c',d',y) \land \ldots$ 
  then in the output, from a pair $\{(p,q),(u,v)\}$ we create a constraint $R(x_p,x_q,x_u',x_v',y)$ where only $y$ is a local variable.)
  The order $(p,q),(u,v)$ is chosen arbitrarily, i.e.,
  we output either an implementation of $R^*(x_p,x_q,x_u,x_v)$ or $R^*(x_u,x_v,x_p,x_q)$, but not both. 
  Clearly this construction can be completed in polynomial time and has a parameter $k$ bounded in terms of $\ell$.
  We show correctness of the reduction.
  
  First assume that the input instance is a \yes-instance of \PMstCl
  and let $B\subseteq A$ be a minimum $st$-cut of cardinality $2\ell$, consisting of the union of exactly $\ell$ pairs $\B\subseteq\A$.
  Necessarily, $B$ cuts each $st$-path in $\P$ exactly once, splitting its vertices into reachable and not reachable from $s$ in $D-B$ (and equivalently not reaching vs.\ reaching $t$ in $D-B$). We create an assignment $\alpha$ for $X\cup X'$ as follows: Set $\alpha(x_v)=1$ if $v$ is reachable from $s$ in $D-B$, otherwise $\alpha(x_v)=0$, and set $\alpha(x'_v)=1-\alpha(x_v)$ for every $v \in U$. Let $P=(s=v_0,\ldots,t=v_{r+1})$ be a path in $\P$, and assume that $(v_i,v_{i+1}) \in B$ for some $0 \leq i \leq r$. Then $\alpha$ restricted to the variables used in $\F_P$ matches the assignment $\alpha_i$
  in the definition of a complementary chain formula, hence $\F_P$ contains precisely $c$ false constraints under $\alpha$. 
  Thus $\alpha$ violates precisely $2\ell(\ell+1) c=k-\ell$ constraints in $\F_0$. 
  In addition, consider an application of $R$ corresponding to a pair $\{(p,q),(u,v)\} \in \A$. 
  Then either $\alpha(x_p)=\alpha(x_q)$ and $\alpha(x_u)=\alpha(x_v)$ or $\{(p,q),(u,v)\} \in \B$. 
  Since the latter occurs for only $\ell$ constraints, assume the former. 
  Note that by design, the constraint over $R$ corresponds to an implementation of a double implication $R^*$,
  since we have $\alpha(x_v')=1-\alpha(x_v)$ for every $v \in U$. Thus for every constraint over $R$
  that does not correspond to a pair of $\B$, there is an extension of $\alpha$ to the local variables of $R$
  that satisfies the constraint. Thus at most $\ell$ constraints outside of $\F_0$ are violated, 
  and the assignment $\alpha$, suitably extended to all local variables of constraints over $R$,
  has cost at most $k$ and $(\F,k)$ is a \yes-instance. 

  In the other direction, let $\alpha$ be an assignment that violates at most $k=2\ell(\ell+1)c+\ell$ constraints of $\F$
  Since every assignment violates at least $k-\ell$ constraints of $\F_0$, and since every path $P$ corresponds to $\ell+1$
  disjoint complementary chain formulas in $\F_0$ over the same variable set, we conclude that for every path $P$, 
  the restriction of $\alpha$ to the variables of $\F_P$ corresponds to one of the assignments $\alpha_i$
  in the definition of a complementary chain formula. Thus we let $B \subseteq A$ be the arcs $(u,v) \in A$
  such that $\alpha(x_u)=1$ and $\alpha(x_v)=0$ and note that $B$ contains precisely one arc for every path $P \in \P$.
  This also implies that $B$ is an $st$-cut in $D$. Indeed, let $S \subseteq U$ be the set of vertices $v \in U$
  such that $\alpha(x_v)=1$. Then $s \in S$, $t \notin S$ and $B$ contains every arc $(u,v)$ with $u \in S$ and $v \notin S$. 
  Additionally, $|B|=2\ell$ since $B$ contains only one arc per path and every arc of $A$ occurs on a path in $\P$.
  We argue that furthermore $B$ is the union of at most $\ell$ pairs from $\A$. For this, consider a pair $\{(p,q),(u,v)\} \in \A$
  where $(p,q) \in B$. Then $\F$ contains a constraint $R(x_p^*,x_q^*,x_u^*,x_v^*,\ldots,y_r)$ or $R(x_u^*,x_v^*,x_p^*,x_q^*,\ldots,y_r)$
  such that this constraint, together with negations $(x_p \neq x_p') \land \ldots \land (x_v \neq x_v')$, implements a double
  implication $R^*(x_p,x_q,x_u,x_v)$. Furthermore, $\alpha(x_z') \neq \alpha(x_z)$ holds for every $z \in U$. 
  Thus the constraint $R$ implies $(x_u \to x_v)$ and $(x_p \to x_q)$ and $(p,q) \in B$ implies that $R$ is violated.
  Therefore, for every pair $\{(p,q),(u,v)\} \in \A$ that contains an arc of $B$, the corresponding $R$-constraint in $\F$
  is violated by $\alpha$. It follows from the budget bound $k$ that there are at most $\ell$ such pairs.  
  Therefore the input is a \yes-instance of \PMstCl.
\end{proof}

It now remains to prove for various languages $\Gamma$ that \MinSAT{\Gamma} meets the conditions of Lemma~\ref{lm:unified-hardness}.
For the hardness proof, complementing the positive cases of Theorems~\ref{ithm:gaifman} and~\ref{ithm:arrow}, we need to cover the following cases. 
\begin{enumerate}
\item $\Gamma$ is not IHS-B (but may be bijunctive), and there is a relation
  $R \in \Gamma$ such that the Gaifman graph $G_R$ contains a $2K_2$
\item $\Gamma$ is not bijunctive (but may be IHS-B), and there is a relation
  $R \in \Gamma$ such that the arrow graph $H_R$ contains a $2K_2$
\item $\Gamma$ may be both bijunctive and ISH-B, but there are relations
  $R, R' \in \Gamma$ such that the Gaifman graph $G_R$ and the
  arrow graph $H_{R'}$ both contain $2K_2$'s
\end{enumerate}
Due to the structure of Boolean languages (cf.~Proposition~\ref{prop:posts-lattice-relations}), in the first case we may assume that $\Gamma$ supports crisp $(x \neq y)$-constraints,
and in the second case we may assume that $\Gamma$ supports crisp constraints $(\neg x \lor \neg y \lor \neg z)$.
We proceed with the constructions. The first is the easiest.

\begin{lemma} \label{lemma:hardness:gaifman}
  Let $\Gamma$ be a Boolean constraint language
  that supports crisp constraints $(x \neq y)$
  and let $R \in \Gamma$ be a relation such that
  the Gaifman graph $G_R$ contains a $2K_2$. 
  Then \MinSAT{\Gamma} is \classWone-hard.
\end{lemma}
\begin{proof}
  By Lemma~\ref{lemma:neg-and-gr}, $\Gamma$ has a proportional implementation of a double implication $R^*(a,b,c,d)$.
  Then in particular, $\Gamma$ has a proportional implementation of $R_2(a,b)=\exists_{c,d} R^*(a,b,c,d)$,
  which from the definition of a double implication implies that $R_2$ is either soft equality $(a=b)$
  or soft implication $(a \to b)$. The construction of a complementary chain formula is now easy. 
  Note that \MinSAT{\Gamma} is not trivial since it pp-defines $(x \neq y)$; thus by Lemma~\ref{lemma:trivial:or:constants}
  and Proposition~\ref{prop:reduction-basics} it suffices to show hardness for \MinSAT{\Gamma \cup \{(x=0),(x=1),R_2\}}.
  We thus construct a complementary chain formula for a variable set $\{s,t,x_1,x_1',\ldots,x_n,x_n'\}$
  as
  \[
    (s=1) \land R_2(s,x_1) \land \ldots \land R_2(x_n,t) \land (t=0) \land \bigwedge_{i=1}^n (x_i \neq x_i').
  \]
  Clearly this meets the definition, since it is unsatisfiable and the assignments $\alpha_i$, $i \in [n]$
  corresponds precisely to satisfying assignments after the deletion of a single $R_2$-constraint.
  Note that all $(x \neq y)$-constraints can be taken to be crisp. 
  Thus the language $\Gamma \cup \{(x=0),(x=1),R_2\}$ meets all the conditions of Lemma~\ref{lm:unified-hardness}
  and \MinSAT{\Gamma} is \classWone-hard. 
\end{proof}

When disequality constraints $(x \neq y)$ are not available, we need a somewhat more intricate construction of the domain encoding (i.e., complementary chain formulas). 
We show that it suffices that there is a relation $R$ available such that the arrow graph $H_R$ contains a $2K_2$. 

\begin{lemma} \label{lm:support:negpath}
 Let $\Gamma$ be a Boolean constraint language 
 that supports crisp constraints $(x=1)$ and $(x=0)$,
 and let $R \in \Gamma$ be a relation such that the arrow graph $H_R$ contains a $2K_2$. 
 Then $\Gamma$ allows efficiently constructible, bounded cost complementary chain formulas. 
\end{lemma}

\begin{proof}
  We first show a consequence similar to Lemma~\ref{lemma:neg-and-gr}, showing that $\Gamma$ proportionally implements a useful relation $R^*$. 

  \begin{claim}
    $\Gamma$ has a proportional implementation of a 4-ary relation $R^*$ with
    \begin{align*}
      (a=b\neq c=d) \implies R^*(a,b,c,d)\implies (a\rightarrow b)\wedge (c\rightarrow d),
    \end{align*}
    i.e., $(0,0,1,1), (1,1,0,0) \in R^*$ and $R^* \subseteq \{(0,0),(0,1),(1,1)\} \times \{(0,0),(0,1),(1,1)\}$.    
  \end{claim}
  \begin{proofofclaim}
    Let $R\in\Gamma$ such that $H_R$ contains a $2K_2$ as an induced subgraph. Concretely and without loss of generality, let $R$ be $\ell$-ary, with $\ell\geq 4$, and let $\{1,2,3,4\}$ induce a $2K_2$ in $H_R$ with (at least) the directed edges $(1,2)$ and $(3,4)$ being present. (Recall that existence of a $2K_2$ is checked in the underlying undirected, simple graph of $H_R$. Thus, we may also have edges $(2,1)$ and/or $(4,3)$ present but no other edges exist in $H_R[\{1,2,3,4\}]$.) We implement $R^*$ by
    \begin{align*}
      R^*(a,b,c,d):=\exists x_5,\ldots,x_\ell: R(a,b,c,d,x_5,\ldots,x_\ell),
    \end{align*}
    which is clearly a proportional implementation. It remains to check that $R^*$ includes respectively excludes the relevant tuples. To this end, we observe that the arrow graph $H_{R^*}$ of $R^*$ is (canonically) isomorphic to $H_R[\{1,2,3,4\}]$: The existence of edges in the arrow graph $H_R$ of a relation $R$ depends solely on the projection of tuples in $R$ to any two of its positions, and the above pp-implementation of $R^*$ is exactly a projection of $R$ to its first four positions, preserving projection to any two of them. Thus, $\{1,2,3,4\}$ induces the same $2K_2$ subgraph in $H_{R^*}$ with edges $(1,2)$ and $(3,4)$, possibly their reversals $(2,1)$ and/or $(4,3)$, but no further edges being present.
    By the definition of $H_{R^*}$, the edges $(1,2)$ and $(3,4)$ require that $R^*(a,b,c,d)$ has no satisfying assignment with $a=1$ and $b=0$ nor with $c=1$ and $d=0$. Thus, $R^*(a,b,c,d)$ indeed implies both $(a\rightarrow b)$ and $(c\rightarrow d)$, i.e., it excludes all tuples of form $(1,0,\cdot,\cdot)$ and/or $(\cdot,\cdot,1,0)$. 
 
    Now, assume for contradiction that $(1,1,0,0)\notin R^*$. In order for $(1,2)$ to be present in $H_{R^*}$ there must be some tuple $(1,1,\cdot,\cdot)\in R^*$, as some satisfying assignment must extend $a=b=1$. As $(1,1,0,0),(1,1,1,0)\notin R^*$, we must have some tuple $(1,1,\cdot,1)\in R^*$. Similarly, because we need a tuple $(\cdot,\cdot,0,0)$ for $(3,4)$ to be present in $H_{R^*}$ while $(1,0,0,0),(1,1,0,0)\notin R^*$, we must have some tuple $(0,\cdot,0,0)\in R^*$. Consider now which assignments to $a$ and $d$ can be extended to satisfying assignments of $R^*(a,b,c,d)$, i.e., the projection $R^*_{ad}(a,d):=\exists b,c: R^*(a,b,c,d)$: We have $(1,1),(0,0)\in R^*_{ad}$ because $(1,1,\cdot,1),(0,\cdot,0,0)\in R^*$. However, we have $(1,0)\notin R^*_{ad}$ because all tuples $(1,\cdot,\cdot,0)$ are excluded from $R^*$: concretely, $(1,0,0,0)$,$(1,0,1,0))$, and $(1,1,1,0)$ were excluded in the previous paragraph, while $(1,1,0,0)\notin R^*$ is our local assumption. But then $(1,1),(0,0)\in R^*_{ad}$ and $(1,0)\notin R^*_{ad}$ would imply that $(1,4)$ must be an edge of $H_{R^*}$ (and of $H_R$); a contradiction. It follows that $(1,1,0,0)\in R^*$.
 
    We briefly discuss the symmetric argument for inclusion of $(0,0,1,1)$ in $R^*$. Assume for contradiction that $(0,0,1,1)\notin R^*$; we also have $(0,0,1,0),(1,0,1,1)\notin R^*$. For edge $(1,2)$ to be in $H_{R^*}$ we require some tuple $(0,0,\cdot,\cdot)\in R^*$, which implies that there is a tuple $(0,0,0,\cdot)\in R^*$. For edge $(3,4)$ we require some tuple $(\cdot,\cdot,1,1)\in R^*$, which implies that there is a tuple $(\cdot,1,1,1)\in R^*$. Thus, $(0,0),(1,1)\in R^*_{bc}$, where $R^*_{bc}(b,c):=\exists a,d: R^*(a,b,c,d)$. At the same time $(0,1)\notin R^*_{bc}$ because all tuples $(\cdot,0,1,\cdot)$ are excluded from being in $R^*$. Thus, the edge $(2,3)$ must be present in $H_{R^*}$ (and in $H_R$); a contradiction. It follows that $(0,0,1,1)\in R^*$.
 
    To summarize, we get a proportional implementation of $R^*(a,b,c,d)$ where $(0,0,1,1),(1,1,0,0)\in R^*$ while all tuples of form $(1,0,\cdot,\cdot),(\cdot,\cdot,1,0)$ are excluded from being in $R^*$, i.e., $R^*(a,b,c,d)$ requires $(a\rightarrow b)$ and $(c\rightarrow d)$ while $a=b\neq c=d$ is sufficient for $R^*(a,b,c,d)$ to hold.
  \end{proofofclaim}

  We now proceed with the construction. We create a formula $\F$ over $\Gamma$ on variable set
  $V(\F)=\{s,x_1,\ldots,x_n,x_1',\ldots,x_n',t\}$ such that every assignment violates at least one constraint and in addition the following hold.
  \begin{enumerate}
  \item For every $0 \leq i \leq n$, let $\alpha_i$ be the assignment that sets
    $\alpha_i(s)=1=\alpha_i(x_j)$ for every $1 \leq j \leq i$,
    $\alpha_i(x_j)=0=\alpha_i(t)$ for every $j > i$, and $\alpha_i(x_j')=1-\alpha_i(x_j)$
    for every $j \in [n]$. This assignment violates precisely one constraint of $\cF$.
  \item No other assignment violates fewer than two constraints of $\cF$.
  \end{enumerate}
  Clearly this meets the conditions of a complementary chain formula, with cost $1$.
  
  Construct a formula $\cF$ on variable set $V(\cF)=\{s,x_1,\ldots,x_n,x_1',\ldots,x_n',t\}$
  consisting of crisp constraints $(s=1)$ and $(t=0)$ (e.g., two copies of each constraint) and one copy each of the constraints $R^*(s,x_1,x_1',t)$, $R^*(x_i,x_{i+1},x_{i+1}',x_i')$ for $1 \leq i \leq n-1$, and $R^*(x_n,t,s,x_n')$. We show that $\cF$ meets the conditions of the lemma.

  We consider the structure of $\cF$ first. Recall that every constraint $R^*(a,b,c,d)$ implies $(a \to b)$ and $(c \to d)$. Therefore, the constraints of $\cF$ imply two disjoint chains of implications, one $s \to x_1 \to \ldots \to x_n \to t$ and the other $s \to x_n' \to \ldots \to x_1' \to t$. Since $\cF$ also implies $s=1$ and $t=0$, this shows that $\cF$ is not satisfiable.
  Furthermore, informally, the two chains are arc-disjoint, so every assignment $\alpha$ must violate at least two implications, and every constraint $R^*(x_i,x_{i+1},x_{i+1}',x_i')$ contributes only two implications $(x_i \to x_{i+1})$ and $(x_{i+1}' \to x_i')$. Therefore, any assignment $\alpha$ that violates only a single constraint must break these chains in coordinated positions $(x_i \to x_{i+1})$ and $(x_{i+1}' \to x_i')$ coming from the same constraint $R^*(x_i,x_{i+1},x_{i+1}',x_i')$, and under such an assignment we have $\alpha(x_j)=1$ if and only if $j \leq i$, and $\alpha(x_j')=1$ if and only if $j>i$, from which it follows that $\alpha=\alpha_i$. 

  More formally, we first consider the assignment $\alpha_i$ for some $0 \leq i \leq n$, which for all $j, j' \in [n]$ with $j \leq i < j'$ sets $\alpha_i(s)=\alpha_i(x_j)=1$, $\alpha_i(x_{j'})=\alpha_i(t)=0$, and $\alpha(x_j')=1-\alpha(x_j)$ for every $j \in [n]$. We show that $\alpha_i$ satisfies all but one constraint of $\cF$.  Clearly, the constraints $(s=1)$ and $(t=0)$ are satisfied by $\alpha_i$. 
  For the constraints $R^*(s,x_1,x_1',t)$ (if $i>0$) and $R^*(x_j,x_{j+1},x_{j+1}',x_j')$ for $j < i$,
  $\alpha_i$ gives assignment $(1,1,0,0)$, and since $(1,1,0,0) \in R^*$, these constraints are satisfied. Similarly, for any constraint $R^*(x_j,x_{j+1},x_{j+1}',x_j')$ for $j > i$, $\alpha_i$ gives assignment $(0,0,1,1)$ which also satisfies $R^*$. Finally, $R^*(x_n,t,s,x_n')$ is satisfied if and only if $i<n$. Hence every constraint of $\cF$ except $R^*(x_i,x_{i+1},x_{i+1}',x_i')$ (respectively $R^*(s,x_1,x_1',t)$ if $i=0$ and $R^*(x_n,t,s,x_n')$ if $i=n$) is satisfied by $\alpha_i$.
  
  Finally, assume that $\alpha$ is an assignment that only violates one constraint. Then as outlined above, if the chain $s \to x_1 \to \ldots \to x_n \to t$ of implications is broken in more than one place, then $\alpha$ violates two distinct constraints $R^*$, since every implication of the chain is implied by a distinct constraint $R^*$ of $\cF$. Similarly, all implications of the chain $s \to x_n' \to \ldots \to x_1' \to t$ are implied by distinct constraints $R^*$ of $\cF$. Hence $\alpha$ violates precisely one arc $(x_i \to x_{i+1})$ and one arc $(x_{j+1}' \to x_j')$, taking $s=x_0=x_{n+1}'$ and $t=x_{n+1}=x_0'$ for brevity. Furthermore, if $\alpha$ has cost 1, then these implications must be contributed by the same constraint $R^*(x_i,x_{i+1},x_{j+1}',x_j')$ which is possible only if $i=j$. 
  Thus $\alpha$ is constrained to be identical to $\alpha_i$ for some $i \in [n]$. 
\end{proof} 

This allows us to finish off the last two cases of the hardness conditions.

\begin{lemma}\label{lemma:hardness:arrow}
  Let $\Gamma$ be a Boolean constraint language that supports crisp $(\neg x_1\vee \neg x_2 \vee \neg x_3)$ and let $R \in \Gamma$ be a relation such that the arrow graph $H_R$ contains a $2K_2$.
 Then \MinSAT{\Gamma} is either trivial or \classWone-hard. 
\end{lemma}

\begin{proof}
  Assume that \MinSAT{\Gamma} is not trivial.  
  Using Lemma~\ref{lemma:trivial:or:constants}, it suffices to give a reduction from \PMstCl to \MinSAT{\Gamma\cup\{(x=0),(x=1)\}}. 
  By Lemma~\ref{lm:support:negpath}, $\Gamma$ supports efficiently constructible, bounded cost complementary chain formulas,
  hence it remains to present a relation $R$ whose Gaifman graph $G_R$ contains a $2K_2$. 
  For this, consider the implementation
  \[
    R(a,b,c,d) := \exists_y (\neg a \lor \neg b \lor \neg y) \land (\neg c \lor \neg d \lor \neg y) \land (y=1),
  \]
  where the negative 3-clauses may be crisp and $(y=1)$ is soft. Clearly, this is an proportional implementation of the relation
  \[
    R(a,b,c,d) \equiv (\neg a \lor \neg b) \land (\neg c \lor \neg d),
  \]
  whose Gaifman graph is a $2K_2$. Thus \MinSAT{\Gamma \cup \{(x=0),(x=1),R\}} is \classWone-hard,
  hence \MinSAT{\Gamma} is \classWone-hard as well.
\end{proof}  

Finally, we cover the case where we do not have access to crisp disequalities or crisp negative 3-clauses but have both a relation $R$ whose Gaifman graph $G_R$ contains a $2K_2$ and a relation $R'$ whose arrow graph $H_{R'}$ contains a $2K_2$. 
This case is immediate given the above preparations. 

\begin{lemma}\label{lemma:hardness:arrowandgaifman}
  Let $\Gamma$ be a Boolean constraint language that contains relations $R$ and $R'$ such that the Gaifman graph $G_{R}$ and the arrow graph $H_{R'}$ 
  both contain a $2K_2$. Then \MinSAT{\Gamma} is either trivial or \classWone-hard. 
\end{lemma}

\begin{proof}
  Assume that \MinSAT{\Gamma} is not trivial. By Lemma~\ref{lemma:trivial:or:constants},
  we may assume access to crisp constraints $(x=0)$ and $(x=1)$. 
  Then we have efficiently constructible, bounded cost complementary chain formulas by Lemma~\ref{lm:support:negpath},
  using $(x=0)$, $(x=1)$ and the relation $R'$, 
  and \classWone-hardness follows from Lemma~\ref{lm:unified-hardness} from the chain formulas and the relation $R$. 
\end{proof}

The final hardness proofs pertain to \WeightedMinSAT{\Gamma}. These build on constructions used so far, but we present them separately to avoid mixing weighted and unweighted problems. With weighted constraints, a kind of ``weighted complementary chain formula'' can be constructed using only (proportional implementations of) $(x\rightarrow y)$ or $(x=y)$ together with crisp negative clauses $(\neg x\vee\neg y)$; as in Lemma~\ref{lm:support:negpath}, we use this to implement variable negation. 

The construction of weighted complementary chains is as follows. 

\begin{lemma} \label{lm:support:wtpath}
  Let $\Gamma$ be a Boolean constraint language that proportionally
  implements constraints $(x \to y)$ or $(x=y)$ and supports crisp
  constraints $(x=1)$, $(x=0)$ and $(\neg x \lor \neg y)$. 
  Let $V=(s,x_1,\ldots,x_n,t)$ be a sequence of variables and
  let $W \in \N$ be a target weight, $W>n$.
  There is a formula $\cF$ over $\Gamma$ 
  on variable set $V(\cF)=\{s,x_1,\ldots,x_n,x_1',\ldots,x_n',t\}$
  with the following properties.
  \begin{enumerate}
  \item Every assignment violates at least two constraints of $\F$.
  \item For every $0 \leq i \leq n$, let $\alpha_i$ be the assignment
    that sets
    $\alpha_i(s)=1=\alpha_i(x_j)$ for every $1 \leq j \leq i$,
   $\alpha_i(x_j)=0=\alpha_i(t)$ for every $j > i$, and $\alpha_i(x_j')=1-\alpha_i(x_j)$
   for every $j \in [n]$. This assignment violates precisely two
   constraint of $\cF$ and has weight $W$.
 \item Every other assignment either violates more than two
   constraints of $\cF$ or has weight at least $W+1$. 
 \end{enumerate}
\end{lemma}
\begin{proof}
  It suffices to deal with the case that $\Gamma$ proportionally implements $(x=y)$: Indeed, given proportional implementation of $(x\rightarrow y)$, we can use $(x=y)\equiv(x\rightarrow y)\wedge(y\rightarrow x)$ as a proportional implementation of $(x=y)$, since every assignment violating $(x=y)$ violates precisely one of the constraints $(x\rightarrow y)$ and $(y\rightarrow x)$, and both violations have the same cost $\alpha\geq 1$ (with $\alpha$ depending on the implementation of $(x \to y)$). For ease of description, we will describe the reduction assuming that we have $(x=y)$ constraints available (rather than implementations thereof); as these will be the only soft constraints, replacing them by the implementation (while blowing up the parameter value by a factor of $\alpha$) completes the proof.

We create a formula with the following constraints. For convenience, let $x_0$ and $x_{n+1}'$ denote $s$, and $x_{n+1}$ and $x_0'$ denote $t$. 
 \begin{itemize}
 \item For each $i\in\{0,1,\ldots,n\}$ we add a constraint $(x_{i}=x_{{i+1}})$ of weight $W-i$. 
  \item For each $i\in\{0,1,\ldots,n\}$ we add a constraint $(x'_{i}=x'_{{i+1}})$ of weight $i$. 
  \item In addition, add crisp constraints $(s=1)$, $(t=0)$, and for every $1 \leq i \leq n$ a crisp constraint $(\neg x_i \lor \neg x_i')$. 
  \end{itemize}
  Note that these constraints form two chains of equality constraints between $s$ and $t$, one in increasing order of index ($s$, $x_1$, \ldots, $x_n$, $t$) via variables $x_i$ and one in the opposite order $(s$, $x_n'$, \ldots, $x_1'$, $t$) via variables $x_i'$. Hence it is immediate that $\cF$ is unsatisfiable, and that every assignment violates at least two constraints. We proceed to show the claims about assignment weights.

  First, consider the assignment $\alpha_i$ for some $0 \leq i \leq n$. 
  Since $\alpha_i(x_j')=1-\alpha_i(x_j)$ for every $j \in [n]$, clearly
  $\alpha_i$ satisfies all crisp constraints. Furthermore, it violates
  precisely two soft constraints, the edge $(x_i=x_{i+1})$ on the first
  path and the edge $(x_i'=x_{i+1}')$ on the second path; together they have weight $W-i+i=W$. We proceed to show that all other assignments are worse.

  Let $\alpha$ be an assignment that violates precisely two constraints in $\cF$, satisfies all crisp constraints, and is not equal to $\alpha_i$ for any $0 \leq i \leq n$. Then there is precisely one index $i$ such that $\alpha(x_i) \neq \alpha(x_{i+1})$ (taking $x_0=s$ and $x_{n+1}=t$)
  and precisely one index $j$ such that $\alpha(x_j') \neq \alpha(x_{j+1}')$ (taking $x_0'=t$ and $x_{n+1}'=s$), and by assumption $i \neq j$. 
  If $i>j$, then $\alpha(x_i)=\alpha(x_i')=1$, contradicting the constraint $(x_i \lor \neg x_i')$. Thus $i < j$. But then $\alpha$ has weight $W-i+j \geq W+1$, as claimed. 
\end{proof}

We now give the main weighted hardness proof. 

\begin{lemma}\label{lemma:hardness:weighted2}
 Let $\Gamma$ be a Boolean constraint language that contains a relation $R$ whose Gaifman graph $G_R$ contains a $2K_2$, that supports crisp $(\neg x \vee \neg y)$, and proportionally implements $(x=y)$ or $(x \to y)$. Then \WeightedMinSAT{\Gamma} is either trivial or \classWone-hard.
\end{lemma}
\begin{proof}
  Assume that \WeightedMinSAT{\Gamma} is not trivial. By Lemma~\ref{lemma:trivial:or:constants},
  it suffices to give a reduction from \PMstCl to \WeightedMinSAT{\Gamma\cup\{(x=0),(x=1)\}}.  
  Let $(D,s,t,\ell,\A,\P)$ be an instance of \PMstCl with $D=(U,A)$. Again, $\P$ partitions $A$ into $2\ell$ arc-disjoint $st$-paths.
  Define variable sets $X=\{x_v \mid v \in U\}$ and $X'=\{x_v' \mid v \in U\}$.
  For every $P \in \P$, use Lemma~\ref{lm:support:wtpath} with target weight $W'=n=|U|$
  to create a formula $\F_P$ using variables from the common set $X \cup X'$. Let $\F_0$ be the disjoint union
  of $\ell+1$ copies of the formulas $\F_P$ for every path $P \in \P$. Set $k=4\ell(\ell+1)+\ell$ and $W=2\ell(\ell+1)n+\ell$.
  By Lemma~\ref{lm:support:wtpath}, any assignment $\alpha$ violates at least $2(\ell+1)$ constraints
  per path $P \in \P$, and any assignment which violates only that number of constraints for the path
  has a weight of at least $(\ell+1)n$ for the path, for a total minimum
  of $4\ell(\ell+1)$ violated constraints to a weight of at least $2\ell(\ell+1)n$.
  This leaves a remaining budget of at most $\ell$ violated constraints, to a weight of
  at most $\ell$, in the rest of $\F$. Furthermore, if any assignment $\alpha$
  is used such that for some path $P \in \P$, $\alpha$ does not conform to one
  of the assignments $\alpha_i$ of Lemma~\ref{lm:support:wtpath}, then, by the
  duplication, $\alpha$ will break either the violation budget $k$ or the weight
  budget $W$. Hence we may assume that any assignment relevant to us
  satisfies $\alpha(x')=1-\alpha(x)$ for every variable $x \in X$, hence
  we may proceed as if we have crisp constraints $(x \neq x')$ in our formula.   
  We will use this fact in constructions over $R$, defining a formula $\F'$
  that, subject to the assumption $(x \neq x')$, encodes the pairs of $\A$.

  By Lemma~\ref{lemma:neg-and-gr}, there is a proportional implementation of a double implication using the relation
  $R \in \Gamma$ and crisp constraints $(x \neq y)$. As in Lemma~\ref{lm:unified-hardness},
  we then use the structure imposed on min-cost, min-weight solutions by Lemma~\ref{lm:support:wtpath} to replace any variables in
  $\neq$-constraints in the implementation by variables $x'$.
  Hence, we proceed as if we have a proportional implementation of $R^*$,
  and create $\F'$ by, for every pair $\{(u,v),(p,q)\} \in \A$,
  adding an implementation of the constraint $R^*(x_u,x_v,x_p,x_q)$. Give these
  constraints weight 1 each. Let $\F=\F_0 \cup \F'$. This finishes the construction.

  On the one hand, assume that the input is a \yes-instance, and let $B \subseteq A$
  be an $st$-mincut in $D$ which is the union of $\ell$ pairs from $\A$. 
  Let $\B \subseteq \A$ be the pairs present in $B$. Define $\alpha$
  by letting $\alpha(x_v)=1$ for $v \in U$ if and only if $v$ is reachable
  from $s$ in $D-B$, and $\alpha(x_v')=1-\alpha(x_v)$ for every $v \in U$.
  For every path $P \in \P$, $\alpha$ restricted to $V(\F_P)$ conforms
  to one of the assignments $\alpha_i$ of Lemma~\ref{lm:support:wtpath},
  hence for every path $P \in \P$, $\alpha$ violates precisely two constraints,
  to a weight of $n$, for each copy of $\F_P$. 
  Hence $\alpha$ violates $4\ell(\ell+1)=k-\ell$ constraints, to a weight
  of $(\ell+1)n(2\ell)=W-\ell$, from $\F_0$. Furthermore, for every pair
  $\{(u,v),(p,q)\} \in \B$, $\alpha$ sets $\alpha(x_u)=\alpha(x_p)=1$ and
  $\alpha(x_v)=\alpha(x_q)=0$, violating one $R^*$-constraint, of weight 1,
  per pair. For every pair $\{(u,v),(p,q)\} \in \A \setminus \B$, $\alpha$
  sets $\alpha(x_u)=\alpha(x_v)$ and $\alpha(x_p)=\alpha(x_q)$,
  thereby satisfying the corresponding $R^*$-constraint.
  Since $\alpha(x')=1-\alpha(x)$ for every variable $x$ by construction,
  the formulas used in implementing the constraints $R^*$ behave
  as required of the implementation, and in total $\alpha$ 
  violates precisely $k$ constraints,
  to a total weight of precisely $W$.

  On the other hand, let $\alpha$ be an assignment that violates at
  most $k=4\ell(\ell+1)+\ell$ constraints of $\F$, to a total weight of at
  most $W=(\ell+1)n(2\ell)+\ell$. By Lemma~\ref{lm:support:wtpath},
  for each path $P$, $\alpha$ must violate at least 2 constraints to
  a weight of at least $n$ in each copy of the formula $\F_P$ for each
  path $P \in \P$. Hence $\alpha$ violates at least $4\ell(\ell+1)$
  constraints to a weight of at least $(\ell+1)n(2\ell)$ from $\F_0$ alone.
  Furthermore, by the duplication, if there is a path $P \in \P$
  such that $\alpha$ restricted to $V(\F_P)$ does not conform
  to one of the assignments $\alpha_i$ of Lemma~\ref{lm:support:wtpath},
  then $\alpha$ violates either at least $\ell+1$ further constraints,
  or violates constraints to a further weight of at least $\ell+1$, both
  of which would contradict $(\F,k,W)$ being a \yes-instance.
  Hence we may in particular assume $\alpha(x')=1-\alpha(x)$
  for every variable $x$, implying that the formulas added to $\F'$
  act as proper proportional implementations of the constraint $R^*$
  as above. Hence by the budget bound, there are at most $\ell$
  pairs $\{(u,v),(p,q)\} \in \A$ such that $\alpha$ 
  violates a constraint $R^*(x_u,x_v,x_p,x_q)$. 
  The proof now finishes as for Lemma~\ref{lm:unified-hardness}
  to conclude that $\alpha$ corresponds to an $st$-cut $B$ in $D$
  consisting of $\ell$ pairs $\B \subseteq \A$. 
\end{proof}

Finally, we handle the case with a non-bijunctive weighted language, similarly to Lemma~\ref{lemma:hardness:arrow}.

\begin{lemma}\label{lemma:hardness:weighted}
 Let $\Gamma$ be a Boolean constraint language that proportionally implements $(x\rightarrow y)$ or $(x=y)$ and that supports crisp $(\neg x\vee \neg y \vee \neg z)$. Then \WeightedMinSAT{\Gamma} is either trivial or \classWone-hard.
\end{lemma}

\begin{proof}
  Assume that \WeightedMinSAT{\Gamma} is not trivial. Then by Lemma~\ref{lemma:trivial:or:constants}, it suffices to show hardness for \WeightedMinSAT{\Gamma \cup \{(x=0),(x=1)\}}. 
  Furthermore, as in Lemma~\ref{lemma:hardness:arrow} we can use the crisp negative 3-clause and soft constraints $(x=1)$
  for a proportional implementation of the relation
  \[
    R(a,b,c,d) \equiv (\neg a \lor \neg b) \land (\neg c \lor \neg d),
  \]
  whose Gaifman graph is a $2K_2$. Now \WeightedMinSAT{\Gamma \cup \{(x=1),(x=0),R\}} meets all the conditions of Lemma~\ref{lemma:hardness:weighted2} 
  and hardness of \WeightedMinSAT{\Gamma} follows from Lemma~\ref{lemma:trivial:or:constants} and Proposition~\ref{prop:reduction-basics}. 
\end{proof}

This finishes our list of hardness reductions.

\subsection{Dichotomy top-level case distinction}

We are now ready to show that our results so far form a dichotomy
for both \MinSAT{\Gamma} and \WeightedMinSAT{\Gamma}
as being FPT or \classWone-hard. We first present our top-level case distinction. 
The distinction follows the structure of Post's lattice of co-clones~\cite{PostsLattice41},
as reviewed in Section~\ref{sec:post}. See also the illustration of Bonnet et al.~\cite[Fig.~1]{BonnetEM16ESA}
regarding cases with constant-factor FPT approximations (FPA); naturally, the languages
for which \MinSAT{\Gamma} is FPT are contained in co-clones for which \MinSAT{\Gamma} is FPA.

Below, for a constraint language $\Gamma$ and a family of relations $\Gamma'$,
we say that $\Gamma$ \emph{efficiently pp-defines} $\Gamma'$ if,
for every relation $R \in \Gamma'$ of arity $n$, a pp-definition of $R$
over $\Gamma$ can be constructed in time polynomial in $n$. 


\begin{lemma} \label{lemma:top-level-cases}
  Assume that $\Gamma$ is finite and that \MinSAT{\Gamma} is not
  trivial. Then one of the following holds.
  \begin{enumerate}
  \item $\Gamma$ can efficiently pp-define positive clauses of every arity
    \label{case:hard1}.
  \item $\Gamma$ can efficiently pp-define negative clauses of every arity 
    \label{case:hard2}.
  \item $\Gamma$ can efficiently pp-define all even-arity linear equations
    $(x_1 \oplus \ldots \oplus x_{2r} = b)$, $b \in \{0,1\}$ over GF(2)
    \label{case:hard3}.
  \item $\Gamma$ is qfpp-definable
    over $\{(x=0), (x=1), (x \lor y), (x \to y), (\neg x \lor \neg y)\}$,
    i.e., bijunctive
    \label{case:id2}.
  \item $\Gamma$ is qfpp-definable over
    $\{(x=0), (x=1), (x \to y)\}$ and positive clauses $(x_1 \lor \ldots \lor x_d)\}$
    for some $d \in \mathbb{N}$, 
    i.e., $\Gamma$ is IHS-B+.
    \label{case:is00-group}
  \item $\Gamma$ is qfpp-definable over
    $\{(x=0), (x=1), (x \to y)\}$ and negative clauses $(\neg x_1 \lor \ldots \lor \neg x_d)\}$
    for some $d \in \mathbb{N}$,
    i.e., $\Gamma$ is IHS-B-.
    \label{case:is10-group}
  \end{enumerate}
\end{lemma}
\begin{proof}
  The result follows from the structure of Post's lattice of co-clones
  and the notion of a plain basis of a co-clone. Specifically,
  cases~\ref{case:id2}--\ref{case:is10-group} above correspond
  precisely to plain bases of the co-clones ID$_2$, IS$_{00}^d$
  and IS$_{10}^d$, respectively~\cite{CreignouKZ08plain-coclone},
  as surveyed in Section~\ref{sec:post}. 
  Hence it remains to consider languages $\Gamma$ not contained in
  any of these co-clones.

  Let $\Gamma$ be such a language and let $C=\langle \Gamma \rangle$
  be the co-clone generated by $\Gamma$. 
  By a simple inspection of Post's lattice, $C$ contains either
  IS$_{00}^n$ for every $n \in \N$, or IS$_{10}^n$ for every $n \in \N$,
  or IL$_3$. In the former two cases, $\Gamma$ can pp-define all
  negative clauses respectively all positive clauses. 
  In the latter case, $\Gamma$ can pp-define all linear equations
  $(x_1 \oplus \ldots \oplus x_r=b)$, $b \in \{0,1\}$ of even arity over GF(2)~\cite{BohlerRSV05coclones}.
  We finally note that the pp-definitions can also be made efficient.
  In particular, $C$ does not equal IS$_{00}$ or IS$_{10}$ since $\Gamma$ is finite 
  and these co-clones have no finite basis. Thus $C$ contains IE$_2$, IV$_2$ or IL$_3$
  and there are finite-sized pp-definitions over $\Gamma$
  of, respectively, a ternary Horn clause $(\neg x \lor \neg y \lor z)$,
  a ternary dual Horn clause $(x \lor y \lor \neg z)$,
  or the 4-ary linear equations $(x_1 \oplus \ldots \oplus x_4=b)$, $b \in \{0,1\}$,
  and it is trivial to chain a linear number of such constraints 
  to create the corresponding constraints of arbitrary arity.
  In the former two cases, there are also pp-definitions of $(z=0)$
  and $(z=1)$, completing the pp-definitions. 
\end{proof}

We note that cases~\ref{case:hard1}--\ref{case:hard3} of this lemma
define hard problems. The first two are immediate.

\begin{lemma} \label{lemma:hard-HS}
  Assume that $\Gamma$ efficiently pp-defines positive or negative
  clauses of every arity. Then \MinSAT{\Gamma} is either trivial or
  W[2]-hard. 
\end{lemma}
\begin{proof}
  We focus on the case of positive clauses, as the other case is
  dual. By Lemma~\ref{lemma:trivial:or:constants} we may assume that
  we have soft constraints $(x=0)$ and $(x=1)$. We also note by
  Proposition~\ref{prop:reduction-basics} that \MinSAT{\Gamma} supports
  crisp positive clauses. There is now an
  immediate FPT-reduction from \textsc{Hitting Set}.
  Let the input be a hypergraph $\mathcal{H}=(V,\mathcal{E})$
  and an integer $k$. The task is to find a set $S \subseteq V$
  such that $|S| \leq k$ and $S$ hits every hyperedge $E \in \mathcal{E}$.
  We create a \MinSAT{\Gamma} instance as follows. Let the variable set be $V$. 
  For every hyperedge $E=\{v_1,\ldots,v_r\} \in \mathcal{E}$, 
  we add to the output the crisp clause $(v_1 \lor \ldots \lor v_r)$.
  For every variable $v \in V$, we add a soft clause $(x=0)$.
  Finally, we set the parameter to $k$. It is now clear
  that an assignment which satisfies all crisp clauses,
  and where at most $k$ soft clauses are violated,
  is precisely a hitting set for $\mathcal{H}$ of cardinality at most $k$.
\end{proof}

For the case of linear equations, we use results of Bonnet et al.~\cite{BonnetELM-arXiv,BonnetEM16ESA} and Lin~\cite{Lin18JACM}.

\begin{lemma} \label{lemma:il-hard}
  Assume that $\Gamma$ efficiently pp-defines linear equations over GF(2) of every even arity. Then \MinSAT{\Gamma} is \classWone-hard.
\end{lemma}
\begin{proof}
  Note that \MinSAT{\Gamma} is by assumption not trivial, since it defines $(x=0)$ and $(x=1)$. 
  Then by Lemma~\ref{lemma:have-constants}, \MinSAT{\Gamma \cup \{(x=0),(x=1)\}}
  FPT-reduces to \MinSAT{\Gamma}.
  Bonnet et al.~\cite{BonnetELM-arXiv} (preliminary version in~\cite{BonnetEM16ESA})
  showed that in this case, \MinSAT{\Gamma \cup \{(x=0),(x=1)\}} does not 
  even have an FPT-time constant-factor approximation unless $\classFPT=\classWone$,
  thus certainly \MinSAT{\Gamma} is \classWone-hard as a decision problem.
\end{proof}

\subsection{The bjiunctive cases}

Case~\ref{case:id2} is easily handled by Theorem~\ref{thm:alg-ID2}
and Lemma~\ref{lemma:hardness:gaifman}.

\begin{lemma} \label{lemma:dich-id2}
  Let $\Gamma$ be a finite Boolean language contained in ID$_2$
  but not contained in IS$_{00}$ or IS$_{10}$. Then the following apply.
  \begin{enumerate}
  \item If there is a relation $R \in \Gamma$ such that the Gaifman graph $G_R$
    contains a $2K_2$, then \MinSAT{\Gamma} is \classWone-hard.
  \item Otherwise \WeightedMinSAT{\Gamma} is FPT.
  \end{enumerate}
\end{lemma}
\begin{proof}
  We begin by observing that $\Gamma$ supports crisp $(x \neq y)$. Indeed,
  if $\Gamma$ is contained in ID$_2$ but not in IS$_{10}^2$ or IS$_{00}^2$,
  then as noted in Proposition~\ref{prop:posts-lattice-relations},
  $\Gamma$ spans the co-clone ID~\cite{BohlerRSV05coclones}
  and can pp-define $(x \neq y)$. By Proposition~\ref{prop:reduction-basics},
  \MinSAT{\Gamma} then supports crisp $(x \neq y)$. 
  Hence if there is a relation $R \in \Gamma$ such that the Gaifman graph $G_R$
  contains a $2K_2$, then Lemma~\ref{lemma:hardness:gaifman} applies
  and \MinSAT{\Gamma} is \classWone-hard (note that $(x \neq y)$ excludes the case
  that \MinSAT{\Gamma} is trivial).
  
  Otherwise $\Gamma \subseteq \idtwopositive$ for some finite arity bound $\maxarity$,
  and \WeightedMinSAT{\Gamma} is FPT by Theorem~\ref{thm:alg-ID2}. 
\end{proof}

\subsection{The IHS cases}

We now consider cases~\ref{case:is00-group} and \ref{case:is10-group}
of Lemma~\ref{lemma:top-level-cases}. Since these cases are each others' duals
it suffices to consider case~\ref{case:is10-group}.
We therefore assume that every relation in $\Gamma$ can be qfpp-defined in the language
$\Gamma_d=\{(x=0), (x=1), (x \to y), (\neg x_1 \lor \ldots \lor \neg x_d)\}$
for some $d \in \mathbb{N}$. 

We begin with a support observation. 

\begin{proposition} \label{prop:hr-gives-eq} 
  Let $R \in \Gamma$ be a relation.  If $H_R$ contains an edge,
  then $(x=y)$ has a proportional implementation over $R$. 
  On the other hand, if $H_R$ is edgeless, then $R$ has a qfpp-definition over $\Gamma_d$
  using only assignments and negative clauses, i.e., 
  without using any clauses $(x \to y)$ or $(x=y)$,
  and $\Gamma$ cannot proportionally implement $(x=y)$. 
\end{proposition}
\begin{proof}
  Let $R \subseteq \{0,1\}^r$ for some $R \in \Gamma$ and assume
  w.l.o.g.\ that $H_R$ has a directed edge on $(1,2)$. Consider the implementation
  $f'(x,y)=\min_{x_3,\ldots,x_r} f_R(x,y,x_3,\ldots,x_r)$ where $f_R$ is the cost function of $R$. 
  Clearly, this is a proportional implementation of the cost function $f_{R'}$ 
  derived from the relation $R'(x,y) \equiv \exists_{x_3,\ldots,x_r} R(x,y,x_3,\ldots,x_r)$. 
  Then by the definition of $H_R$ we have $(0,0), (1,1) \in R'$
  and $(1,0) \notin R'$, i.e., $R'(x,y)$ is either $(x \to y)$
  or $(x=y)$. In the latter case we are done, in the former case
  we may use $R'(x,y) \land R'(y,x)$ as a proportional implementation
  of $(x=y)$. 

  For the second part, let $H_R$ be empty and consider a qfpp-definition
  of $R$ via a formula $\F$ that uses a minimum number of implication and equality clauses.
  Assume that this definition contains a clause $(u \to v)$ (as $(u=v)$
  can be defined as $(u \to v) \land (v \to u)$).
  By assumption, removing this clause yields an incorrect definition, 
  i.e., if $\F'$ is $\F$ with the clause $(u \to v)$ removed, 
  then $\F'$ qfpp-defines a relation $R' \supset R$
  and there is a tuple $t \in R' \setminus R$ where $t[u]=1$ and $t[v]=0$.
  Also by assumption $t$ cannot be eliminated by introducing constraints
  $(u=0)$ or $(v=1)$. Hence there are tuples $t_1, t_2 \in R$
  with $t_1[v]=0$, implying $t_1[u]=0$, and with $t_2[u]=1$, implying $t_2[v]=1$.
  Then there is an edge $(u,v)$ in $H_R$, contrary to assumption. 
\end{proof}

We begin with the case $d>2$ since $d=2$ requires some special
attention. We begin by noting (by Lemma~\ref{lemma:hardness:weighted})
that \WeightedMinSAT{\Gamma} is mostly \classWone-hard.

\begin{lemma} \label{lemma:isd-weighted}
  Let $\Gamma$ be a Boolean language contained in IS$_{10}^d$ for some $d>2$
  but not in IS$_{10}^2$. Then the following apply.
  \begin{enumerate}
  \item If $\Gamma$ proportionally implements $(x=y)$, then \WeightedMinSAT{\Gamma}
    is either trivial or \classWone-hard.  
  \item Otherwise \WeightedMinSAT{\Gamma} is FPT with a running time of $2^{\Oh(k)}$. 
  \end{enumerate}
\end{lemma}
\begin{proof}
  Recall that every relation $R \in \Gamma$ has a qfpp-definition over $\Gamma_d$,
  and first consider the case that for every relation $R \in \Gamma$, the arrow graph $H_R$
  is edgeless. Then by Proposition~\ref{prop:hr-gives-eq} $R$ has a qfpp-definition over $\Gamma_d$
  without using any clauses $(x=y)$ or $(x \to y)$, i.e., using only negative clauses
  and assignments. In this case, \WeightedMinSAT{\Gamma} allows for a simple branching algorithm. 
  Let $I$ be an instance of \WeightedMinSAT{\Gamma} with parameter $k$. Assume that $I$ is not satisfiable,
  as otherwise $I$ is trivially a \yes-instance. There are now only two ways for $I$ to be unsatisfiable.
  Either there is a variable $x \in V(I)$ such that $I$ contains constraints implying both $(x=0)$ and $(x=1)$,
  or there is a set $X=\{x_1,\ldots,x_r\}$ of $r \leq d$ variables in $V(I)$ such that $I$
  contains constraints implying both the clause $C=(\neg x_1 \lor \ldots \lor \neg x_r)$ and $(x_i=1)$ for every $i \in [r]$.
  In the former case, we may branch on $x=1$ or $x=0$, and in the latter on setting $x_i=0$
  for some $i \in [r]$, or on every constraint implying $C$ being falsified. It is easy
  to complete this observation into a bounded-depth branching algorithm,
  solving \WeightedMinSAT{\Gamma} in time $O^*((d+1)^k)$. 
  
  In the remaining case, we show hardness. 
  As noted in Proposition~\ref{prop:posts-lattice-relations},
  $\Gamma$ can pp-define a negative 3-clause $(\neg x \lor \neg y \lor \neg z)$
  (indeed, $\Gamma$ can pp-define every relation of IS$_1^3$~\cite{BohlerRSV05coclones}).
  Furthermore, there is a relation $R \in \Gamma$ such that $H_R$ is not edgeless.
  Then by Proposition~\ref{prop:hr-gives-eq}, $\Gamma$ proportionally implements $(x=y)$. 
  Now \WeightedMinSAT{\Gamma} is either trivial or \classWone-hard by Lemma~\ref{lemma:hardness:weighted}. 
\end{proof}

For the unweighted case, it is easy to show that there is no space between the positive
case of Theorem~\ref{thm:alg-IS} and the negative case of Lemma~\ref{lemma:hardness:arrow}.

\begin{lemma} \label{lemma:isd-uwcomplete}
  Let $\Gamma$ be a finite Boolean language contained in IS$_{10}^d$ for some $d>2$
  but not in IS$_{10}^2$. Then the following apply.
  \begin{enumerate}
  \item If there is a relation $R \in \Gamma$ such that the arrow graph $H_R$ contains a $2K_2$,
    then \MinSAT{\Gamma} is either trivial or \classWone-hard.
  \item Otherwise \MinSAT{\Gamma} is FPT.
  \end{enumerate}
\end{lemma}
\begin{proof}
  First assume that there is a relation $R \in \Gamma$ such that the arrow graph $H_R$
  contains a $2K_2$. By Propositions~\ref{prop:posts-lattice-relations} and~\ref{prop:reduction-basics},
  \MinSAT{\Gamma} supports crisp negative 3-clauses. Thus by Lemma~\ref{lemma:hardness:arrow},
  \MinSAT{\Gamma} is either trivial or \classWone-hard.

  Otherwise, let $\maxarity$ be the maximum arity of a relation $R \in \Gamma$.
  Now $\Gamma \subseteq \isdpositive$ and \MinSAT{\Gamma} is FPT by Theorem~\ref{thm:alg-IS}.
\end{proof}

It remains to consider languages in the co-clone IS$_{10}^2$.
Note that since IS$_{10}^2$ is contained in both ID$_2$ and IS$_{10}^d$ for $d>2$, 
both Theorem~\ref{thm:alg-ID2} and Theorem~\ref{thm:alg-IS} potentially apply.

\begin{lemma} \label{lemma:dich-is2}
  Let $\Gamma$ be a finite Boolean language contained in IS$_{10}^2$.
  The following apply.
  \begin{enumerate}
  \item If the Gaifman graph $G_R$ is $2K_2$-free for every $R \in \Gamma$,
    or if the arrow graph $H_R$ is edgeless for every $R \in \Gamma$, then
    \WeightedMinSAT{\Gamma} is FPT.
  \item If $H_R$ is $2K_2$-free for every $R \in \Gamma$
    but the previous case does not apply,
    then \WeightedMinSAT{\Gamma} is either trivial or \classWone-hard,
    but \MinSAT{\Gamma} is FPT.
  \item Otherwise \MinSAT{\Gamma} is either trivial or \classWone-hard.
  \end{enumerate}
\end{lemma}
\begin{proof}
  Since $\text{IS}_{10}^2 \subset \text{ID}_2$, $\Gamma$ is bijunctive
  and every relation $R \in \Gamma$ can be defined as a conjunction of 2-clauses.
  If $G_R$ is $2K_2$-free for every $R \in \Gamma$, then $\Gamma \subseteq \idtwopositive$
  for some arity bound $\maxarity$, and \WeightedMinSAT{\Gamma} is FPT by Theorem~\ref{thm:alg-ID2}. 
  Similarly, if $H_R$ is edgeless for every $R \in \Gamma$, then
  by Proposition~\ref{prop:hr-gives-eq}, every relation $R \in \Gamma$ can be qfpp-defined
  using assignments and negative 2-clauses, without use of implications
  or equality clauses, and $\Gamma$ does not proportionally implement $(x=y)$. 
  Then, since $\Gamma \subseteq \text{IS}_{10}^2 \subset \text{IS}_{10}^3$, 
  Lemma~\ref{lemma:isd-weighted} implies that \WeightedMinSAT{\Gamma} is FPT.
  (Note that the requirement of Lemma~\ref{lemma:isd-weighted} that $\Gamma$
  not be contained in IS$_{10}^2$ does not apply to the algorithm,
  e.g., the language $\Gamma'=\Gamma \cup \{(\neg x \lor \neg y \lor \neg z)\}$ meets
  the condition of Lemma~\ref{lemma:isd-weighted}, and clearly
  an FPT algorithm for \MinSAT{\Gamma'} can also be applied to \MinSAT{\Gamma}.)
  Hence if one of the conditions of case 1 applies, then \WeightedMinSAT{\Gamma} is FPT.

  Assume that case 2 applies but case 1 does not. Then by assumption $H_R$
  is $2K_2$-free for every $R \in \Gamma$, and $\Gamma \subseteq \isdpositive$
  for some arity bound $\maxarity$. Thus \MinSAT{\Gamma} is FPT by Theorem~\ref{thm:alg-IS}.
  We show that \WeightedMinSAT{\Gamma} is trivial or \classWone-hard by showing
  that Lemma~\ref{lemma:hardness:weighted2} applies.
  We need to show that $\Gamma$ proportionally implements $(x=y)$
  and supports crisp constraints $(\neg x \lor \neg y)$. 
  For the first, by assumption there is a relation $R \in \Gamma$ such
  that the arrow graph $H_R$ contains an edge, and
  $\Gamma$ proportionally implements $(x=y)$  by Proposition~\ref{prop:hr-gives-eq}.
  For the second, let $R \in \Gamma$ be a relation
  such that the Gaifman graph $G_R$ contains an induced $2K_2$, say with edges $\{a,b\}$ and $\{c,d\}$.
  Note that the underlying undirected graph of the arrow graph $H_R$ is a subgraph of $G_R$.
  Since $H_R$ does not contain a $2K_2$, either $H_R$ contains no edge on $\{a,b\}$
  or no edge on $\{c,d\}$; assume the latter by symmetry. 
  Now $R$ cannot imply an assignment to $c$ or $d$, since this would
  imply edges between $\{a,b\}$ and $\{c,d\}$ in $G_R$.
  Thus the projection of $R$ to $\{c,d\}$ contains the tuples $(1,0)$ and $(0,1)$,
  and the only option under the restrictions on $\Gamma$ is
  that this projection equals $(\neg c \lor \neg d)$. 
  Thus $\Gamma$ pp-defines $(\neg x \lor \neg y)$, and 
  $\Gamma$ supports crisp negative 2-clauses by Proposition~\ref{prop:reduction-basics}.
  Then Lemma~\ref{lemma:hardness:weighted2} applies
  and \WeightedMinSAT{\Gamma} is trivial or \classWone-hard.

  Finally, if neither case 1 nor case 2 applies, then there is a
  relation $R \in \Gamma$ such that $G_R$ contains a $2K_2$,
  and a relation $R' \in \Gamma$ such that $H_{R'}$ contains a $2K_2$.
  Then \MinSAT{\Gamma} is trivial or \classWone-hard
  by Lemma~\ref{lemma:hardness:arrowandgaifman}.
\end{proof}

\subsection{Completing the proof}

We summarize all the above in the following lemma, which directly
implies Theorem~\ref{ithm:dichotomy}.

\begin{lemma} \label{lemma:dich:full-list}
  Let $\Gamma$ be a finite, Boolean constraint language.
  If one of the following cases applies, then \WeightedMinSAT{\Gamma} is FPT.
  \begin{enumerate}
  \item[(1a)] \WeightedMinSAT{\Gamma} is trivial, i.e., every non-empty
    $R \in \Gamma$ is 0-valid or every non-empty $R \in \Gamma$ is 1-valid.
  \item[(1b)] $\Gamma$ is qfpp-definable using constraints of arity at most 2,
    and for every $R \in \Gamma$ the Gaifman graph $G_R$ is $2K_2$-free.
  \item[(1c)] $\Gamma$ is qfpp-definable using negative clauses and assignments
    and cannot implement equality.
  \item[(1d)] $\Gamma$ is qfpp-definable using positive clauses and assignments,
    and cannot implement equality. 
  \end{enumerate}
  If one of the following cases applies but the above cases do not,
  then \WeightedMinSAT{\Gamma} is \classWone-hard but \MinSAT{\Gamma} is FPT.
  \begin{enumerate}
  \item[(2a)] $\Gamma$ is qfpp-definable using negative clauses, implications
    and assignments, and for every $R \in \Gamma$ the arrow graph $H_R$ is $2K_2$-free.
  \item[(2b)] $\Gamma$ is qfpp-definable using positive clauses, implications
    and assignments, and for every $R \in \Gamma$ the arrow graph $H_R$ is $2K_2$-free.   
  \end{enumerate}
  Otherwise, \MinSAT{\Gamma} is \classWone-hard. 
\end{lemma}
\begin{proof}
  We verify all positive statements (1a)--(2b) first. 
  For case (1a), if \WeightedMinSAT{\Gamma} is trivial, then it is clearly solvable in polynomial
  time (since the all-0, respectively the all-1 assignment will always be optimal).
  Case (1b) is covered by Theorem~\ref{thm:alg-ID2} (via Lemma~\ref{lemma:dich-id2}),
  case (1c) by Lemma~\ref{lemma:isd-weighted}, and case (1d) by
  Lemma~\ref{lemma:isd-weighted} applied to the dual of $\Gamma$
  (via Proposition~\ref{prop:dual}). Thus \WeightedMinSAT{\Gamma} is FPT in all cases (1a)--(1d).

  Similarly, case (2a) is covered by Theorem~\ref{thm:alg-IS} (via
  Lemma~\ref{lemma:isd-uwcomplete}), and case (2b) by the same theorem
  applied to the dual language by Proposition~\ref{prop:dual}.
  Thus \MinSAT{\Gamma} is FPT in cases (2a) and (2b).

  It remains to show hardness.     
  Assume that none of cases (1a)-(1d) applies; we will show that \WeightedMinSAT{\Gamma}
  is \classWone-hard. Since \MinSAT{\Gamma} is not trivial, one of the six cases of 
  the top-level case distinction of Lemma~\ref{lemma:top-level-cases} applies.
  If case 1 or 2 of Lemma~\ref{lemma:top-level-cases} applies, then 
  hardness follows from Lemma~\ref{lemma:hard-HS}, and if case 3 applies
  then hardness follows from Lemma~\ref{lemma:il-hard}.
  The remaining cases are that $\Gamma$ is bijunctive, IHS-B+, or IHS-B-.
      
  First assume $\Gamma$ is bijunctive. By Post's lattice~\cite{BohlerRSV05coclones},
  there are three options. Either $\Gamma$ pp-defines $(x \neq y)$, or $\Gamma$
  is qfpp-defiable over $\Gamma_2$, or $\Gamma$ is qfpp-definable over the dual of $\Gamma_2$.
  By Proposition~\ref{prop:dual}, the last two cases can be treated as one case.
  In case $\Gamma$ pp-defines $(x \neq y)$, since case (1b) fails to apply,
  \MinSAT{\Gamma} is \classWone-hard by Lemma~\ref{lemma:dich-id2}.
  Then \WeightedMinSAT{\Gamma} is \classWone-hard as well.
  The remaining subcase is that $\Gamma$ is qfpp-definable
  over $\Gamma_2$. By Lemma~\ref{lemma:dich-is2}, if none of cases (1a)--(1d) apply,
  then \WeightedMinSAT{\Gamma} is \classWone-hard, and if furthermore cases (2a)--(2b)
  do not apply then \MinSAT{\Gamma} is also \classWone-hard. Hence
  the characterization of bijunctive languages is finished.

  If $\Gamma$ is not bijunctive, by the above and by Proposition~\ref{prop:dual}
  we can assume that $\Gamma$ is qfpp-definable over $\Gamma_d$ for some $d>2$
  but not over $\Gamma_2$. By Proposition~\ref{prop:hr-gives-eq}, if $H_R$ were edgeless
  for every $R \in \Gamma$ then case (1c) or (1d) would apply, contrary to assumption.
  Hence Proposition~\ref{prop:hr-gives-eq} implies that $\Gamma$ proportionally
  implements $(x=y)$. Then \WeightedMinSAT{\Gamma} is \classWone-hard by
  Lemma~\ref{lemma:isd-weighted},
  and under the further assmuption that cases (2a)--(2b) do not apply,
  \MinSAT{\Gamma} is \classWone-hard by Lemma~\ref{lemma:isd-uwcomplete}.

  We find that for every language $\Gamma$ not covered by cases (1a)--(1d),
  \WeightedMinSAT{\Gamma} is \classWone-hard, and for every language $\Gamma$
  not covered by cases (1a)--(2b), \MinSAT{\Gamma} is \classWone-hard.
  Hence the case enumeration is complete.
\end{proof}

\section{Conclusions}
The established dichotomy naturally raises the question
where else one can establish a similar \classFPT{} vs \classWone{}-hard distinction. 
Recent results highlighted \emph{temporal CSPs} as a promising direction of future work. 
In this class of CSPs, the domain of the variables is $\mathbb{Q}$, and the values are accessed through predicates $=$, $<$, $\leq$, $\neq$.
That is, only the relative order of the values on the number line matters.
In the most classic applications temporal CSPs are used to model problems of planning 
events on a timeline while respecting various constraints about precedence.

In the world of temporal CSPs, the \MinCSP{\Gamma} problem can express many problems
in digraphs. For example, if $\Gamma = \{x < y\}$, \MinCSP{\Gamma} is equivalent to 
\textsc{Directed Feedback Vertex Set}; a problem, whose \classFPT{} status has been 
a major open problem 15--20 years ago~\cite{ChenLLOR08}.
Applications of flow-augmentation and the algorithms presented in this work
(mostly, Theorem~\ref{ithm:gaifman}) to other parameterized graph separation problems
have been explored in~\cite{GalbyMSST22,KimMPSW24}. 

Osipov and Wahlstr\"{o}m established a \classFPT{} vs \classWone{}-hard dichotomy
for the finite \emph{equality \MinCSP{\Gamma}} languages,
i.e., languages where constraints only have access to the equality predicate
(but they can be combined into more complex constraints using
 any first-order expression)~\cite{OsipovW23}.
The dichotomy of temporal \MinCSP{\Gamma} for $\Gamma$ being a subset of 
the basic relations $\{<, \leq, =, \neq\}$ has been established in~\cite{OsipovPW24}. 
A similar result for \emph{Allen's interval algebra}, where the domain is
the set of all intervals on a line, was established in~\cite{DabrowskiJOOPS23}.
In all these works, many new tractability results rely on a reduction to one of the positive
cases in the dichotomy of this work (usually, Theorem~\ref{ithm:gaifman}). 

The full \classFPT{} vs \classWone{}-hard dichotomy for temporal \MinCSP{\Gamma}, i.e., where the constraints can be expressed as any first-order expression with
 access to all four predicates $<$, $\leq$, $=$, $\neq$, remains open. 

\bibliographystyle{alpha}
\bibliography{references}




\end{document}